\documentclass[sn-mathphys,Numbered]{sn-jnl}% Math and Physical Sciences Reference Style
\usepackage{graphicx}
\usepackage{subfigure}
\usepackage{multirow}%
\usepackage{amsmath,amssymb,amsfonts}%
\usepackage{amsthm}%
\usepackage{mathrsfs}%
\usepackage[title]{appendix}%
\usepackage{xcolor}%
\usepackage{textcomp}%
\usepackage{manyfoot}%
\usepackage{booktabs}%
\usepackage{algpseudocode}%
\usepackage{listings}%
\usepackage{matlab-prettifier}
\usepackage{changes}
\RequirePackage{algorithm,algpseudocode}
\usepackage{rotating}
\usepackage{tikz}
\usetikzlibrary{arrows, decorations.pathmorphing, backgrounds, positioning, fit, petri, automata}
\usetikzlibrary{shadows,arrows,positioning}
\newtheorem{thm}{Theorem}
\newtheorem{defin}{Definition}
\newtheorem{lem}{Lemma}
\newtheorem{assum}{Assumption}
\newtheorem{rem}{Remark}

\newtheorem{con}{Condition}
\newtheorem{prop}{Proposition}
\newtheorem{Ex}{Instance}

\begin{document}

\title[WGoM]{Mixed membership estimation for categorical data with weighted responses}

\author*[1]{\fnm{Huan} \sur{Qing}}\email{qinghuan@u.nus.edu}
\affil[1]{\orgdiv{School of Economics and Finance}, \orgname{Chongqing University of Technology}, \city{Chongqing}, \postcode{400054}, \state{Chongqing}, \country{China}}


\abstract{The Grade of Membership (GoM) model, which allows subjects to belong to multiple latent classes, is a powerful tool for inferring latent classes in categorical data. However, its application is limited to categorical data with nonnegative integer responses, as it assumes that the response matrix is generated from Bernoulli or Binomial distributions, making it inappropriate for datasets with continuous or negative weighted responses. To address this, this paper proposes a novel model named the Weighted Grade of Membership (WGoM) model. Our WGoM is more general than GoM because it relaxes GoM's distribution constraint by allowing the response matrix to be generated from distributions like Bernoulli, Binomial, Normal, and Uniform as long as the expected response matrix has a block structure related to subjects' mixed memberships under the distribution. We show that WGoM can describe any response matrix with finite distinct elements. We then propose an algorithm to estimate the latent mixed memberships and other WGoM parameters. We derive the error bounds of the estimated parameters and show that the algorithm is statistically consistent. We also propose an efficient method for determining the number of latent classes $K$ for categorical data with weighted responses by maximizing fuzzy weighted modularity. The performance of our methods is validated through both synthetic and real-world datasets. The results demonstrate the accuracy and efficiency of our algorithm for estimating latent mixed memberships, as well as the high accuracy of our method for estimating $K$,  indicating their high potential for practical applications.}

\maketitle
\section{Introduction}\label{sec1}
Many systems take the form of categorical data: sets of subjects (individuals) and items joined together by categorical responses. Categorical data are widely collected across diverse fields, such as social science, psychological science, behavioral science, biological science, and transportation
\citep{landis1977measurement,liang1992multivariate,pokholok2005genome,jaeger2008categorical,agresti2012categorical,oser2013online,huang2014feature,AOS1610,ovaskainen2016using,skinner2019analysis,brown2022intermittent}. Categorical data usually has a finite number of values, where each value represents an outcome from a finite number of possible choices, such as age group, educational level, and preference level.

Categorical data can be represented mathematically using a $N$-by-$J$ response matrix $R$, where $N$ indicates the number of subjects and $J$ indicates the number of items. Each element $R(i,j)$ constitutes the observed response of the $i$-th subject to the $j$-th item. When the response is binary, the elements of $R$ are either 0 or 1. The binary responses can be agree/disagree (or yes/no) responses in a psychological test, correct/wrong responses in an educational assessment, spam/legitimate emails in an email classification project, and so on. When the response is categorical, the elements of $R$ are commonly nonnegative integers. The categorical responses can be different types of accommodation such as house/condominium/apartment in an accommodation recommendation, different strengths of agreement such as strongly disagree/disagree/ moderate/agree/strongly agree in a psychological test, and so on. In this paper, the sorting aspect of $R(i,j)$ is an important concept that captures the ordering or ranking of different response values. This sorting aspect is particularly relevant in modeling categorical data with weighted responses, where the values of $R(i,j)$ can have various interpretations depending on the context. In the example of accommodation recommendations, though the values of responses (house, condominium, apartment) are not sorted numerically, $R(i,j)$ can take values from $\{0,~1,~2,~3\}$ with $0$, $1$, $2$, and $3$ denoting no-response, house, condominium, and apartment, respectively. These numerical values $\{1,~2,~3\}$ are used to differentiate between various types of accommodation. For the psychological test example, $R(i,j)$ can take values from $\{0,~1,~2,~3,~4,~5\}$ with $0$, $1$, $2$, $3$, $4$, and $5$ denoting no-response, strongly disagree, disagree, moderate, agree, and strongly agree, respectively. In this example, the numerical values $\{1,~2,~3,~4,~5\}$ represent different strengths or levels of agreement, and the sorting aspect is implicitly embedded in the ordering of these agreement levels. Additional examples of categorical data can be found in Agresti's book \citep{agresti2012categorical}.

Based on categorical data, researchers in the aforementioned fields are typically interested in identifying latent classes (i.e., underlying groups) of subjects such that subjects within the same class share similar characteristics \citep{klonsky2008identifying,lanza2013latent,nylund2018ten,ulbricht2018use,weller2020latent,sinha2021practitioner}. For example, latent classes may represent different types of mental states like schizophrenia/depression/neurosis/normal in a mental disorder test, different ability levels like high/medium/low/basic/unqualified levels in an educational assessment, different political ideologies like liberal/neutral/conservative parties in a political survey, and different types of personalities like openness/conscientiousness/extraversion/agreeableness/neuroticism in a psychological test. Here, test items are used in assessments to measure specific skills, abilities, knowledge, or traits, and they usually have clear, objective answers or scoring criteria (like multiple-choice or true/false questions). Survey items, on the other hand, are used in surveys to gather information or opinions and may include more subjective or open-ended questions. The main difference between these two items lies in the purpose and nature of the questions in each type of data collection instrument.

The latent class model (LCM) \citep{goodman1974exploratory} is a popular statistical model for identifying latent classes for categorical data in the aforementioned fields. In recent years, Bayesian inference techniques \citep{garrett2000latent,asparouhov2011using,white2014bayeslca,li2018bayesian}, maximum likelihood estimation (MLE) methods \citep{van1996estimating,bakk2016robustness,chen2022beyond,gu2023joint}, and tensor-based algorithms \citep{anandkumar2014tensor,zeng2023tensor} are used to estimate latent classes for data generated from LCM. One main limitation of LCM is that it assumes that each subject belongs to a single class and there is no overlap between different classes. However, for most real-world data, a subject may belong to multiple classes. For example, a conscientiousness subject in a psychological test may also be openness; a neutral subject in a political survey can be thought of as having mixed memberships such that it belongs to the liberal class and the conservative class simultaneously but with different weights for each class when we assume that there are two latent classes \citep{jin2024mixed}. The Grade of Membership (GoM) model \citep{woodbury1978mathematical} was proposed as a generalization of LCM by allowing each subject to belong to more than one class. GoM is also known as ``admixture" or ``topic" or Latent Dirichlet Allocation model \citep{blei2003latent,al2019inference,dey2017visualizing,li2021grade,ke2022using,jainarayanan2023pseudotime}. GoM assumes that each subject has a membership vector (i.e., probability mass function) that sums to one. The $k$-th value of the membership vector denotes the probability (i.e., extent) that a subject partially belongs to the $k$-th class. Thus, GoM is more flexible than LCM. Bayesian inference methods \citep{erosheva2007describing,Gormley2009AGO,gu2023dimension}, MLE approach \citep{tolley1992large}, variational Expectation-Maximization algorithm \citep{beal2003variational}, and spectral methods \citep{chen2023spectral,qing2024finding} are designed to infer the latent mixed memberships for data generated from GoM. A detailed comparison between GoM and LCM can be found in \citep{erosheva2005comparing}.

Unfortunately, the Grade of Membership model mentioned above has a significant limitation: it only functions for categorical data with binary responses or nonnegative integer responses, and it does not apply to data with weighted responses, where the responses can be continuous or even negative. Importantly, these weighted responses are not sample weights typically used in statistical contexts to adjust for differences in the sampling probability of observations. Instead, they represent meaningful and inherent characteristics of the data, such as levels of agreement, trust relationships, or ratings on a discrete or continuous scale. Due to this limitation, the GoM model is unable to model datasets that contain weighted responses. Such data is frequently collected in real-world situations and ignoring weighted responses may result in losing potentially valuable information about the unobserved structures \citep{newman2004analysis}. For example, in the Wikipedia elections data \citep{leskovec2010governance}, which is a dataset of users from the English Wikipedia that voted for or against each other in admin elections, both subjects and items represent individual users, and the responses can be positive (i.e., voted for) and negative (i.e., voted against), i.e., each element of $R$ is either $-1$ or 1 (we call such $R$ signed response matrix in this paper). In the trust data of Advogato \citep{massa2009bowling}, both subjects and items represent users of the online community platform Advogato, and the responses represent trust relationships with three different levels: apprentice (0.6), journeyer (0.8), and master (1), i.e., the responses are already sorted numerically and each element of $R$ ranges in $\{0.6,~0.8,~1\}$. In Jester 100 \citep{goldberg2001eigentaste}, which is a dataset that records information about how users rated a total amount of 100 jokes, subjects are users and items represent jokes, and the rating values are continuous values between $-10$ and $10$, i.e., the responses are inherently sorted numerically and each entry of $R$ ranges in $[-10,~10]$. In Bitcoin Alpha \citep{kumar2016edge}, which is a user-user trust/distrust dataset from the Bitcoin Alpha platform, the responses represent trust or distrust sorted numerically on a scale from $-10$ to 10, i.e., all elements of $R$ take value from the set $\{-10,~ -9,~-8,~\ldots,~8,~9,~10\}$. In Dutch college \citep{van1999friendship}, which is a dataset that records ratings of friendship between 32 university freshmen, the responses range from $-1$ for risk of getting into conflict to 3 for best friend, i.e., the responses are sorted numerically and the elements of $R$ take value in $\{-1,~0,~1,~2,~3\}$. In the FilmTrust ratings data \citep{guo2016novel}, subjects (users) rate items (films), and the responses (ratings) range in $\{0.5,~1,~1.5,~2,~2.5,~3,~3.5,~4\}$. These datasets are collected by \citep{kunegis2013konect} and can be obtained from the following URL link:  \url{http://konect.cc/networks/}. GoM fails to model all the aforementioned data with weighted responses. Therefore, it is desirable to develop a more general and more flexible statistical model for weighted response data, where the responses can be continuous and negative. With this motivation, the present paper executes the following key advancements:
\begin{itemize}
\item To model categorical data with weighted responses, we propose a novel generative model, the Weighted Grade of Membership (WGoM) model. It achieves the goal of modeling weighted response categorical data by allowing the observed response matrix $R$ to be generated from an arbitrary distribution provided that the expectation of $R$ enjoys a structure reflecting the membership of each subject. For example, under WGoM, $R$ with a latent class structure can be generated from Bernoulli, Binomial, Uniform, and Normal distributions, etc. Especially, any observed response matrix $R$ with a finite number of distinct elements, where the responses can be any real value, can be modeled by our WGoM, and this is guaranteed by our Theorem \ref{ExistDisF}. Compared with GoM, our model WGoM releases GoM's distribution restriction on $R$ and GoM is a special case of our WGoM.
\item To facilitate a broad application of our model WGoM, we propose an efficient and easy-to-implement method for estimating subjects' memberships and the other WGoM parameters. We derive the theoretical convergence rates of our method and show that the method yields consistent estimation. We also provide some instances for further analysis and find that our method may behave differently when the observed response matrices are generated from different distributions under WGoM. We also introduce a method for inferring the number of latent classes $K$ in weighted response categorical data generated from the proposed model. This method determines $K$ by maximizing fuzzy weighted modularity, utilizing the estimated memberships.
\item To verify our theoretical analysis and assess the accuracy and effectiveness of the proposed methods, we conduct extensive experiments and apply our methods to real-world datasets with meaningful results.
\end{itemize}
%\subsection{More related work}
\subsection{Notation and organization}
The following notations will be used throughout this paper. For any positive integer $m$, we denote $[m]=\{1,~2,~\ldots,~m\}$ and define $I_{m\times m}$ as the $m\times m$ identity matrix. For any scale $a$, $|a|$ denotes its absolute value. For any vector $x$, $x'$ denotes its transpose and $\|x\|_{q}$ denotes its $l_{q}$-norm for any $q>0$. For any matrix $M$, we use $M'$, $M(i,:)$, $M(:,j)$, $M(i,j)$, $M(S,:)$, $M^{-1}$, $\|M\|$, $\|M\|_{F}$, $\|M\|_{2\rightarrow\infty}$, $\mathrm{rank}(M)$, $\sigma_{k}(M)$, $\lambda_{k}(M)$, $\kappa(M)$, and $\mathrm{max}(0, M)$ to denote its transpose, $i$-th row, $j$-th column, $ij$-th entry, submatrix formed by rows in the set $S$, inverse when $M$ is nonsingular, spectral norm, Frobenius norm, maximum $l_{2}$-norm among all rows, rank, $k$-th largest singular value, $k$-th largest eigenvalue ordered by magnitude when $M$ is a square matrix, conditional number, and nonnegative part, respectively. Here, $\|M\|=\sigma_{1}(M)$, $\|M\|_{F}=\mathrm{trace}(M'M)$ (i.e., the trace of the square matrix $M'M$), and  $\|M\|_{2\rightarrow\infty}=\mathrm{max}_{i}\|M(i,:)\|_{2}$. Let notations $\mathbb{R}$ and $\mathbb{N}$ denote the sets of real numbers and nonnegative integers, respectively. For any random variable $Y$, $\mathbb{E}(Y)$ denotes its expectation and $\mathbb{P}(Y=a)$ denotes the probability that $Y$ equals to $a$. Let $\mathcal{X}^{n\times m}$ denote the space of all $n\times m$ matrices, whose entries are taken from the set $\mathcal{X}$. $e_{i}$ represents a vector whose $i$-th entry is 1 and all the other entries are 0. For any values $a$ and $b$, the big $O$ notation $a = O(b)$ signifies that these two values are of the same order, meaning they increase or decrease at the same rate.

The rest of this paper is organized as follows. In Section \ref{sec2}, we formally introduce our model and demonstrate its generality and identifiability. In Section \ref{sec3}, we propose the algorithm, explain its rationality, and provide its computational complexity. In Section \ref{sec4}, we establish the theoretical result and explain its generality by providing several instances. In Section \ref{sec5}, we estimate the number of latent classes by maximizing the fuzzy weighted modularity. Sections \ref{sec6}-\ref{sec7} contain the numerical and empirical results, respectively.  Section \ref{sec8} concludes this paper with a discussion. The main notations used in this paper, detailed technical proofs, and extra instances are included in the Appendix.
\section{Weighted Grade of Membership model}\label{sec2}
In the weighted response setting presented in this paper, let $R \in \mathbb{R}^{N \times K}$ be the \emph{observed weighted response matrix}, such that $R(i,j)$ denotes the weighted response of subject $i$ to the $j$-th item for all $i\in[N]$, $j\in[J]$. Here, the response $R(i,j)$ can take any real value in our weighted response setting.

Our Weighted Grade of Membership (WGoM) model can be described by two modeling facets: the population aspect and the individual aspect. On the population aspect, WGoM assumes that there are $K$ latent classes that describe response patterns. Throughout this paper, the number of latent classes, $K$, is assumed to be known. Let $\Theta \in \mathbb{R}^{J \times K}$ be the item parameter matrix, with elements that can take any real value. To render our model identifiable, we require $\mathrm{rank}(\Theta) = K$. For $k \in [K]$, our model WGoM assumes that $\Theta(j,k)$ captures the conditional response expectations under an arbitrary distribution $\mathcal{F}$. Specifically, for $j\in [J]$, $k\in[K]$, WGoM assumes Equation (\ref{Thetajk}) below holds under an arbitrary distribution $\mathcal{F}$.
\begin{align}\label{Thetajk}
\mathbb{E}(R(i,j)|\mathrm{subject~}i~\mathrm{belongs~to~the~}k\mathrm{-th~latent~class})=\Theta(j,k).
\end{align}

On the individual aspect, since a subject may belong to multiple latent classes with different weights in this paper, we let $\Pi$ be the $N\times K$ \emph{membership matrix} such that $\Pi(i,k)$ denotes the extent that $i$-th subject partially belongs to the $k$-th latent class for $i\in[N]$, $k\in[K]$. $\Pi$ should satisfy the following condition:
\begin{align}\label{DefinePI}
&\Pi(i,k)\geq 0, \sum_{l=1}^{K}\Pi(i,l)=1 \mathrm{~for~}i\in[N],~k\in[K].
\end{align}
In Equation (\ref{DefinePI}), we know that $\|\Pi(i,:)\|_{1}=1$ should be satisfied because the summation of all probabilities for subject $i$ belonging to the $K$ latent classes is 1. For $i\in[N]$, call subject $i$ a ``pure" subject if one element of $\Pi(i,:)$ is 1 and all others $(K-1)$ elements are 0; otherwise, call it a ``mixed" subject. To guarantee that the model is identifiable (i.e., the model parameters $\Pi$ and $\Theta$ can be exactly recovered when $R$ 's expectation is given under the WGoM model, we need the following condition:
\begin{con}\label{pure}
Each of the $K$ latent classes has at least one pure subject.
\end{con}
Condition \ref{pure} imposes a mild restriction on the mixed membership matrix $\Pi$, as it merely requires that each latent class has at least one representative subject. Note that when Equation (\ref{DefinePI}) and Condition \ref{pure} hold, the rank of $\Pi$ is $K$ by basic algebra. For convenience, define $\mathcal{I}$ as the index of subject corresponding to $K$ pure subjects, one from each latent class, i.e., $\mathcal{I}=\{s_{1},~s_{2},~\ldots,~s_{K}\}$ with $s_{k}$ being a pure subject in the $k$-th latent class for $k\in[K]$. Similar to \citep{mao2021estimating}, without loss of generality, reorder the subjects so that $\Pi(\mathcal{I},:)=I_{K\times K}$.

For subject $i\in[N]$, given its membership score $\Pi(i,:)$ and the item parameter matrix $\Theta$, our WGoM assumes that for an arbitrary distribution $\mathcal{F}$, the conditional response expectation of the $i$-th subject to the $j$-th item is
\begin{align}\label{Rij}
\mathbb{E}(R(i,j)|\Pi(i,:),\Theta)=\sum_{k=1}^{K}\Pi(i,k)\Theta(j,k)\equiv\Pi(i,:)\Theta'(j,:).
\end{align}

Equation (\ref{Rij}) means that the conditional response expectation of $R(i,j)$ is a convex linear combination of the item parameter $\Theta(j,k)$ weighted by the probability $\Pi(i,k)$. We summarize our WGoM model using the following formal definition.
\begin{defin}\label{WGoM}
Let $R\in\mathbb{R}^{N\times J}$ be the observed weighted response matrix. Let $\Pi\in[0,1]^{N\times K}$ be the membership matrix satisfying Equation (\ref{DefinePI}) and Condition \ref{pure},  $\Theta\in\mathbb{R}^{J\times K}$ be the item parameter matrix with rank $K$, and $R_{0}=\Pi\Theta'$. For $i\in[N]$, $j\in[J]$, our Weighted Grade of Membership (WGoM) model assumes that for an arbitrary distribution $\mathcal{F}$, $R(i,j)$ are independently distributed according to $\mathcal{F}$ with expectation
\begin{align}\label{RFR0}
\mathbb{E}(R(i,j))=R_{0}(i,j).
\end{align}
\end{defin}
By Definition \ref{WGoM}, we see that WGoM is characterized by the membership matrix $\Pi$, the item parameter matrix $\Theta$, and the distribution $\mathcal{F}$. To emphasize this, we denote WGoM by $WGoM(\Pi,\Theta,\mathcal{F})$. Given that $R_{0}$ is the product of $\Pi$ and $\Theta'$, we refer to this form as a block structure. Under our WGoM model, $\mathcal{F}$ can be any distribution provided that Equation (\ref{RFR0}) is satisfied, i.e., WGoM only requires $R$'s expectation to be $R_{0}=\Pi\Theta'$ under distribution $\mathcal{F}$ and it does not limit $\mathcal{F}$ to be a particular distribution. Further insights into how Equation (\ref{RFR0}) applies to various distributions are provided in Instances \ref{Bernoulli}-\ref{Signed}. Recall that in Section \ref{sec1}, the elements of $R$ typically signify the sorting aspect of various responses across diverse examples. This sorting aspect is also readily comprehensible in Equation (\ref{RFR0}). Specifically, when $R(i,j)$ is generated from $WGoM(\Pi,\Theta,\mathcal{F})$ for any distribution $\mathcal{F}$ that satisfies the condition that the expectation of $R$ is $R_{0}$ under the distribution $\mathcal{F}$, a larger value of $R_{0}(i,j)$ generally indicates a larger value of $R(i,j)$.
\begin{rem}
WGoM includes two popular statistical models in latent class analysis as special cases.
\begin{itemize}
  \item If we let $\mathcal{F}$ be the Bernoulli distribution such that $R\in\{0,1\}^{N\times J}$ (i.e., the binary response setting case), Equation (\ref{Thetajk}) turns to be $\mathbb{P}(R(i,j)=1|\mathrm{subject~}i~\mathrm{belongs~to~the~}k\mathrm{-th~latent~class})=\Theta(j,k)$ and Equation (\ref{Rij}) is $\mathbb{P}(R(i,j)=1|\Pi(i,:),\Theta)=\Pi(i,:)\Theta'(j,:)$, where $\Theta(j,k)$ is a probability that ranges between 0 and 1. For this case, WGoM degenerates to the GoM model with binary responses, where this model is considered in \citep{chen2023spectral}. If we let $\mathcal{F}$ be the Binomial distribution such that $R(i,j)\sim\mathrm{Binomial}(m,\frac{R_{0}(i,j)}{m})$ for a positive integer $m$ (i.e., $R(i,j)\in\{0,~1,~2,~\ldots,~m\}$ for $i\in[N]$, $j\in[J]$), WGoM reduces to the GoM model with categorical responses, where this model is studied in \citep{qing2024finding}.
  \item If we let $\mathcal{F}$ be the Bernoulli (Binomial) distribution and there are no mixed subjects (i.e., all subjects are pure), WGoM reduces to the LCM model with binary (categorical) responses.
\end{itemize}
\end{rem}
\begin{rem}\label{RangeDifferent}
The ranges of $R$ and $\Theta$ depend on distribution $\mathcal{F}$. For example, for $i\in[N]$, $j\in[J]$, when $\mathcal{F}$ is Bernoulli distribution, $R(i,j)\in\{0,1\}$ and $\Theta(i,j)\in[0,1]$; when $\mathcal{F}$ is Binomial distribution with $m$ being the number of independent trials, $R(i,j)\in\{0,1,2,\ldots,m\}$ and $\Theta(i,j)\in[0,m]$; when $\mathcal{F}$ is Normal distribution, $R(i,j)\in\mathbb{R}$ and $\Theta(i,j)\in\mathbb{R}$; when $\mathcal{F}$ is a specified discrete distribution, $R(i,j)\in\{-1,1\}$ (i.e., the signed response setting case) and $\Theta(i,j)\in[-1,1]$. For details, see Instances \ref{Bernoulli}-\ref{Signed} and Table \ref{TableRanges}.
\end{rem}
\begin{rem}
Not all distributions satisfy Equation (\ref{RFR0}). For example, $\mathcal{F}$ cannot be a t-distribution, Cauchy distribution, or Chi-square distribution because their means (0, undefined, or a positive integer, respectively) cannot capture the latent class structure inherent in Equation (\ref{RFR0}).
\end{rem}
Theorem \ref{ExistDisF} below ensures the generality of our model WGoM. It says that for any observed weighted response matrix $R$ with $Q$ distinct elements, i.e., $R(i,j)\in\{a_{1},~a_{2},~ \ldots,~a_{Q}\}$ for $i\in[N]$, $j\in[J]$, where $a_{q}$ can be any real value for $q\in[Q]$, such $R$ can be modeled by using a carefully designed discrete distribution $\mathcal{F}$ under WGoM and the existence of such $\mathcal{F}$ is guaranteed by Theorem \ref{ExistDisF}. Given that $a_{1}<~a_{2}<~\ldots<~a_{Q}$, the $Q$ distinct values are inherently arranged in numerical order, allowing them to serve as a representation of the sorting aspect among various responses. Furthermore, since the response of subject $i$ to item $j$ can be any of the $Q$ distinct choices $\{a_{1},~a_{2},~\ldots,~a_{Q}\}$, Theorem \ref{ExistDisF} indicates that our WGoM model can effectively describe the subjects' mixed memberships in response matrices with multinomial responses, where a multinomial response is a categorical response where multiple choices are possible. Specifically, when $a_{i}=i-1$ for $i \in [Q]$, such a response matrix can be modeled by WGoM with a Binomial distribution, as illustrated in Instance \ref{Binomial} provided later. For response matrices with other sets of choices $\{a_{1},~a_{2},~\ldots,~a_{Q}\}$, Theorem \ref{ExistDisF} asserts that our WGoM can always model the data by employing a thoughtfully designed discrete distribution $\mathcal{F}$.
\begin{thm}\label{ExistDisF}
Suppose $R$'s elements take only $Q$ distinct values $\{a_{1},~a_{2},~\ldots,~a_{Q}\}$, where $a_{1}<~a_{2}<~\ldots<~a_{Q}$, $a_{q}$ can be any real value instead of simply positive integers for $q\in[Q]$, and $Q$ is a positive integer. Then,
\begin{description}
  \item[(1)] When $Q=2$, there exists only one discrete distribution $\mathcal{F}$ satisfying  Equation (\ref{RFR0}) such that $R\in\{a_{1},~a_{2}\}^{N\times J}$ can be generated from this $\mathcal{F}$ under our WGoM. The exact form of $\mathcal{F}$ is
      \begin{align*}
      \mathbb{P}(R(i,j)=a_{1})=\frac{a_{2}-R_{0}(i,j)}{a_{2}-a_{1}}, ~\mathbb{P}(R(i,j)=a_{2})=\frac{R_{0}(i,j)-a_{1}}{a_{2}-a_{1}},
      \end{align*}
      where $R_{0}(i,j)\in[a_{1},~a_{2}]$.
  \item [(2)] When $Q\geq3$, there exist at least $Q$ distinct discrete distributions $\mathcal{F}$ satisfying Equation (\ref{RFR0}) such that $R\in\{a_{1},~a_{2},~\ldots,~a_{Q}\}^{N\times J}$ can be generated from these $\mathcal{F}$s under our WGoM.
\end{description}
\end{thm}
The identifiability of our WGoM model is guaranteed by the following proposition.
\begin{prop}\label{idWGoM}
(Identifiability). The WGoM model is identifiable up to a permutation of the $K$ latent classes when $\Pi$ satisfies Equation (\ref{DefinePI}) and Condition \ref{pure} and  $\mathrm{rank}(\Theta)=K$.
\end{prop}
After specifying the model WGoM, a weighted response matrix $R$ with ground truth latent membership matrix $\Pi$ and item parameter matrix $\Theta$ can be generated through the following three procedures.
\begin{description}
  \item[Step (a).] Given the number of subjects $N$, the number of items $J$, the number of latent classes $K$, the $N\times K$ latent class membership matrix $\Pi$ that satisfies  Equation (\ref{DefinePI}) and Condition \ref{pure}, and the $J\times K$ item parameter matrix $\Theta$ with rank $K$, where the ranges of $\Theta$'s entries are dictated by the chosen distribution $\mathcal{F}$.
  \item[Step (b).] Compute the expectation response matrix $R_{0}=\Pi\Theta'$.
  \item[Step (c).] For $i\in[N]$, $j\in[J]$, let $R(i,j)$ be a random variable generated from distribution $\mathcal{F}$ with expectation $R_{0}(i,j)$.
\end{description}

Suppose that we are given the observed weighted response matrix $R$ generated from $WGoM(\Pi,\Theta,\mathcal{F})$ using Steps (a)-(c), the goal of the grade of membership analysis is to infer $\Pi$ and $\Theta$. The identifiability result, Proposition \ref{idWGoM}, ensures that $\Pi$ and $\Theta$ can be reliably estimated from the observed weighted response matrix $R$. In the next two sections, we will propose a method to estimate $\Pi$ and $\Theta$ and establish its convergence rate.
\section{An algorithm for parameters estimation}\label{sec3}
We recall that $R\in\mathbb{R}^{N\times J}$ is the observed weighted response matrix and the $N\times J$ matrix $R_{0}=\Pi\Theta'$ is the expectation of $R$ under the model WGoM. Below, in Section \ref{Secoracle}, we provide our intuition on designing an algorithm for WGoM by considering an oracle case with known $R_{0}$. We give an ideal algorithm for exactly recovering $\Pi$ and $\Theta$ from $R_{0}$. In Section \ref{Secreal}, we consider the real case with knowing the observed weighted response matrix $R$ instead of its expectation $R_{0}$ and provide our final algorithm.
\subsection{The oracle case}\label{Secoracle}
Suppose that the expectation response matrix $R_{0}$ is observed. Recall that $\mathrm{rank}(\Pi)=K$, $\mathrm{rank}(\Theta)=K$, and $R_{0}=\Pi\Theta'$, it follows that $\mathrm{rank}(R_{0})=K$. The low-dimensional structure of $R_{0}$ with only $K$ nonzero singular values, due to the small number of latent classes $K$ considered in this paper compared to $N$ and $J$, aids in the design of a procedure to estimate $\Pi$ and $\Theta$ under WGoM.

Let $R_{0}=U\Sigma V'$ be the top $K$ singular value decomposition (SVD) of $R_{0}$ such that $\Sigma$ is a $K\times K$ diagonal matrix containing the $K$ nonzero singular values of $R_{0}$, $U\in\mathbb{R}^{N\times K}$ (and $V\in\mathbb{R}^{J\times K}$) collects the corresponding left (and right) singular vectors and satisfies $U'U=I_{K\times K}$ (and $V'V=I_{K\times K}$). We have the following lemma which guarantees the existence of a simplex structure in the $N\times K$ left singular vectors matrix $U$ and serves as a foundation for building our algorithm.
\begin{lem}\label{UVWGoM}
(\emph{Ideal Simplex}) Under $WGoM(\Pi,\Theta,\mathcal{F})$, we have $U=\Pi X$, where  $X=U(\mathcal{I},:)$.
\end{lem}

\begin{figure}
\centering
{\includegraphics[width=0.68\textwidth]{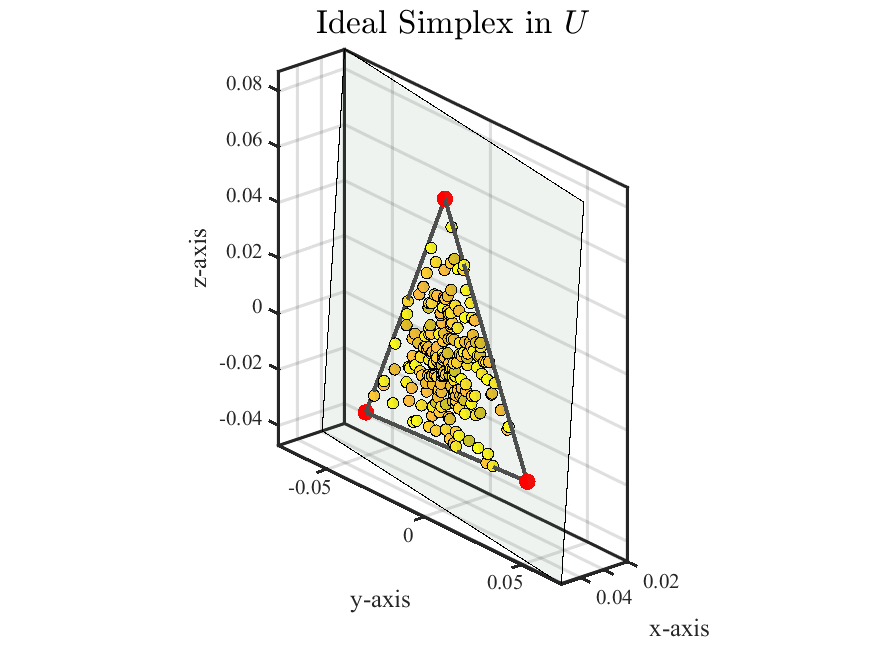}}
\caption{Illustration of the Ideal Simplex structure in $U$ when $K=3$, $N=800$, $J=400$, and each latent class has 200 pure subjects. Here, we set the mixed membership score $\Pi(\bar{i},:)$ for any mixed subject $\bar{i}$ as follows: $\Pi(\bar{i},1)$ and $\Pi(\bar{i},2)$ are independently generated random values from $\mathrm{Uniform}(0,\frac{1}{3})$, and $\Pi(\bar{i},3)=1-\Pi(\bar{i},1)-\Pi(\bar{i},2)$. In this figure, red dots denote pure rows of $U$ while other dots denote mixed rows, and the hyperplane is formed by the triangle formed by pure rows, where we call $U(i,:)$ a pure row if subject $i$ is pure and mixed row otherwise for $i\in[N]$.}
\label{ISU} %% label for entire figure
\end{figure}

It turns out that, all rows of the expectation response matrix $R_{0}$ form a simplex in $R^{K}$ which we call the Ideal Simplex, with the $K$ rows of $X$ as the vertices. Denoting the simplex by $\mathcal{S}^{\mathrm{ideal}}(v_{1},~v_{2},~\ldots,~v_{K})$, where $v_{k}=X(k,:)$ for $k\in[K]$, by Lemma \ref{UVWGoM}, we have
\begin{itemize}
  \item If subject $i$ is a pure subject such that $\Pi(i,k)=1$, then $U(i,:)$ falls on the vertex $v_{k}$; if subject $i$ is a mixed subject, then $U(i,:)$ locates in the interior of $\mathcal{S}^{\mathrm{ideal}}(v_{1},~v_{2},~\ldots,~v_{K})$.
  \item Each $U(i,:)$ is a convex linear combination of $v_{1},~v_{2},~\ldots,~v_{k}$ because $U(i,:)=\Pi(i,:)X=\sum_{k=1}^{K}\Pi(i,k)X(k,:)=\sum_{k=1}^{K}\Pi(i,k)v_{k}$.
\end{itemize}

Fig.~\ref{ISU} depicts the Ideal Simplex structure embedded within $U$ when $K=3$ in $\mathbb{R}^{3}$. This figure reveals that the pure rows of $U$ correspond to the vertices of the simplex, whereas the mixed rows of $U$ reside within the interior of the simplex. The simplex structure like $U=\Pi X=\Pi U(\mathcal{I},:)$ is also found in the area of mixed membership community detection \citep{mao2021estimating,qing2023regularized,jin2024mixed} and the area of topic modeling \citep{ke2022using,klopp2023assigning}.

Given $R_{0}$ and $K$, $U$ can be computed immediately from the top $K$ SVD of $R_{0}$. Then, once we know the $K$ vertices in the simplex $\mathcal{S}^{\mathrm{ideal}}(v_{1},~v_{2},~\ldots,~v_{K})$, the membership matrix $\Pi$ can be exactly recovered by setting $\Pi=UX^{-1}=UU^{-1}(\mathcal{I},:)$ based on Lemma \ref{UVWGoM} since the $K\times K$ corner matrix $U(\mathcal{I},:)$ is full rank. To modify the oracle procedure to the real procedure easily, we set $Z=UX^{-1}=UU^{-1}(\mathcal{I},:)$. Then, we get $\Pi(i,:)=\frac{Z(i,:)}{\|Z(i,:)\|_{1}}$ since $Z=\Pi$ and $\|\Pi(i,:)\|_{1}$ for $i\in[N]$.

Now, the question is how to find the $K$ vertices. Thanks to the simplex structure $U=\Pi X$, as mentioned in \citep{mao2021estimating}, by applying the successive projection (SP) algorithm \citep{araujo2001successive,gillis2015semidefinite} to all rows of $U$ assuming that there are $K$ vertices, we can obtain the $K$ vertices $(v_{1},~v_{2},~\ldots,~v_{K})$, i.e., we can obtain the index set $\mathcal{I}$ since $X=U(\mathcal{I},:)$. The exact form of the SP algorithm is summarized in Algorithm 1 in \citep{gillis2015semidefinite}, so we omit it here.

After recovering $\Pi$ from $U$ by SP, now we aim at recovering $\Theta$. Recall that $R_{0}=\Pi\Theta'=U\Sigma V'$, we have $\Theta=V\Sigma U'\Pi(\Pi'\Pi)^{-1}\equiv R'_{0}\Pi(\Pi'\Pi)^{-1}$ because the $K\times K$ matrix $\Pi'\Pi$ is nonsingular when $\Pi$'s rank is $K$.

The above analysis leads to Algorithm \ref{alg:IdealSCGoMA} called Ideal SCGoMA, where SCGoMA stands for spectral clustering for grade of membership analysis.

\begin{algorithm}
\caption{\textbf{Ideal SCGoMA}}
\label{alg:IdealSCGoMA}
\begin{algorithmic}[1]
\Require The $N\times J$ expectation response matrix $R_{0}$ and the number of latent classes $K$.
\Ensure The membership matrix $\Pi$ and the item parameter matrix $\Theta$.
\State Compute $U\Sigma V'$, the top $K$ SVD of $R_{0}$.
\State Apply the SP algorithm on all rows of $U$ with $K$ clusters to get the index set $\mathcal{I}$.
\State Set $Z=UU^{-1}(\mathcal{I},:)$.
\State Recover $\Pi$ by $\Pi(i,:)=\frac{Z(i,:)}{\|Z(i,:)\|_{1}}$ for $i\in[N]$.
\State Recover $\Theta$ by $\Theta=R'_{0}\Pi(\Pi'\Pi)^{-1}$.
\end{algorithmic}
\end{algorithm}
\subsection{The real case}\label{Secreal}
We extend the Ideal SCGoMA to the real case, where $R$ is observed instead of $R_{0}$. Let $\hat{R}=\hat{U}\hat{\Sigma}\hat{V}'$ be the top $K$ SVD of the observed weighted response matrix $R$, where the $K$-by-$K$ diagonal matrix $\hat{\Sigma}$ contains the top $K$ singular values of $R$, $\hat{U}\in\mathbb{R}^{N\times K}$ (and $\hat{V}\in\mathbb{R}^{J\times K}$) collects the corresponding left (and right) singular vectors and satisfies $\hat{U}'\hat{U}=I_{K\times K}$ (and $\hat{V}'\hat{V}=I_{K\times K}$). Recall that $R$ is a noisy version of $R_{0}$ by Equation (\ref{RFR0}) under the proposed model and $R_{0}$ has $K$ nonzero singular values, intuitively, the top $K$ SVD of $R$ should be close to that of $R_{0}$, i.e., $\hat{R}$, $\hat{U}$, $\hat{\Sigma}$, and $\hat{V}$ should be good approximations of $R_{0}$, $U$, $\Sigma$, and $V$, respectively. The $N$ rows of $\hat{U}$ form a noise-corrupted version of the ideal simplex $\mathcal{S}^{\mathrm{ideal}}(v_{1},~v_{2},~\ldots,~v_{K})$. Then, the estimated index set $\hat{\mathcal{I}}$ obtained by running the SP algorithm on all rows of $\hat{U}$ with $K$ clusters should be a good estimation of the index set $\mathcal{I}$. Set $\hat{Z}=\mathrm{max}(0, \hat{U}\hat{U}^{-1}(\hat{\mathcal{I}},:))$, we expect that $\hat{Z}$ should be a good estimation of $Z$, where we set $\hat{Z}$ as the nonnegative part of $\hat{U}\hat{U}^{-1}(\hat{\mathcal{I}},:)$ because $\hat{U}\hat{U}^{-1}(\hat{\mathcal{I}},:)$ may contain negative elements in practice while all elements of $Z$ must be larger than or equal to 0. Let $\hat{\Pi}$ be the estimation of $\Pi$ such that $\hat{\Pi}(i,:)=\frac{\hat{Z}(i,:)}{\|\hat{Z}(i,:)\|_{1}}$ for $i\in[N]$. Finally, we estimate $\Theta$ using $\hat{\Theta}$ computed by $\hat{\Theta}=\hat{V}\hat{\Sigma}\hat{U}'\hat{\Pi}(\hat{\Pi}'\hat{\Pi})^{-1}\equiv\hat{R}'\hat{\Pi}(\hat{\Pi}'\hat{\Pi})^{-1}$. We expect that $\hat{\Pi}$ and $\hat{\Theta}$ are good estimations of $\Pi$ and $\Theta$, respectively. We now present our main algorithm SCGoMA in Algorithm \ref{alg:SCGoMA}, which extends the Ideal SCGoMA to the real case naturally because steps 1-5 of Algorithm \ref{alg:SCGoMA} are similar to those in the oracle case. Our SCGoMA is a weighted variant of Algorithm 2 presented in \citep{chen2023spectral}. While the work in \citep{chen2023spectral} exclusively deals with categorical data with binary responses, our study extends this framework to handle categorical data with weighted responses. Fig.~\ref{IllusSCGoMA} displays the flowchart of the SCGoMA algorithm. The MATLAB codes for SCGoMA are provided in the Supplementary file.
\begin{algorithm}
\caption{\textbf{Spectral clustering for grade of membership analysis (SCGoMA for short)}}
\label{alg:SCGoMA}
\begin{algorithmic}[1]
\Require The observed weighted response matrix $R\in\mathbb{R}^{N\times J}$ and the number of latent classes $K$.
\Ensure $\hat{\Pi}$ and $\hat{\Theta}$.
\State Compute $\hat{R}=\hat{U}\hat{\Sigma} \hat{V}'$, the top $K$ SVD of $R$.
\State Apply the SP algorithm on all rows of $\hat{U}$ with $K$ clusters to get the estimated index set $\hat{\mathcal{I}}$.
\State Compute $\hat{Z}=\mathrm{max}(0, \hat{U}\hat{U}^{-1}(\hat{\mathcal{I}},:))$.
\State Compute $\hat{\Pi}(i,:)=\frac{\hat{Z}(i,:)}{\|\hat{Z}(i,:)\|_{1}}$ for $i\in[N]$.
\State Compute $\hat{\Theta}=\hat{R}'\hat{\Pi}(\hat{\Pi}'\hat{\Pi})^{-1}$.
\end{algorithmic}
\end{algorithm}

\begin{figure}[H]
\centering
\begin{tikzpicture}
[auto,
decision/.style={diamond, draw=blue, thick, fill=blue!20,
    text width=4.5em,align=flush center,
    inner sep=1pt},
block/.style ={rectangle, draw=blue, thick, fill=gray!20,
    text width=6em,align=center, rounded corners,
    minimum height=2em},
block1/.style ={diamond, draw=blue, thick, fill=orange!20,
    text width=1.4em,align=center, rounded corners,
    minimum height=2em},
block2/.style ={rectangle, draw=blue, thick, fill=blue!20,
    text width=3.2em,align=center, rounded corners,
    minimum height=2em},
block3/.style ={rectangle, draw=blue, thick, fill=gray!20,
    text width=6em,align=center, rounded corners,
    minimum height=2em},
block4/.style ={rectangle, draw=blue, thick, fill=gray!20,
    text width=6em,align=center, rounded corners,
    minimum height=2em},
line/.style ={draw, thick, -latex',shorten >=2pt},
cloud/.style ={draw=red, thick, ellipse,fill=red!20,
    minimum height=2em}]
\matrix [column sep=5mm,row sep=7mm]
{
&\node [block2] (input) {Input: $R, K$};
&\node [block3] (pca) {Compute $R$'s top $K$ SVD $\hat{R}=\hat{U}\hat{\Sigma}\hat{V}'$};
&\node [block4] (vh) {Run SP on $\hat{U}$ to get $\hat{\mathcal{I}}$};
&\node [block] (mr) {Compute $\hat{Z}, \hat{\Pi},$ and $\hat{\Theta}$};
&\node [block2] (output) {Output: $\hat{\Pi}, \hat{\Theta}$};\\
};
\begin{scope}[every path/.style=line]
\path (input) -- (pca);
\path (pca) -- (vh);
\path (vh) -- (mr);
\path (mr) -- (output);
\end{scope}
\end{tikzpicture}
\caption{Flowchart of Algorithm \ref{alg:SCGoMA}.}\label{IllusSCGoMA}
\end{figure}

Here, we present the computational complexity analysis of our SCGoMA method. The computational complexity of the top $K$ SVD in step 1 of Algorithm \ref{alg:SCGoMA} is $O(\mathrm{max}(N^{2}, J^{2})K)$. The complexity of the SP algorithm is $O(NK^{2})$ \citep{jin2024mixed}. The complexities of steps 3, 4, and 5 of Algorithm \ref{alg:SCGoMA} are $O(NK^{2})$, $O(NK)$, and $O(NJK)$, respectively. Because $K\ll\mathrm{min}(N, J)$ in this paper, as a result, the total computational complexity of SCGoMA is $O(\mathrm{max}(N^{2}, J^{2})K)$. We see that large-scale datasets (large $N$ and/or $J$) can still pose challenges in terms of computational efficiency. The complexity analysis of the SCGoMA algorithm is independent of the particular distribution from which the response matrix $R$ is derived under the WGoM model. This independence arises because the crucial steps of the SCGoMA algorithm—computing the top-$K$ SVD of $R$ and applying the SP algorithm—solely rely on the matrix $R$ itself, without presuming any specific characteristics about its underlying distribution. Consequently, although the estimation accuracy of SCGoMA may vary for different distributions (as evidenced in our experiments in Section \ref{sec6}), its computational complexity remains the same. This resilience to the choice of distribution is a notable advantage of our method, enabling SCGoMA to be widely applied to diverse categorical datasets with weighted responses, regardless of their specific distributional properties.
\section{Theoretical results}\label{sec4}
Next, we demonstrate that the per-subject error rate for membership score and the relative $l_{2}$ error for item parameters approach zero with probability approaching one under mild conditions.

For convenience, let $\Theta=\rho B$ such that $\mathrm{max}_{j\in[J],k\in[K]}|B(j,k)|=1$ (i.e., $\rho=\mathrm{max}_{j\in[J],k\in[K]}|\Theta(j,k)|$) and call $\rho$ the scaling parameter since $\rho$ controls the scaling of all item parameters. In particular, when $\mathcal{F}$ follows a Bernoulli distribution, $\rho$ is known as the sparsity parameter \citep{lei2015consistency,qing2024finding} since it captures the sparsity (i.e., number of zero elements) of an observed binary response matrix. We consider the scaling parameter $\rho$ in this paper since we find that the range of $\Theta$'s elements can be different for different distribution $\mathcal{F}$ as stated in Remark \ref{RangeDifferent} and we aim at studying the influence of the scaling parameter $\rho$ on the performance of the proposed algorithm when we generate the observed weighted response matrix from different distributions by letting $\rho$ enter the error bounds.

Let $\tau=\mathrm{max}_{i\in[N],j\in[J]}|R(i,j)-R_{0}(i,j)|$ and $\gamma=\frac{\mathrm{max}_{i\in[N],j\in[J]}\mathrm{Var}(R(i,j))}{\rho}$ where $\mathrm{Var}(R(i,j))\equiv\mathbb{E}((R(i,j)-R_{0}(i,j))^{2})$ denotes the variance of $R(i,j)$ under distribution $\mathcal{F}$. The quantity $\tau$ characterizes the maximum absolute difference between $R$ and $R_{0}$, and $\gamma$ bounds the variances for all elements of $R$. Unlike the scaling parameter $\rho$ which controls the scaling of the item parameter matrix $\Theta$ and can approach zero in this paper, $\gamma$ is a ``determined" parameter because we can always determine $\gamma$'s exact value or upper bound as long as the distribution $\mathcal{F}$ is known, while $\tau$ is usually uncertain even when the distribution $\mathcal{F}$ is known because the observed weighted response matrix $R$ is a random matrix generated by WGoM. For some specific distributions, we can obtain an exact upper bound for $\tau$, while for others we cannot. Both $\tau$ and $\gamma$ are closely related to the distribution $\mathcal{F}$ and their upper bounds differ for different distributions. For details, please refer to Instances \ref{Bernoulli}-\ref{Signed}.

To establish our theoretical guarantees, we need the following assumption.
\begin{assum}\label{a1}
Assume that $\rho\gamma\mathrm{max}(N, J)\geq \tau^{2}\mathrm{log}(N+J).$
\end{assum}
Assumption \ref{a1} is required when we aim at bounding $\|\hat{U}\hat{U}'-UU'\|_{2\rightarrow\infty}$ by using Theorem 4.4 \citep{chen2021spectral}, where $\|\hat{U}\hat{U}'-UU'\|_{2\rightarrow\infty}$ is closely related to our final theoretical bound of SCGoMA's error rates. Without this assumption, we can not obtain any insights about how our SCGoMA performs under different distributions. For details, see Remark \ref{WithourAC}.

When $\mathcal{F}$ is the Bernoulli distribution (i.e., the case where $R(i,j)\in\{0,1\}$ for $i\in[N]$, $j\in[J]$), Assumption \ref{a1} translates to $\rho\geq\frac{\mathrm{log}(N+J)}{\mathrm{max}(N, J)}$, which means a requirement on the sparsity of a binary response matrix. On the one hand, for this distribution, we observe that $\mathbb{P}(R(i,j)=1)=R_{0}(i,j)\leq\rho$ and $\mathbb{P}(R(i,j)=0)=1-R_{0}(i,j)\geq1-\rho$. These probabilities reveal that an increase in $\rho$ leads to a decrease in the probability of generating 0s in $R$. We know that if there are too many 0s in $R$ (i.e., $R$ is too sparse), it becomes harder for latent class analysis \citep{lei2015consistency}. Assumption \ref{a1} says that SCGoMA functions when $\rho$ is at least $\frac{\mathrm{log}(N+J)}{\mathrm{max}(N, J)}$, a mild condition on $R$'s sparsity. For example, with $N=500$ and $J=200$, $\rho$ can be as small as $\frac{\mathrm{log}(700)}{500}\approx0.01310216$. However, though SCGoMA remains functional at such low $\rho$ values, by numerical results in Section \ref{sec6}, we know that SCGoMA's error rates are large when $\rho$ is too small. Consequently, Assumption \ref{a1} serves primarily for our theoretical analysis and does not imply that setting $\rho$ at its minimal value satisfying the assumption will yield a negligible error rate for SCGoMA. Similar sparsity requirements also apply to community detection in social network analysis. For example, in \citep{lei2015consistency}, their main theoretical result (Theorem 3.1) requires $\rho N\geq\mathrm{log}(N)$, where $\rho$ controls the network's sparsity and $N$ denotes the number of nodes in an undirected unweighted network in \citep{lei2015consistency}. Certainly, $\rho N\geq\mathrm{log}(N)$ is also a mild requirement on $\rho$, especially for large $N$. However, this assumption is also only used for theoretical analysis, and it does not imply that when we let $\rho=\mathrm{log}(N)/N$, the Algorithm 1 studied in \citep{lei2015consistency} can estimate nodes' communities with small error rate. On the other hand, if $\rho$ is fixed (say $\rho=1$), Assumption \ref{a1} means that $\mathrm{max}(N, J)\geq\mathrm{log}(N+J)$ for the Bernoulli distribution, meaning that the number of items/subjects cannot be too small if we aim to estimate $\Pi$ and $\Theta$ with low error rates. Similar arguments hold for Binomial distribution, given that $\rho$ also controls $R$'s sparsity in this case. When $\mathcal{F}$ is some other distributions satisfying Equation (\ref{RFR0}), Assumption \ref{a1} means a requirement on $\rho\gamma\mathrm{max}(N, J)$ for our theoretical analysis (e.g., lower bound constraints on $N$ and $J$ when the scaling parameter $\rho$ is fixed). Recall that $\gamma$'s upper bound is determined, and $\tau$'s upper bound may be obtained when we know the distribution $\mathcal{F}$, the explicit expression of Assumption \ref{a1} varies with different distributions. For details, please refer to Instances \ref{Bernoulli}-\ref{Signed}.

To simplify our theoretical result, we consider the condition below.
\begin{con}\label{c1}
$K=O(1)$, $N\geq J \mathrm{~or~} N=O(J)$, $\lambda_{K}(\Pi'\Pi)=O(\frac{N}{K})$, and $\lambda_{K}(B'B)=O(\frac{J}{K})$.
\end{con}
We explain why Condition \ref{c1} is mild: $K=O(1)$ means that the number of latent classes is a constant; $N\geq J$ or $N=O(J)$ means that the number of subjects should not be too smaller than the number of items; $\lambda_{K}(\Pi'\Pi)=O(\frac{N}{K})$ says that the ``size" of each latent class is in the same order; $\lambda_{K}(B'B)=O(\frac{J}{K})$ means that the summation of item parameters for each latent class is in the same order.

Theorem \ref{mainWGoM} below concerns the consistency of our method and represents our main theoretical result. It illustrates the theoretical error bounds of SCGoMA in terms of some model parameters.
\begin{thm}\label{mainWGoM}
(Main result) Under $WGoM(\Pi,\Theta,\mathcal{F})$, let $\hat{\Pi}$ and $\hat{\Theta}$ be outputs of Algorithm \ref{alg:SCGoMA}. Suppose that Assumption \ref{a1} is satisfied and Condition \ref{c1} holds, there exists a $K\times K$ permutation matrix $\mathcal{P}$ such that with probability at least $1-O((N+J)^{-5})$, for $i\in[N]$, we have
\begin{align*}
\|e'_{i}(\hat{\Pi}-\Pi\mathcal{P})\|_{1}=O(\sqrt{\frac{\gamma\mathrm{log}(N+J)}{\rho J}})\mathrm{~and~}\frac{\|\hat{\Theta}-\Theta\mathcal{P}\|_{F}}{\|\Theta\|_{F}}=O(\sqrt{\frac{\gamma\mathrm{max}(N, J)\mathrm{log}(N+J)}{\rho NJ}}).
\end{align*}
\end{thm}

When $J=\alpha N$ for any $\alpha>0$ (i.e., the case $N\geq J$ or $N=O(J)$), Theorem \ref{mainWGoM} enables us to conclude that SCGoMA enjoys estimation consistency because its error rates converge to zero as the number of subjects $N$ increases to infinity while keeping the scaling parameter $\rho$ and the distribution $\mathcal{F}$ fixed. Meanwhile, if we fix $N, J$, and the distribution $\mathcal{F}$, when $J=\alpha N$, Theorem \ref{mainWGoM} says that $\frac{\rho J}{\gamma}$ should be much larger than $\mathrm{log}(N+J)$ to ensure that SCGoMA's error rates are sufficiently small.

For different distribution $\mathcal{F}$, the ranges of $(R,~\rho,~B)$, the upper bounds of $(\tau,~\gamma)$, and the exact forms of Assumption \ref{a1} and Theorem \ref{mainWGoM} vary. The following instances explain this statement by considering different distributions that satisfy Equation (\ref{RFR0}). For $i\in[N]$, $j\in[J]$, we consider the following distributions.
\begin{Ex}\label{Bernoulli}
Let $\mathcal{F}$ be a \texttt{Bernoulli distribution} such that our WGoM reduces to GoM, we have $R(i,j)\sim \mathrm{Bernoulli}(R_{0}(i,j))$ with $R_{0}(i,j)$ as the success probability for the Bernoulli distribution. Sure, $\mathbb{E}(R(i,j))=R_{0}(i,j)$ holds. By carefully analyzing the properties of the Bernoulli distribution, we obtain the following results.
\begin{itemize}
\item $R(i,j)\in\{0,1\}$ (i.e., $R$ is a binary response matrix), $B(i,j)\in[0,1]$, and $\rho\in(0,1]$ since $R_{0}(i,j)$ is a probability. Here, binary values can represent different levels or categories of responses, such as ``yes/no", ``agree/disagree", or ``success/failure", where 1 can be interpreted as a ``higher" or ``more positive" response compared to 0.
\item $\tau\leq1$ and $\gamma\leq1$ since $\gamma=\frac{\mathrm{max}_{i\in[N],j\in[J]}\mathrm{Var}(R(i,j))}{\rho}=\frac{\mathrm{max}_{i\in[N],j\in[J]}R_{0}(i,j)(1-R_{0}(i,j))}{\rho}\leq\frac{\mathrm{max}_{i\in[N],j\in[J]} R_{0}(i,j)}{\rho}=\frac{\mathrm{max}_{i\in[N],j\in[J]} \rho (\Pi B')(i,j)}{\rho}\leq1$.
\item Let $\tau=1$ and $\gamma=1$, Assumption \ref{a1} is $\rho\geq\frac{\mathrm{log}(N+J)}{\mathrm{max}(N,J)}$, which means a sparsity requirement of a binary response matrix; Error bounds in Theorem \ref{mainWGoM} become
    $\|e'_{i}(\hat{\Pi}-\Pi\mathcal{P})\|_{1}=O(\sqrt{\frac{\mathrm{log}(N+J)}{\rho J}})\mathrm{~and~}\frac{\|\hat{\Theta}-\Theta\mathcal{P}\|_{F}}{\|\Theta\|_{F}}=O(\sqrt{\frac{\mathrm{max}(N,J)\mathrm{log}(N+J)}{\rho NJ}})$, which imply that SCGoMA's error rates decrease when we increase $\rho$.
\end{itemize}
\end{Ex}
\begin{Ex}\label{Binomial}
Let $\mathcal{F}$ be a \texttt{Binomial distribution}, we have $R(i,j)\sim\mathrm{Binomial}(m,\frac{R_{0}(i,j)}{m})$ for a positive integer $m$. We see that $\mathbb{E}(R(i,j))=R_{0}(i,j)$ holds. For Binomial distribution, we have the following results.
\begin{itemize}
  \item $R(i,j)\in\{0,~1,~2,~\ldots,~m\}$, $B(i,j)\in[0,1]$, and $\rho\in(0,m]$ since $\frac{R_{0}(i,j)}{m}$ is a probability. Here, $R(i,j)$ takes a nonnegative value no larger than $m$ and can represent different levels or counts of a response, such as different strengths of agreement, the number of correct answers in a test, or the number of times an event occurred. Higher values of $R(i,j)$ indicate a ``greater" (``higher"/``larger"/``stronger") or ``more frequent" response.
  \item $\tau\leq m$ and $\gamma\leq1$ since $\gamma=\frac{\mathrm{max}_{i\in[N],j\in[J]}\mathrm{Var}(R(i,j))}{\rho}=\frac{m\frac{R_{0}(i,j)}{m}(1-\frac{R_{0}(i,j)}{m})}{\rho}=\frac{R_{0}(i,j)(1-\frac{R_{0}(i,j)}{m})}{\rho}\leq1$.
  \item Let $\tau=m$ and $\gamma=1$, Assumption \ref{a1} is $\rho\geq\frac{m^{2}\mathrm{log}(N+J)}{\mathrm{max}(N, J)}$; Error bounds are the same as Instance \ref{Bernoulli} and we see that increasing $\rho$ decreases SCGoMA's error rates. In this case, Assumption \ref{a1} aligns with the sparsity requirement, while Theorem \ref{mainWGoM} matches with the main result presented in \citep{qing2024finding}.
\end{itemize}
\end{Ex}
\begin{Ex}\label{Uniform}
Let $\mathcal{F}$ be a \texttt{Uniform distribution} such that $R(i,j)\sim\mathrm{Uniform}(0,2R_{0}(i,j))$. We have $\mathbb{E}(R(i,j))=\frac{0+2R_{0}(i,j)}{2}=R_{0}(i,j)$ which satisfies Equation (\ref{RFR0}). For this case, we get the below conclusions.
\begin{itemize}
\item $R(i,j)\in(0,2\rho)$, $B(i,j)\in(0,1]$, and $\rho\in(0,+\infty)$ since $R_{0}(i,j)$ can be any positive value. While $R$'s elements are continuous, they can still represent a sorting aspect in the sense that higher values correspond to ``stronger" or ``more positive" responses. For example, in a rating system where 0 represents the lowest rating and $2\rho$ represents the highest, higher values of $R(i,j)$ would indicate a more positive evaluation.
\item $\tau\leq2\rho$ and $\gamma\leq\frac{\rho}{3}$ since $\gamma=\frac{\mathrm{max}_{i\in[N],j\in[J]}\mathrm{Var}(R(i,j))}{\rho}=\mathrm{max}_{i\in[N],j\in[J]}\frac{(2R_{0}(i,j)-0)^{2}}{12\rho}=\mathrm{max}_{i\in[N],j\in[J]}\frac{R^{2}_{0}(i,j)}{3\rho}\leq\frac{\rho}{3}$.
\item Let $\tau=2\rho$ and $\gamma=\frac{\rho}{3}$, Assumption \ref{a1} is $\mathrm{max}(N, J)\geq12\mathrm{log}(N+J)$, a lower bound requirement of $\mathrm{max}(N, J)$; Error bounds in Theorem \ref{mainWGoM} become $\|e'_{i}(\hat{\Pi}-\Pi\mathcal{P})\|_{1}=O(\sqrt{\frac{\mathrm{log}(N+J)}{J}})\mathrm{~and~}\frac{\|\hat{\Theta}-\Theta\mathcal{P}\|_{F}}{\|\Theta\|_{F}}=O(\sqrt{\frac{\mathrm{max}(N, J)\mathrm{log}(N+J)}{NJ}})$, which are usually close to zero provided that $\mathrm{min}(N, J)$ is large. Meanwhile, we see that $\rho$ disappears in the error bounds. Thus $\rho$ cannot significantly influence SCGoMA's error rates.
\end{itemize}
\end{Ex}

\begin{rem}
$\mathcal{F}$ can also take the form of other Uniform distributions. For instance, we can assign $\mathcal{F}$ to be a Uniform distribution such that $R(i,j)\sim\mathrm{Uniform}(a,2R_{0}(i,j)-a)$ (or alternatively, $R(i,j)\sim\mathrm{Uniform}(2R_{0}(i,j)-a, a)$) provided that $a<\mathrm{min}_{i\in[N],j\in[J]}R_{0}(i,j)$ (or $a>\rho$). In this scenario, it follows that $R(i,j)\in(a,2R_{0}(i,j)-a)$ (or $R(i,j)\in(2R_{0}(i,j)-a,a)$). By adopting a similar analysis as in Instance \ref{Uniform}, we can derive the theoretical error bounds for SCGoMA. For brevity, the detailed analysis is omitted here.
\end{rem}

\begin{Ex}\label{Normal}
Let $\mathcal{F}$ be a \texttt{Normal distribution}, we have $R(i,j)\sim \mathrm{Normal}(R_{0}(i,j),\sigma^{2})$, where $R_{0}(i,j)$ denotes the mean ( i.e., $\mathbb{E}(R(i,j))=R_{0}(i,j)$) and $\sigma^{2}$ denotes the variance. For Normal distribution, we have the following conclusions.
\begin{itemize}
\item $R(i,j)\in\mathbb{R}$ (i.e., $R(i,j)$ can be any real value), $B(i,j)\in[-1,1]$, and $\rho\in(0,+\infty)$ since we can set the mean of Normal distribution as any value. Unlike Instances \ref{Bernoulli} and \ref{Binomial}, $B$'s entries can be negative for Normal distribution. The sorting aspect for this instance is similar to that in the Uniform instance, but it extends to the full range of real numbers, including negative values. The responses are sorted based on their magnitude, with smaller  values indicating more negative (or weaker/lower/less preferred) responses.
\item $\tau$ is unknown and $\gamma=\frac{\sigma^{2}}{\rho}$ since $\gamma=\frac{\mathrm{max}_{i\in[N],j\in[J]}\mathrm{Var}(R(i,j))}{\rho}=\frac{\sigma^{2}}{\rho}$.
\item Let $\gamma=\frac{\sigma^{2}}{\rho}$, Assumption \ref{a1} is $\mathrm{max}(N, J)\geq\frac{\tau^{2}\mathrm{log}(N+J)}{\sigma^{2}}$, a lower bound requirement for $\mathrm{max}(N, J)$; Error bounds in Theorem \ref{mainWGoM} are $\|e'_{i}(\hat{\Pi}-\Pi\mathcal{P})\|_{1}=O(\sqrt{\frac{\sigma^{2}\mathrm{log}(N+J)}{\rho^{2}J}})\mathrm{~and~}\frac{\|\hat{\Theta}-\Theta\mathcal{P}\|_{F}}{\|\Theta\|_{F}}=O(\sqrt{\frac{\sigma^{2}\mathrm{max}(N, J)\mathrm{log}(N+J)}{\rho^{2}NJ}})$. Hence, increasing $\rho$ (or decreasing $\sigma^{2}$) decreases the error ratese.
\end{itemize}
\end{Ex}
\begin{Ex}\label{Signed}
Our WGoM can also model \texttt{signed response matrix} by letting   $\mathbb{P}(R(i,j)=1)=\frac{1+R_{0}(i,j)}{2}$ and $\mathbb{P}(R(i,j)=-1)=\frac{1-R_{0}(i,j)}{2}$ according to Theorem \ref{ExistDisF}. Sure, $\mathbb{E}(R(i,j))=R_{0}(i,j)$ is satisfied. By carefully analyzing the properties of this discrete distribution, we obtain the following conclusions.
\begin{itemize}
\item $R(i,j)\in\{-1,1\}$, $B(i,j)\in[-1,1]$, and $\rho\in(0,1]$ because $\frac{1+R_{0}(i,j)}{2}$ and $\frac{1-R_{0}(i,j)}{2}$ are two probabilities. Note that like Instance \ref{Normal}, $B$'s elements can be negative for this case. Here, the sorting aspect is clear: 1 represents a ``positive" or ``favorable" response, while -1 represents a ``negative" or ``unfavorable" response. The sorting aspect here is binary but still captures the fundamental distinction between positive and negative evaluations.
\item $\tau\leq2$ and $\gamma\leq\frac{1}{\rho}$ since $\gamma=\frac{\mathrm{max}_{i\in[N],j\in[J]}\mathrm{Var}(R(i,j))}{\rho}=\frac{\mathrm{max}_{i\in[N],j\in[J]}(1-R^{2}_{0}(i,j))}{\rho}\leq\frac{1}{\rho}$.
\item Let $\tau=2$ and $\gamma=\frac{1}{\rho}$, Assumption \ref{a1} is $\mathrm{max}(N, J)\geq4\mathrm{log}(N+J)$; Error bounds in Theorem \ref{mainWGoM} are
    $\|e'_{i}(\hat{\Pi}-\Pi\mathcal{P})\|_{1}=O(\sqrt{\frac{\mathrm{log}(N+J)}{\rho^{2}J}})\mathrm{~and~}\frac{\|\hat{\Theta}-\Theta\mathcal{P}\|_{F}}{\|\Theta\|_{F}}=O(\sqrt{\frac{\mathrm{max}(N, J)\mathrm{log}(N+J)}{\rho^{2}NJ}})$. Hence, SCGoMA's error rates decrease as $\rho$ increases for signed response matrices.
\end{itemize}
\end{Ex}
\begin{rem}\label{MissingResponses}
Note that all entries of $R$ in Instances \ref{Uniform}-\ref{Signed} are nonzero, while there may exist zeros in real-world data. To generate missing responses (i.e., zeros in $R$), we can update $R(i,j)$ by $R(i,j)\mathcal{B}(i,j)$, where $\mathcal{B}(i,j)$ is a random value generated form the discrete distribution $\mathbb{P}(\mathcal{B}(i,j)=1)=p$ and $\mathbb{P}(\mathcal{B}(i,j)=0)=1-p$ for $i\in[N]$, $j\in[J]$, i.e., $\mathcal{B}\in\{0,1\}^{N\times J}$. Here, $p$ is a probability and it controls the sparsity (i.e., number of zeros) of the data. For clarity, we refer to $p$ as the sparsity parameter in this paper. It is easy to see that increasing $p$ improves the performance of SCGoMA since the number of missing responses decreases as $p$ increases. This observation is confirmed by our numerical studies presented in Section \ref{sec6}.
\end{rem}

\begin{table}[h!]
\footnotesize
	\centering
	\caption{The ranges of $R$'s elements, $B$'s elements, and $\rho$ for different $\mathcal{F}$ under our WGoM model for $i\in[N]$, $j\in[J]$.}
	\label{TableRanges}
%\resizebox{\columnwidth}{!}{
\begin{tabular}{cccccccccc}
\hline\hline
Distribution $\mathcal{F}$&$R(i,j)$&$B(i,j)$&$\rho$\\
\hline
Bernoulli&$\{0,1\}$&$[0,1]$&$(0,1]$\\
Binomial&$\{0,1,2,\ldots,m\}$&$[0,1]$&$(0,m]$\\
Uniform&$(0,2\rho)$&$[0,1]$&$(0,+\infty)$\\
Normal&Any real value&$[-1,1]$&$(0,+\infty)$\\
Signed&$\{-1,1\}$&$\{-1,1\}$&$(0,1]$\\
\hline\hline
\end{tabular}
%}
\end{table}

Table \ref{TableRanges} summaries the ranges of $R, B$, and $\rho$ for different $\mathcal{F}$. In Instances \ref{Bernoulli}-\ref{Signed}, we have shown that different kinds of weighted response matrices with latent memberships can be modeled by our WGoM model simply by setting $\mathcal{F}$ as different distributions that satisfy Equation (\ref{RFR0}). This supports the generality of our model WGoM. Again, we should emphasize that more than the seven distributions analyzed in Instances \ref{Bernoulli}-\ref{Signed}, weighted response matrix $R$ can be generated from any distribution satisfying Equation (\ref{RFR0}) under our model WGoM. Based on Theorem \ref{ExistDisF}, we also provide another two instances in \ref{ExtraInstances} to show how to analyze SCGoMA's performance when $R$ has only 2 or 3 distinct elements. To generate a weighted response matrix $R$ with elements being different from that of Theorem \ref{ExistDisF} and Instances \ref{Bernoulli}-\ref{Signed}, readers can try some other distributions satisfying Equation (\ref{RFR0}) and analyze SCGoMA's performance based on the properties of different distributions. Moreover, according to Theorem \ref{ExistDisF} and Instances \ref{Bernoulli}-\ref{Signed}, every response matrix $R$ mentioned in the examples in Section \ref{sec1} can be generated from our WGoM.
\begin{rem}\label{WithourAC}
Without Assumption \ref{a1} and Condition \ref{c1}, the theoretical upper bound for SCGoMA's error rate in estimating subjects' mixed memberships is given by $\|e'_{i}(\hat{\Pi}-\Pi\mathcal{P})\|_{1}=O(\kappa(\Pi'\Pi)K\sqrt{\lambda_{1}(\Pi'\Pi)}\|\hat{U}\hat{U}'-UU'\|_{2\rightarrow\infty})$ for $i\in[N]$ by the proof of Theorem \ref{mainWGoM}. However, this bound does not provide any insights into how model parameters such as $N$, $J$, and $\rho$ affect SCGoMA's performance under different distributions. Consequently, to investigate SCGoMA's performance under various distributions, we incorporate Assumption \ref{a1} and Condition \ref{c1} in this paper.
\end{rem}
\section{Estimating the number of latent classes}\label{sec5}
Recall that the inputs of our SCGoMA are the observed weighted response matrix $R$ and the number of latent classes $K$. This implies that we have assumed $K$ to be known in advance. However, this value is often unknown for real-world categorical data. Here, we provide an approach for estimating $K$ in categorical data with weighted responses. Our approach for inferring $K$ is inspired by the strategy used in \citep{qing2024finding}, where the authors determine $K$ for categorical data with polytomous responses (i.e., the case when $\mathcal{F}$ is a Binomial distribution) by maximizing the fuzzy modularity introduced in \citep{nepusz2008fuzzy}. Although fuzzy modularity is a good metric for measuring the quality of estimated mixed memberships for categorical data with polytomous responses, it is unsuitable for categorical data with negative responses. To handle this more general case, we employ a broader measure: the fuzzy weighted modularity proposed in \citep{qing2024mixed}, which evaluates the quality of overlapping communities in weighted networks. In this paper, we select $K$ by maximizing this modularity.

First, we introduce this modularity. Let $\hat{\Pi}\in[0,1]^{N\times k}$ be an estimated mixed membership matrix obtained by applying any algorithm $\mathcal{M}$ to $R$ with $k$ latent classes. Define a $N\times N$ symmetric matrix $A$ as $A=RR'$. Next, define two nonnegative symmetric matrices: $A_{+}=\mathrm{max}(0, A)$ and $A_{-}=\mathrm{max}(0, -A)$. We then define two $N\times1$ vectors $d_{+}$ and $d_{-}$ as $d_{+}(i)=\sum_{j\in[N]}A_{+}(i,j)$ and $d_{-}(i)=\sum_{j\in[N]}A_{-}(i,j)$ for $i\in[N]$. Further, define two values $m_{+}$ and $m_{-}$ as $m_{+}=\frac{\sum_{i\in[N]}d_{+}(i)}{2}$ and $m_{-}=\frac{\sum_{i\in[N]}d_{-}(i)}{2}$. The fuzzy weighted modularity introduced in \citep{qing2024mixed} can be calculated as follows:
\begin{align}\label{Qfwm}
Q_{\mathcal{M}}(k)=\frac{m_{+}Q_{+}-m_{-}Q_{-}}{m_{+}+m_{-}},
\end{align}
where
\[
Q_{+}=
\begin{cases}
\frac{1}{2m_{+}}\sum_{i\in[N], j\in[N]}(A_{+}(i,j)-\frac{d_{+}(i)d_{+}(j)}{2m_{+}})\hat{\Pi}(i,:)\hat{\Pi}'(j,:), & \text{if } m_{+}>0; \\
0, &\text{otherwise},
\end{cases}
\]
and
\[
Q_{-}=
\begin{cases}
\frac{1}{2m_{-}}\sum_{i\in[N], j\in[N]}(A_{-}(i,j)-\frac{d_{-}(i)d_{-}(j)}{2m_{-}})\hat{\Pi}(i,:)\hat{\Pi}'(j,:), & \text{if } m_{-}>0;\\
0, &\text{otherwise}.
\end{cases}
\]

Given that the fuzzy weighted modularity computed by Equation (\ref{Qfwm}) depends on the method $\mathcal{M}$ and the possible number of latent classes $k$, we denote it as $Q_{\mathcal{M}}(k)$ to emphasize this connection. As with the well-known Newman-Girvan modularity \citep{newman2004finding, newman2006modularity}, a larger value of the fuzzy weighted modularity implies better partitions of the estimated classes. Here, we treat $A=RR'$ as an adjacency matrix of a weighted network, allowing us to use the fuzzy weighted modularity computed based on $A$ to measure the quality of the estimated mixed memberships.

Suppose that the true number of latent classes may be one of the set $[K_C]$, where $K_C$ is set as 15 in this paper, as we assume that the true $K$ is not too large for real-world categorical datasets. For method $\mathcal{M}$, given that a larger value of $Q_{\mathcal{M}}(k)$ is preferred, we estimate $K$ by choosing the value that maximizes $Q_{\mathcal{M}}(k)$, i.e., we use $\hat{K}_{\mathcal{M}}=\arg\max_{k\in[K_C]}Q_{\mathcal{M}}(k)$ as the estimated number of latent classes for method $\mathcal{M}$. Consequently, $Q_{\mathcal{M}}(\hat{K}_{\mathcal{M}})$ represents the largest achievable fuzzy weighted modularity for method $\mathcal{M}$. This strategy of determining $K$ by maximizing a modularity metric is widely used in social network analysis \citep{nepusz2008fuzzy}.

A natural question arises: can the fuzzy weighted modularity serve as a reliable measure of the quality of estimated mixed memberships for categorical data with weighted responses, especially when the true mixed memberships are unknown? To address this, we evaluate the Accuracy rate (a metric detailed in Section \ref{sec6}) of method $\mathcal{M}$ in estimating $K$ when the observed weighted response matrix $R$ is generated from our WGoM model under various distributions. The Accuracy rate ranges from 0 to 1, with larger values indicating higher accuracy in estimating $K$. If the Accuracy rate is close to 1, we can confidently assert that the fuzzy weighted modularity is a valid measure for assessing the quality of estimated mixed memberships for categorical data with weighted responses. The numerical results presented in Section \ref{sec6} support the effectiveness of this modularity.
\section{Simulation studies}\label{sec6}
This section empirically investigates the performance of SCGoMA and compares it with four baseline methods.
\subsection{Baseline methods}
The first baseline method can also fit our WGoM model. Here, we briefly introduce this procedure. Similar to SCGoMA, we begin to propose this alternative procedure from the oracle case when the population response matrix $R_{0}$ is known. Because $R_{0}=\Pi\Theta'$ under our WGoM model, we have $R_{0}=\Pi R_{0}(\mathcal{I},:)$, where $R_{0}(\mathcal{I},:)$ is a $K\times J$ matrix with rank $K$. As suggested in \citep{mao2021estimating,jin2024mixed}, $R_{0}=\Pi R_{0}(\mathcal{I},:)$ forms a simplex with the $K$ rows of $R_{0}(\mathcal{I},:)$ being the vertices. For such a simplex structure, similar to the Ideal SCGoMA algorithm, applying the SP algorithm to all rows of $R_{0}$ assuming that there are $K$ vertices can exactly recover the $K$ vertices. Without confusion, we also let $Z=R_{0}R'_{0}(\mathcal{I},:)(R_{0}(\mathcal{I},:)R'_{0}(\mathcal{I},:))^{-1}\equiv \Pi$ by the form $R_{0}=\Pi R_{0}(\mathcal{I},:)$ where $R_{0}(\mathcal{I},:)R'_{0}(\mathcal{I},:)$ is nonsingular because the rank of $K\times J$ matrix $R_{0}(\mathcal{I},:)$ is $K$. Then, we have $\Pi(i,:)=\frac{Z(i,:)}{\|Z(i,:)\|_{1}}$ for $i\in[N]$ since $Z\equiv\Pi$. After obtaining $\Pi$, we can recover $\Theta$ from the form $R_{0}=\Pi\Theta'$ by setting $\Theta=R'_{0}\Pi(\Pi'\Pi)^{-1}$. The above analysis can be summarized by Algorithm \ref{alg:IdealRMSP} called Ideal RMSP.

\begin{algorithm}
\caption{\textbf{Ideal RMSP}}
\label{alg:IdealRMSP}
\begin{algorithmic}[1]
\Require $R_{0}$ and $K$.
\Ensure $\Pi$ and $\Theta$.
\State Apply the SP algorithm on all rows of $R_{0}$ with $K$ clusters to get the index set $\mathcal{I}$.
\State Set $Z=R_{0}R'_{0}(\mathcal{I},:)(R_{0}(\mathcal{I},:)R'_{0}(\mathcal{I},:))^{-1}$.
\State Recover $\Pi$ by $\Pi(i,:)=\frac{Z(i,:)}{\|Z(i,:)\|_{1}}$ for $i\in[N]$.
\State Recover $\Theta$ by $\Theta=R'_{0}\Pi(\Pi'\Pi)^{-1}$.
\end{algorithmic}
\end{algorithm}

Similar to the relationship between the SCGoMA algorithm and the Ideal SCGoMA algorithm, the Ideal RMSP algorithm can be easily extended to the real case. The extension is summarized in Algorithm \ref{alg:RMSP} which we call RMSP. We also provide the flowchart of RMSP in Fig.~\ref{IllusRMSP}. By comparing SCGoMA with RMSP, we see that (a) RMSP estimates the membership matrix and the item parameter matrix without using SVD, while SCGoMA is developed based on an application of SVD; (b) both methods use the SP algorithm to find the index set, and the difference is that the SP algorithm is applied to the $N\times J$ matrix $R$ in RMSP, while SP is used to the $N\times K$ matrix $\hat{U}$ in SCGoMA.
\begin{algorithm}
\caption{\textbf{Response matrix with SP (RMSP for short)}}
\label{alg:RMSP}
\begin{algorithmic}[1]
\Require $R$ and $K$.
\Ensure $\hat{\Pi}$ and $\hat{\Theta}$.
\State Apply the SP algorithm on all rows of $R$ with $K$ clusters to get the estimated index set $\hat{\mathcal{I}}$.
\State Compute $\hat{Z}=\mathrm{max}(0, RR'(\hat{\mathcal{I}},:)(R(\hat{\mathcal{I}},:)R'(\hat{\mathcal{I}},:))^{-1})$.
\State Compute $\hat{\Pi}(i,:)=\frac{\hat{Z}(i,:)}{\|\hat{Z}(i,:)\|_{1}}$ for $i\in[N]$.
\State Compute $\hat{\Theta}=R'\hat{\Pi}(\hat{\Pi}'\hat{\Pi})^{-1}$.
\end{algorithmic}
\end{algorithm}

\begin{figure}[H]
\centering
\begin{tikzpicture}
[auto,
decision/.style={diamond, draw=blue, thick, fill=blue!20,
    text width=4.5em,align=flush center,
    inner sep=1pt},
block/.style ={rectangle, draw=blue, thick, fill=gray!20,
    text width=6em,align=center, rounded corners,
    minimum height=2em},
block1/.style ={diamond, draw=blue, thick, fill=orange!20,
    text width=1.4em,align=center, rounded corners,
    minimum height=2em},
block2/.style ={rectangle, draw=blue, thick, fill=blue!20,
    text width=3.2em,align=center, rounded corners,
    minimum height=2em},
block3/.style ={rectangle, draw=blue, thick, fill=gray!20,
    text width=6em,align=center, rounded corners,
    minimum height=2em},
block4/.style ={rectangle, draw=blue, thick, fill=gray!20,
    text width=6em,align=center, rounded corners,
    minimum height=2em},
line/.style ={draw, thick, -latex',shorten >=2pt},
cloud/.style ={draw=red, thick, ellipse,fill=red!20,
    minimum height=2em}]
\matrix [column sep=5mm,row sep=7mm]
{
&\node [block2] (input) {Input: $R, K$};
&\node [block4] (vh) {Run SP on $R$ to get $\hat{\mathcal{I}}$};
&\node [block] (mr) {Compute $\hat{Z}, \hat{\Pi},$ and $\hat{\Theta}$};
&\node [block2] (output) {Output: $\hat{\Pi}, \hat{\Theta}$};\\
};
\begin{scope}[every path/.style=line]
\path (input) -- (vh);
\path (vh) -- (mr);
\path (mr) -- (output);
\end{scope}
\end{tikzpicture}
\caption{Flowchart of Algorithm \ref{alg:RMSP}.}\label{IllusRMSP}
\end{figure}

The computational costs of the 1st, 2nd, 3rd, and 4th steps in RMSP are $O(NJK)$, $O(NJK)$, $O(NK)$, and $O(NJK)$, respectively. As a result, the total complexity of RMSP is $O(NJK)$.

We also compare our SCGoMA with four other algorithms: GoM-SRSC \citep{qing2024finding}, GeoNMF \citep{GeoNMF}, SVM-cone-DCMMSB (SVM-cD for short) \citep{MaoSVM}, and MixedSCORE \citep{jin2024mixed}. Originally designed for estimating $\Pi$ and $\Theta$ in categorical data with polytomous responses, the GoM-SRSC method encounters issues when negative responses are present. To address this, we introduce $\tilde{R}$, derived by adding a sufficiently large positive constant to all elements of the original $R$, ensuring all entries in $\tilde{R}$ are positive. This adjusted matrix $\tilde{R}$ is then used to compute the regularized Laplacian matrix in \citep{qing2024finding}, enabling GoM-SRSC to function effectively for categorical data with weighted responses. Additionally, we observe that GoM-SRSC can only estimate the item parameter matrix $\Theta$ when $\mathcal{F}$ is a Bernoulli or Binomial distribution. To broaden its applicability, we replace the original $\mathrm{min}(M, \mathrm{max}(0, R'\hat{\Pi}(\hat{\Pi}'\hat{\Pi})^{-1}))$ with $R'\hat{\Pi}(\hat{\Pi}'\hat{\Pi})^{-1}$. The last three algorithms were originally designed for community detection in undirected networks, but we modify them to estimate the mixed membership matrix $\Pi$ and the item parameter matrix $\Theta$ under WGoM for latent class analysis. Here's how we modify them:

Step (1). We create a symmetric matrix $A = RR'$ of size $N \times N$. Note that for GeoNMF, we need to make $R$'s elements nonnegative by adding a large positive constant if it has any negative elements, as GeoNMF only works on symmetric nonnegative matrices.

Step (2). We apply GeoNMF (and SVM-CD and MixedSCORE) to $A$ with $K$ latent classes to obtain $\hat{\Pi}$, an estimation of the mixed membership matrix $\Pi$.

Step (3). We compute $\hat{\Theta}=R'\hat{\Pi}(\hat{\Pi}'\hat{\Pi})^{-1}$.

\begin{figure}[H]
\centering
\begin{tikzpicture}
[auto,
decision/.style={diamond, draw=blue, thick, fill=blue!20,
    text width=4.5em,align=flush center,
    inner sep=1pt},
block/.style ={rectangle, draw=blue, thick, fill=gray!20,
    text width=7.8em,align=center, rounded corners,
    minimum height=2em},
block1/.style ={diamond, draw=blue, thick, fill=orange!20,
    text width=1.4em,align=center, rounded corners,
    minimum height=2em},
block2/.style ={rectangle, draw=blue, thick, fill=blue!20,
    text width=3.2em,align=center, rounded corners,
    minimum height=2em},
block3/.style ={rectangle, draw=blue, thick, fill=gray!20,
    text width=5.7em,align=center, rounded corners,
    minimum height=2em},
block4/.style ={rectangle, draw=blue, thick, fill=gray!20,
    text width=5.8em,align=center, rounded corners,
    minimum height=2em},
line/.style ={draw, thick, -latex',shorten >=2pt},
cloud/.style ={draw=red, thick, ellipse,fill=red!20,
    minimum height=2em}]
\matrix [column sep=5mm,row sep=7mm]
{
&\node [block2] (input) {Input: $R, K$};
&\node [block3] (pca) {Set $A=RR'$};
&\node [block4] (vh) {Run $\mathcal{M}$ on $A$ with $K$ latent classes to get $\hat{\Pi}$};
&\node [block] (mr) {Compute $\hat{\Theta}=R'\hat{\Pi}(\hat{\Pi}'\hat{\Pi})^{-1}$};
&\node [block2] (output) {Output: $\hat{\Pi}, \hat{\Theta}$};\\
};
\begin{scope}[every path/.style=line]
\path (input) -- (pca);
\path (pca) -- (vh);
\path (vh) -- (mr);
\path (mr) -- (output);
\end{scope}
\end{tikzpicture}
\caption{Flowchart of modifying method $\mathcal{M}$ in the problem of mixed membership community detection to the problem of latent class analysis, where method $\mathcal{M}$ refers to GeoNMF, SVM-cD, and MixedSCORE in this paper.}\label{IllusComPetive}
\end{figure}

Fig.~\ref{IllusComPetive} displays the flowchart of Steps (1), (2), and (3). Let's briefly explain the rationales behind these steps:

\begin{itemize}
  \item For Step (1), we set $A = RR'$ as a symmetric matrix because GeoNMF, SVM-cD, and MixedSCORE only work with symmetric matrices.
  \item For Step (2), consider the ideal scenario where $A_{0} = R_{0}R'_{0}$. This leads to $A_{0} = \Pi B\Pi'$, where $B = \Theta'\Theta$ is a $K \times K$ matrix. By performing the top $K$ eigen decomposition on $A_{0} = \Xi\Lambda\Xi'$ , where $\Lambda$ is a diagonal matrix with the $k$-th diagonal entry as the $k$-th largest eigenvalue in magnitude of $A_{0}$ for $k \in [K]$ and $\Xi$ is a matrix collecting the respective eigenvectors, we obtain $\Xi = \Pi\Xi(\mathcal{I},:)$ which suggests applying the SP algorithm to $\Xi$ to obtain $\hat{\Pi}$. In fact, both SVM-cD and MixedSCORE apply some vertex hunting algorithms like SP to the variants of $\Xi$ to estimate the mixed memberships. Thus, we can run SVM-cD (and MixedSCORE) with inputs $RR'$ (i.e., $A$) and $K$ to estimate $\Pi$. For GeoNMF, it adapts symmetric nonnegative matrix factorizations to estimate the mixed memberships, thus we can run GeoNMF with inputs $A$ and $K$ to estimate $\Pi$.
  \item For Step (3), since the original GeoNMF, SVM-cD, and MixedSCORE algorithms cannot directly estimate the item parameter matrix $\Theta$, we can utilize the previously estimated $\Pi$ from Steps (1) and (2). Under the WGoM model, $\Theta$ can be calculated as $\Theta=R_{0}\Pi(\Pi'\Pi)^{-1}$. Therefore, we can estimate $\Theta$ using Step (3).
\end{itemize}

For convenience, we will continue to refer to the modified versions of GeoNMF, SVM-cD, and MixedSCORE as GeoNMF, SVM-cD, and MixedSCORE, respectively. Meanwhile, in this paper, we do not compare our method with the joint maximum likelihood (JML) algorithm developed in the R package \texttt{sirt} \citep{robitzsch2019sirt}, as JML is limited to estimating subjects' mixed memberships for categorical data with binary responses and is known to be computationally intensive \citep{chen2023spectral}.
\subsection{Evaluation metric}
We use the \emph{Hamming error} defined below to measure how close the estimated membership matrix $\hat{\Pi}$ is to the true membership matrix $\Pi$.
\begin{align*}
\mathrm{Hamming~error}=\frac{\mathrm{min}_{\mathcal{P}\in\mathcal{S}}\|\hat{\Pi}-\Pi\mathcal{P}\|_{1}}{N},
\end{align*}
where $\mathcal{S}$ is the set of all $K\times K$ permutation matrices. This metric ranges in $[0,1]$, and it is the smaller the better.

We use the \emph{Relative error} defined below to capture the difference between $\hat{\Theta}$ and $\Theta$.
\begin{align*}
\mathrm{Relative~error}=\mathrm{min}_{\mathcal{P}\in\mathcal{S}}\frac{\|\hat{\Theta}-\Theta\mathcal{P}\|_{F}}{\|\Theta\|_{F}}.
\end{align*}
This measure is nonnegative and it is also the smaller the better.

To evaluate the effectiveness of different methods in determining the number of latent classes $K$, we use the Accuracy rate metric. This metric is the proportion of instances in which a particular method accurately determines $K$ relative to the total number of independent trials conducted. This metric ranges in $[0,1]$, with a higher value indicating greater precision in the estimation of $K$.
\subsection{Numerical experiments}
Next, we aim to verify our theoretical findings and evaluate the numerical performance of SCGoMA by changing the scaling parameter $\rho$ and the number of subjects $N$ when the observed weighted response matrix $R$ is generated from different distribution $\mathcal{F}$ with expectation $R_{0}=\Pi\Theta$’ under the proposed model. All results are reported using MATLAB R2021b on a standard personal computer (Thinkpad X1 Carbon Gen 8).

For each numerical experiment, unless otherwise specified, we set the model parameters $(N,~J,~K,~\Pi,~\rho,~B)$ and the distribution $\mathcal{F}$ as follows. When the number of latent classes $K$ is fixed, we let $K=3$, $J=\frac{N}{2}$. Let each latent class have $N_{0}$ pure subjects and set the $N\times K$ membership matrix $\Pi$ as follows: $\Pi(i,1)=1$ for $1\leq i\leq N_{0}$, $\Pi(i,2)=1$ for $N_{0}+1\leq i\leq 2N_{0}$, and $\Pi(i,3)=1$ for $2N_{0}+1\leq i\leq3N_{0}$. Let subjects in $\{3N_{0}+1,3N_{0}+2,\ldots,N\}$ be mixed with membership score $(\frac{1}{3},\frac{1}{3},\frac{1}{3})$. Unless specified, we set $N_{0}=\frac{N}{4}$ (so $N$ should be set as a multiple of 4 here). The setting of the $J\times K$ matrix $B$ has two cases. For distributions that require $B$'s elements to be nonnegative, set $B(j,k)=\mathrm{rand}(1)$ for $j\in[J]$, $k\in[K]$, where $\mathrm{rand}(1)$ is a random value obtained from the Uniform distribution on $[0,1]$. For distributions that allow $B$'s entries to be negative, set $B(j,k)=2\mathrm{rand}(1)-1$ for $j\in[J],k\in[K]$. For convenience, we do not require $\mathrm{max}_{j\in[J],k\in[K]}|B(j,k)|$ to be 1 here since we will consider the scaling parameter $\rho$. For the settings of $\rho$ and $N$, we set them independently for each numerical experiment. When $K$ is changed, we set $N=100K$, $J=50K$, and $N_{0}=80$, while the membership scores for mixed subjects and $B$ remain the same as when $K$ is fixed. SCGoMA's theoretical performance under Bernoulli distribution is similar to that under Binomial distribution, so we only consider Binomial, Uniform, Normal, and signed responses here. After specifying $N,J,K,\Pi,\rho,B$, and $\mathcal{F}$, the observed weighted response matrix $R$ with expectation $R_{0}=\Pi\Theta'=\rho\Pi B'$ can be generated using Steps (a)-(c) provided after Proposition \ref{idWGoM}. We then generate missing responses by updating $R(i,j)$ by $R(i,j)\mathcal{B}(i,j)$ for $i\in[N]$, $j\in[J]$ as done in Remark \ref{MissingResponses}, where we generate the $N\times J$ binary matrix $\mathcal{B}$ by letting $\mathcal{B}(i,j)\sim\mathrm{Bernoulli}(p)$ and the sparsity parameter $p$ are set independently for each experiment.

After obtaining $R$, applying each algorithm to $R$ with $K$ latent classes yields the estimated membership matrix $\hat{\Pi}$ and the estimated item parameter matrix $\hat{\Theta}$. We then compute the Hamming error, Relative error, and estimated $K$ for each algorithm. In each numerical experiment, we generate 100 independent replicates and report the averaged Hamming error, the averaged Relative error, the averaged running time (in seconds), and the Accuracy rate for each algorithm. We consider four cases: changing $\rho$, changing $N$, changing $K$, and changing $p$, where the range of $\rho$ is set in the theoretical range analyzed in Instances \ref{Bernoulli}-\ref{Signed} for each distribution.
\subsubsection{Binomial distribution}
When $R(i,j)\sim \mathrm{Binomial}(m,\frac{R_{0}(i,j)}{m})$ for $i\in[N]$, $j\in[J]$:

\textbf{Experiment 1(a): changing $\rho$.} Let $N=800$, $m=5$, $p=1$, and $\rho$ range in $\{0.5,~1,~1.5,~\ldots,~5\}$.

\textbf{Experiment 1(b): changing $N$.} Let $\rho=0.5$, $m=5$, $p=1$, and $N$ range in $\{800,~1600,~2400,~\ldots,~6400\}$.

\textbf{Experiment 1(c): changing $K$.} Let $m=5$, $\rho=1$, $p=1$, and $K$ range in $\{1,~2,~3,~4,~5\}$.

The results are presented in Fig.~\ref{S1}. For the estimation of $\Pi$ and $\Theta$: SCGoMA exhibits enhanced performance as the sparsity parameter $\rho$ increases in Experiment 1(a), aligning with our analysis in Instance \ref{Binomial}. In Experiment 1(b), SCGoMA demonstrates superior performance when the number of items $N$ grows, confirming our analysis in Theorem \ref{mainWGoM}. In Experiment 1(c), SCGoMA's performance diminishes as $K$ increases, as estimating $\Pi$ and $\Theta$ becomes more challenging for larger $K$. Except for RMSP and SVM-cD, all methods exhibit improved performance with increasing $\rho$ (and $N$). Notably, SCGoMA, GeoNMF, MixedSCORE, and GoM-SRSC exhibit competitive performances and outperform RMSP and SVM-cD in estimating the WGoM parameters. Additionally, while RMSP is the fastest method, it demonstrates the poorest performance in estimating $\Pi$ and $\Theta$. Although our SCGoMA runs slightly slower than RMSP, it runs faster than the other four methods, and this is also observed in Experiments 2-4. All approaches processed weighted response matrices comprising 6400 subjects and 3200 items within ten seconds. For inferring the number of latent classes $K$: Except for RMSP and GeoNMF, all methods accurately estimate $K$ when the true $K \geq 2$. For $K = 1$, all methods fail to determine the true $K$. Nevertheless, this is not a significant issue because, for real data, we typically prefer to consider cases where $K \geq 2$ for further analysis. Otherwise, if we assume $K = 1$, implying all subjects belong to the same latent class, further analysis becomes unnecessary. The high accuracy of SCGoMA, SVM-cD, Mixed-SCORE, and GoM-SRSC in estimating $K$ by maximizing the fuzzy weighted modularity suggests that this modularity is a reliable measure for assessing the quality of estimated mixed memberships for categorical data with weighted responses. Lastly, we did not consider increasing $p$ in this experiment because $\rho$ already controls the sparsity of $R$ for the Binomial distribution.
\begin{figure}
\centering
\resizebox{\columnwidth}{!}{
\subfigure[Experiment 1(a)]{\includegraphics[width=0.33\textwidth]{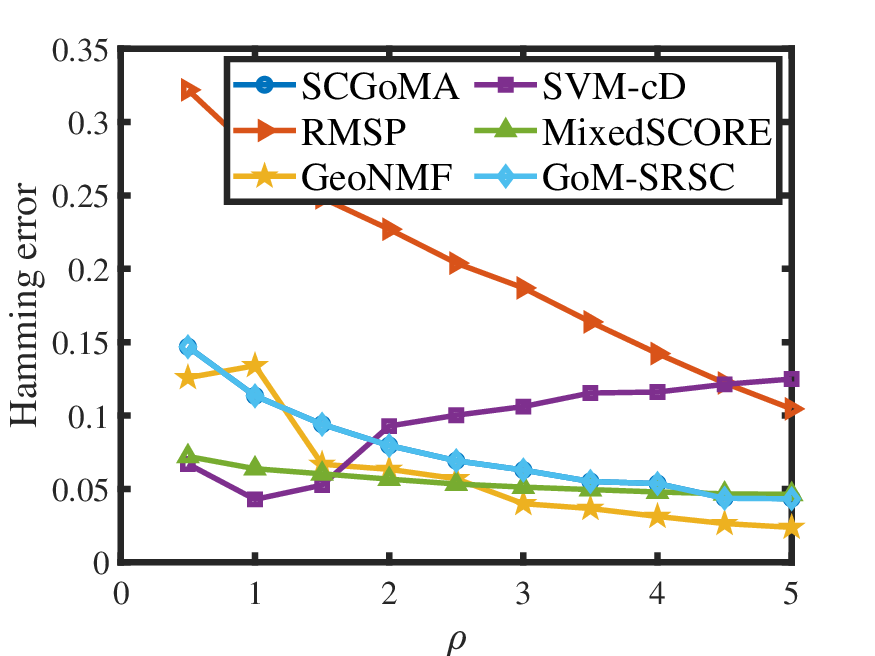}}
\subfigure[Experiment 1(a)]{\includegraphics[width=0.33\textwidth]{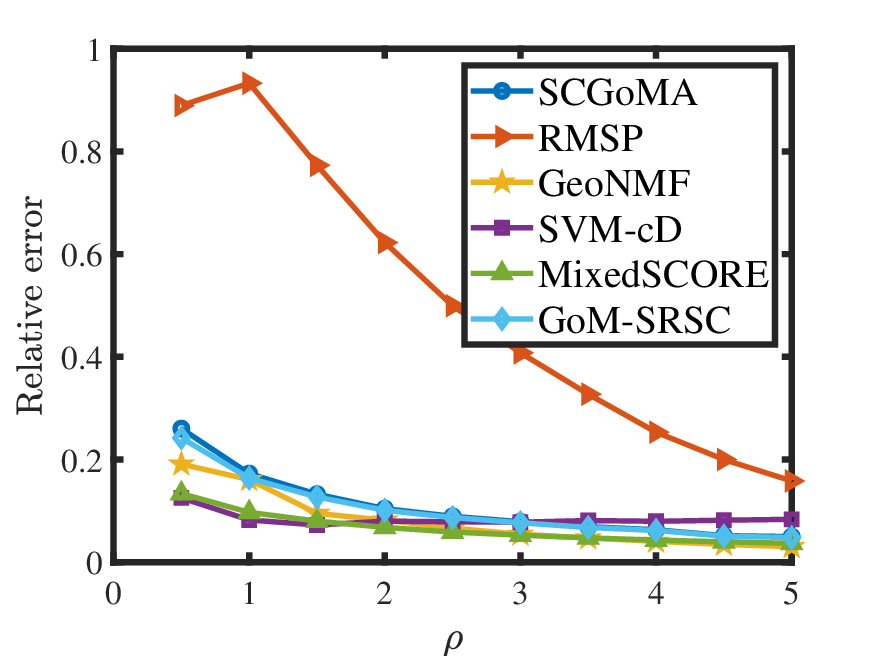}}
\subfigure[Experiment 1(a)]{\includegraphics[width=0.33\textwidth]{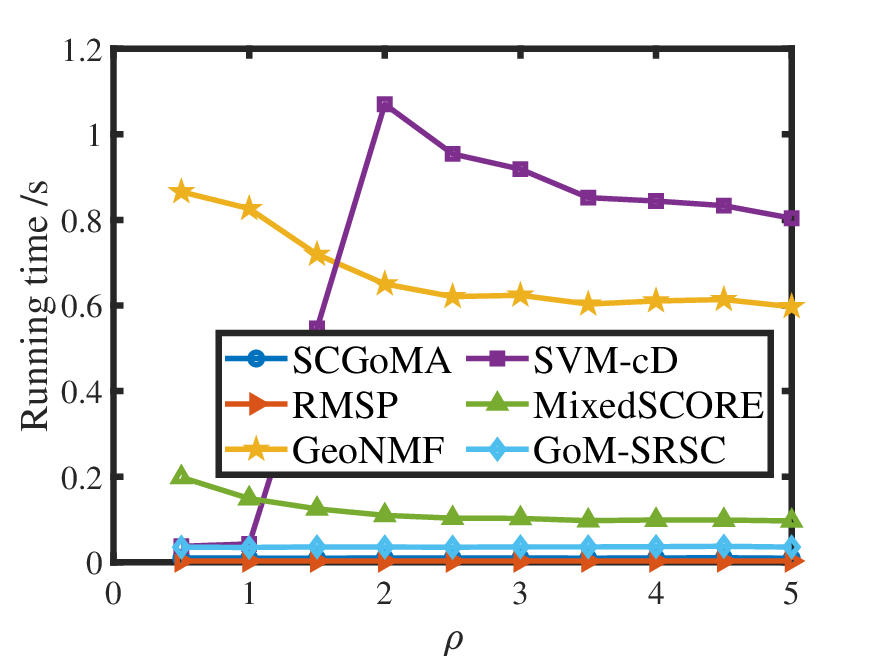}}
\subfigure[Experiment 1(a)]{\includegraphics[width=0.33\textwidth]{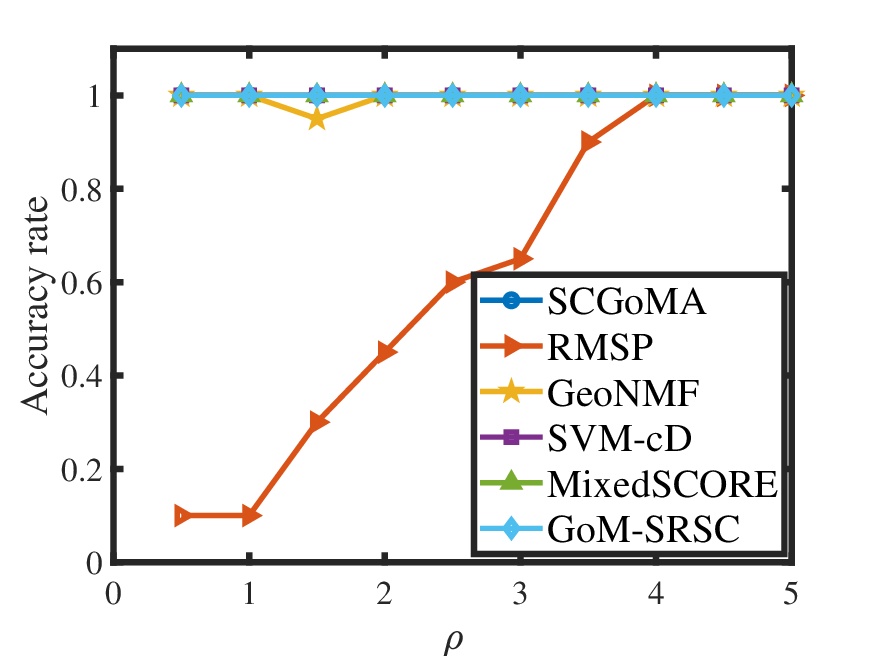}}
}
\resizebox{\columnwidth}{!}{
\subfigure[Experiment 1(b)]{\includegraphics[width=0.33\textwidth]{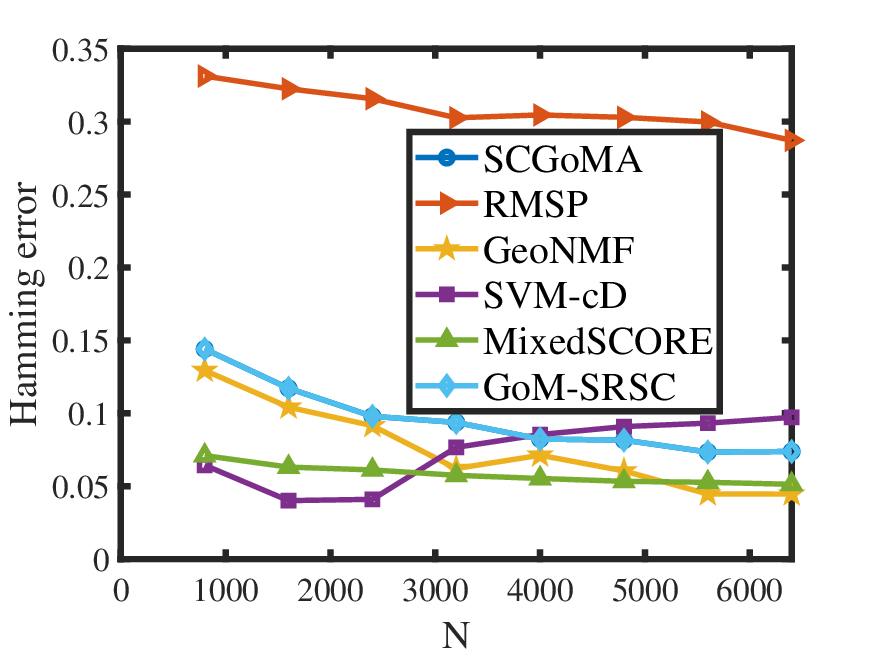}}
\subfigure[Experiment 1(b)]{\includegraphics[width=0.33\textwidth]{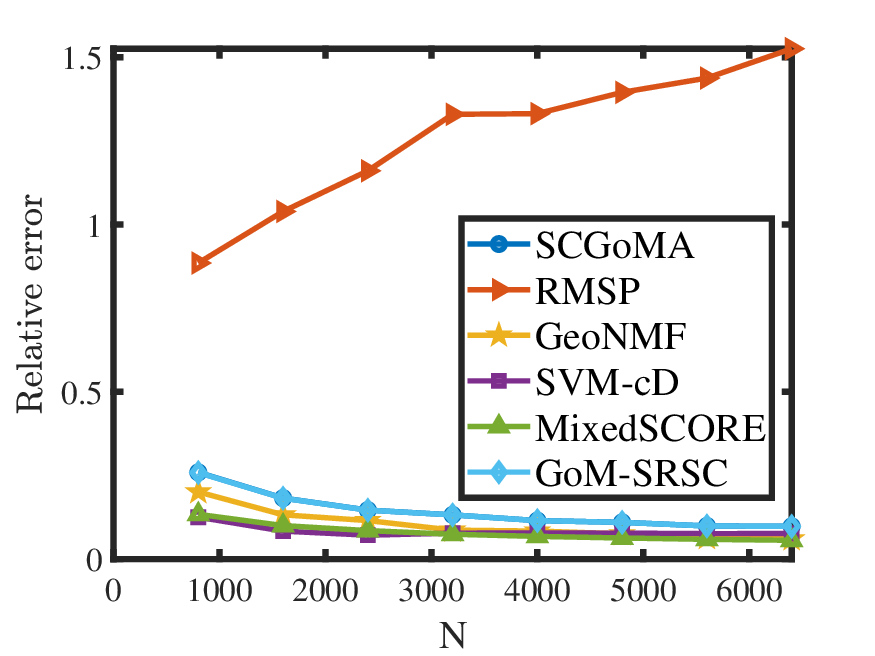}}
\subfigure[Experiment 1(b)]{\includegraphics[width=0.33\textwidth]{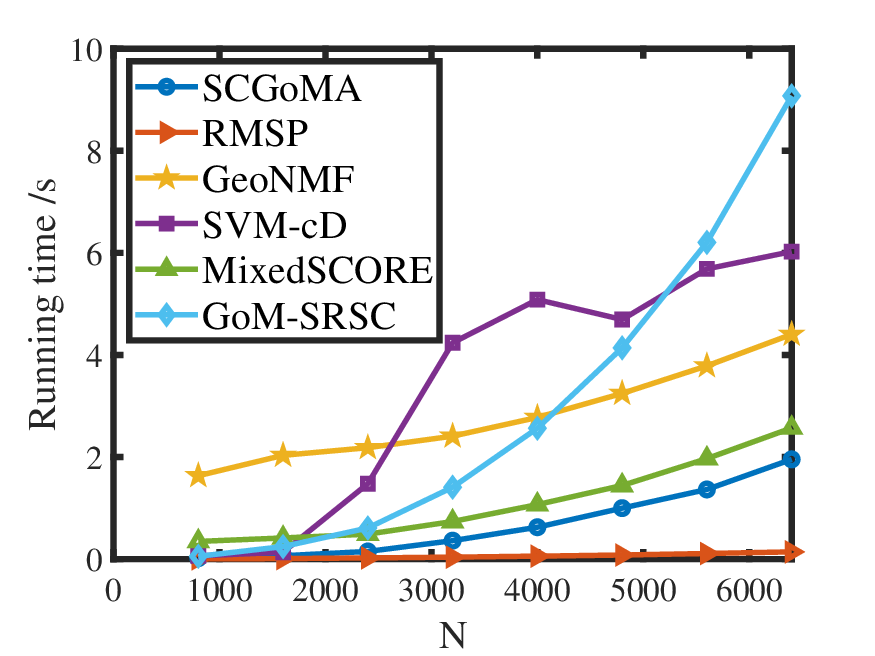}}
\subfigure[Experiment 1(b)]{\includegraphics[width=0.33\textwidth]{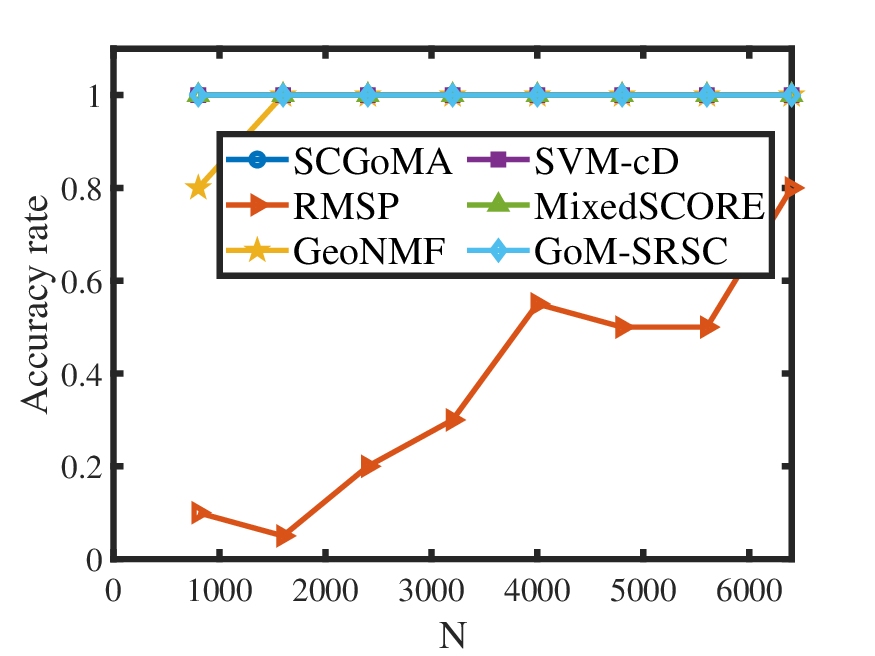}}
}
\resizebox{\columnwidth}{!}{
\subfigure[Experiment 1(c)]{\includegraphics[width=0.33\textwidth]{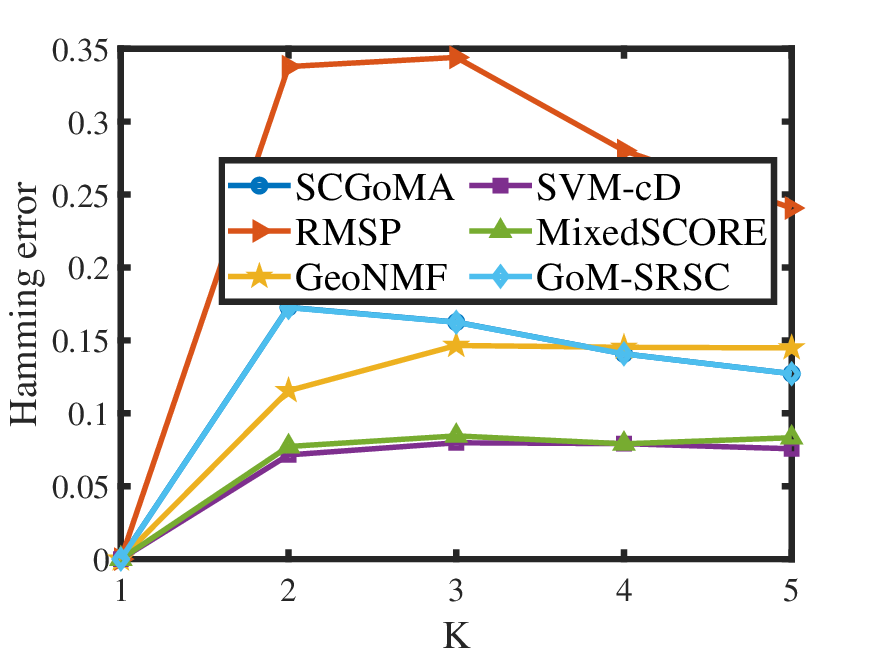}}
\subfigure[Experiment 1(c)]{\includegraphics[width=0.33\textwidth]{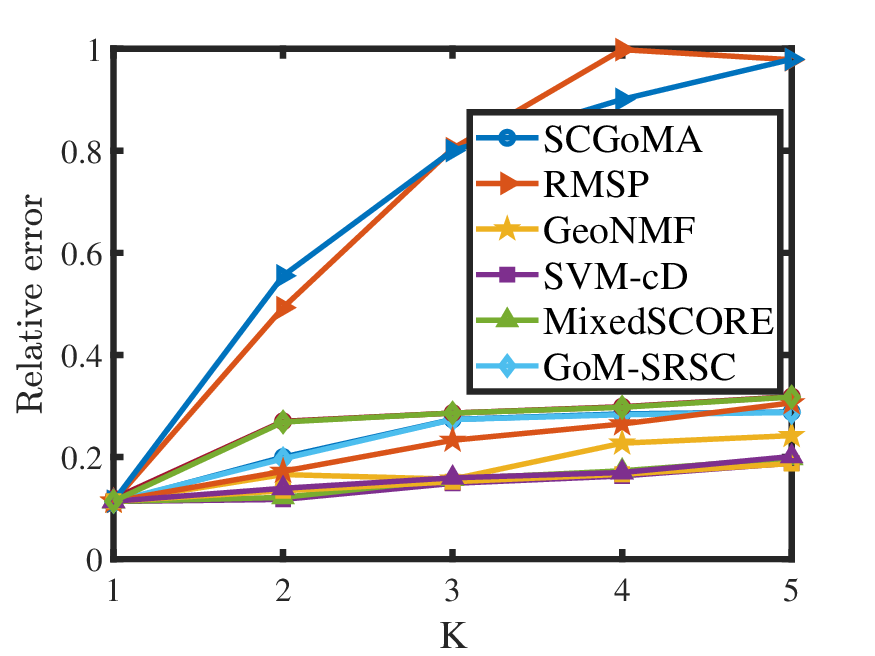}}
\subfigure[Experiment 1(c)]{\includegraphics[width=0.33\textwidth]{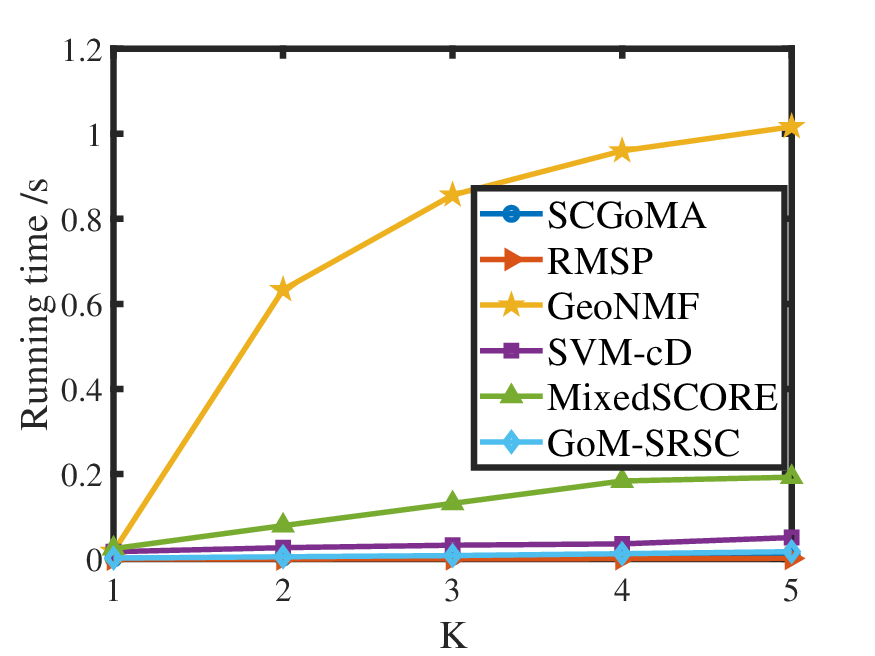}}
\subfigure[Experiment 1(c)]{\includegraphics[width=0.33\textwidth]{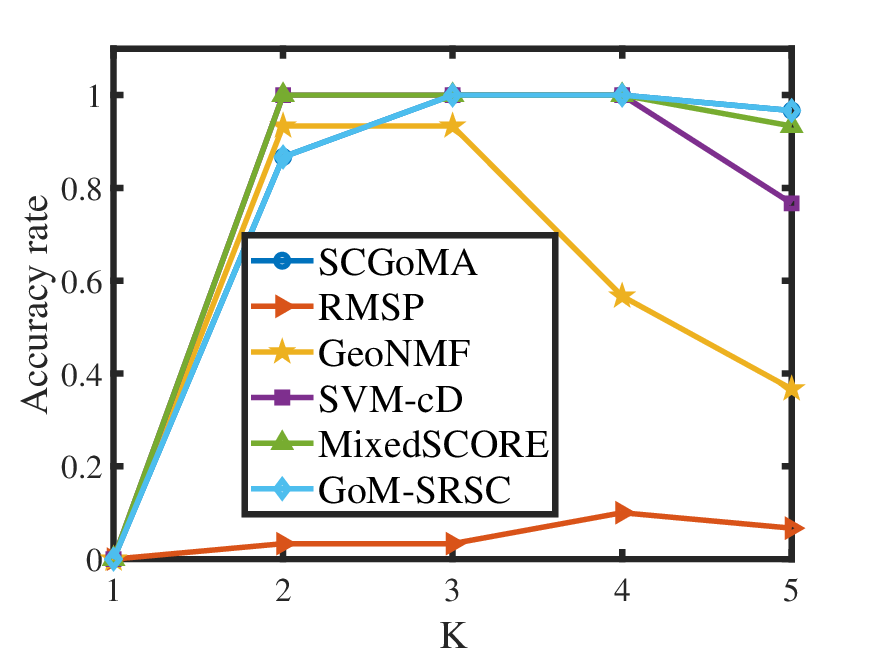}}
}
\caption{Binomial distribution.}
\label{S1} %% label for entire figure
\end{figure}
\subsubsection{Uniform distribution}
When $R(i,j)\sim \mathrm{Uniform}(0,2R_{0}(i,j))$ for $i\in[N]$, $j\in[J]$:

\textbf{Experiment 2(a): changing $\rho$.} Let $N=800$, $p=1$ and $\rho$ range in $\{10,~20,~\ldots,~100\}$.

\textbf{Experiment 2(b): changing $N$.} Let $\rho=1$, $p=1$ and $N$ range in $\{800,~1600,~2400,~\ldots,~6400\}$.

\textbf{Experiment 2(c): changing $K$.} Let $\rho=1$, $p=1$, and $K$ range in $[5]$.

\textbf{Experiment 2(d): changing $p$.} Let $N=800$, $\rho=1$, and $p$ range in $\{0.1,~0.2,~\ldots,~1\}$.

The results are displayed in Fig.~\ref{S2}. For the estimation of $\Pi$ and $\Theta$: All methods exhibit insensitivity to the scaling parameter $\rho$, consistent with our findings in Instance \ref{Uniform}. Except for SVM-cD and RMSP, all methods demonstrate enhanced performance as the number of subjects $N$ increases, corroborating our analysis based on Theorem \ref{mainWGoM}. As the sparsity parameter $p$ increases, all methods, except for SVM-cD, exhibit improved performance. Notably, SCGoMA, GeoNMF, MixedSCORE, and GoM-SRSC exhibit comparable performances and typically outperform RMSP and SVM-cD in error rates. As $K$ increases, all methods perform poorer. For the estimation of $K$, all methods, except for RMSP, determine $K$ with high accuracy when the true $K\geq2$, highlighting the utility of the fuzzy weighted modularity as a measure for evaluating the quality of mixed membership analysis for categorical data with weighted responses.

\begin{figure}
\centering
\resizebox{\columnwidth}{!}{
\subfigure[Experiment 2(a)]{\includegraphics[width=0.33\textwidth]{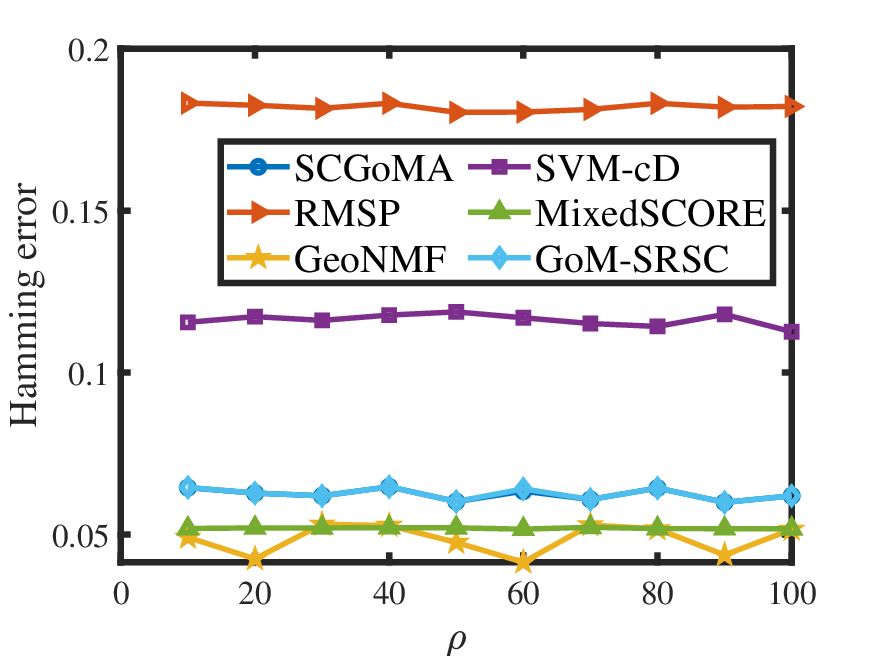}}
\subfigure[Experiment 2(a)]{\includegraphics[width=0.33\textwidth]{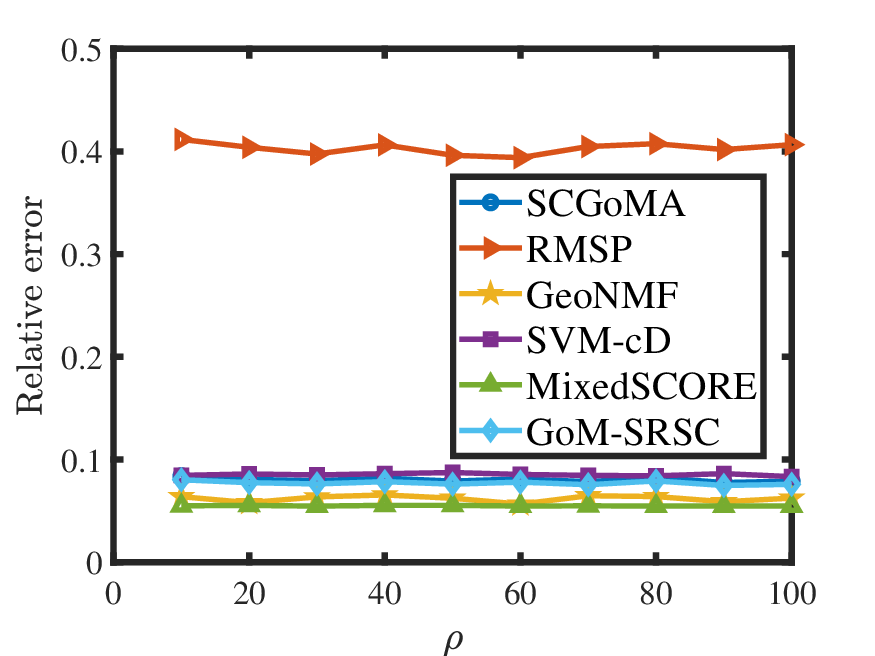}}
\subfigure[Experiment 2(a)]{\includegraphics[width=0.33\textwidth]{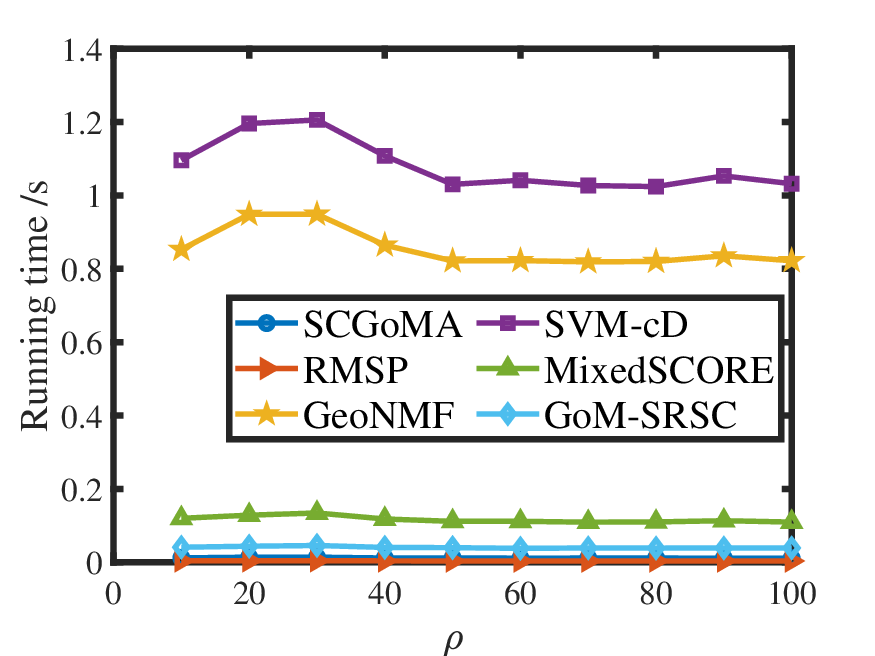}}
\subfigure[Experiment 2(a)]{\includegraphics[width=0.33\textwidth]{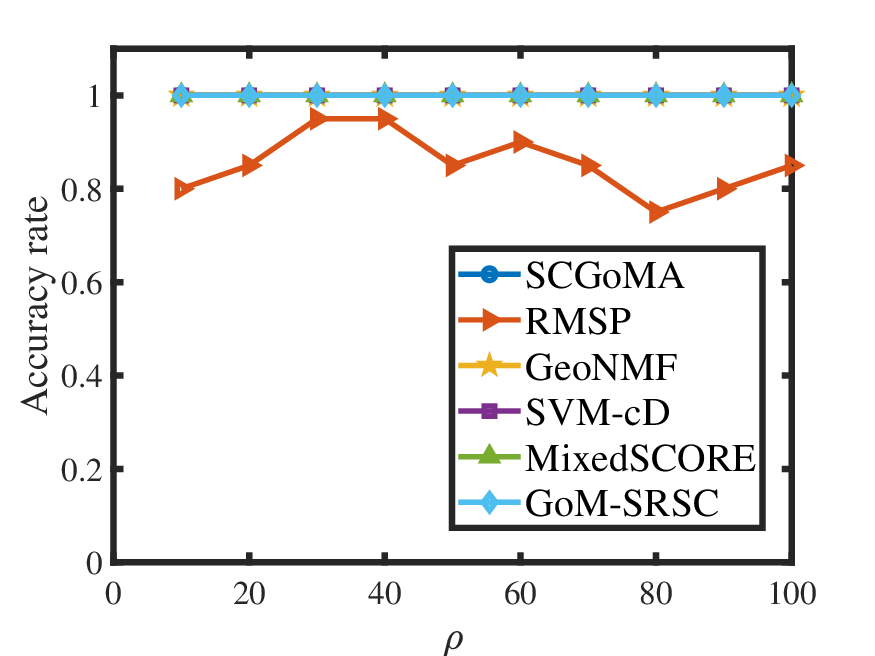}}
}
\resizebox{\columnwidth}{!}{
\subfigure[Experiment 2(b)]{\includegraphics[width=0.33\textwidth]{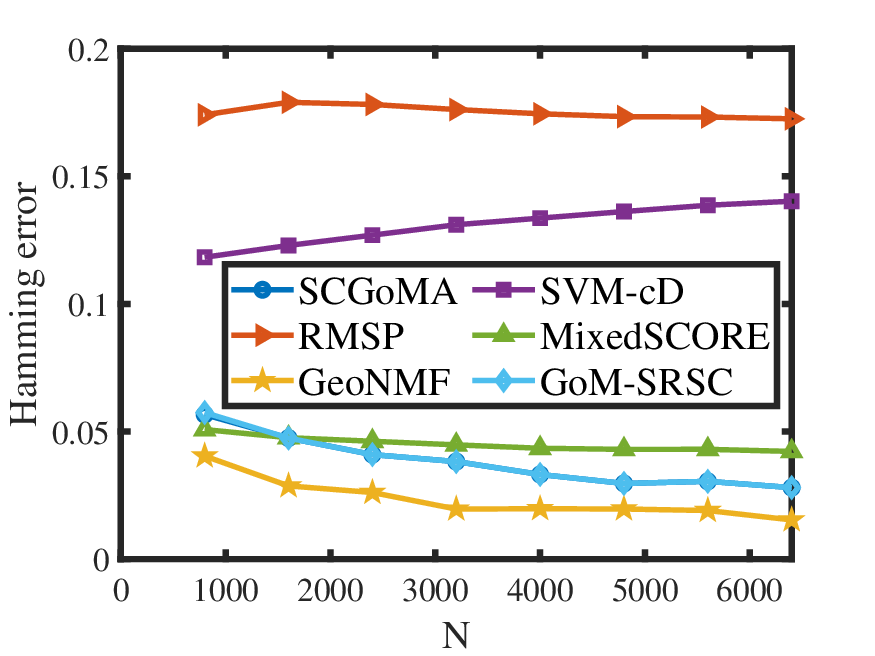}}
\subfigure[Experiment 2(b)]{\includegraphics[width=0.33\textwidth]{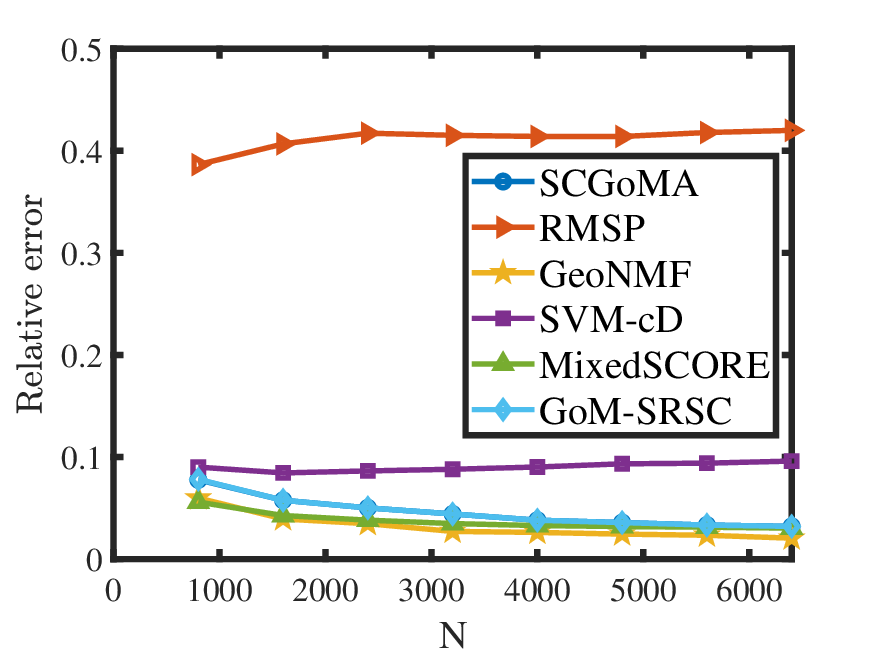}}
\subfigure[Experiment 2(b)]{\includegraphics[width=0.33\textwidth]{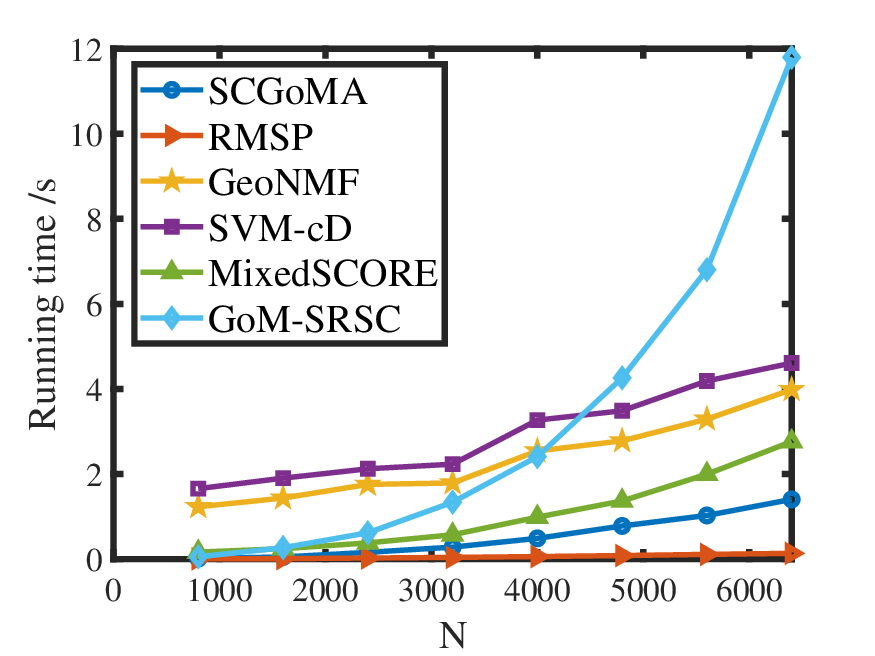}}
\subfigure[Experiment 2(b)]{\includegraphics[width=0.33\textwidth]{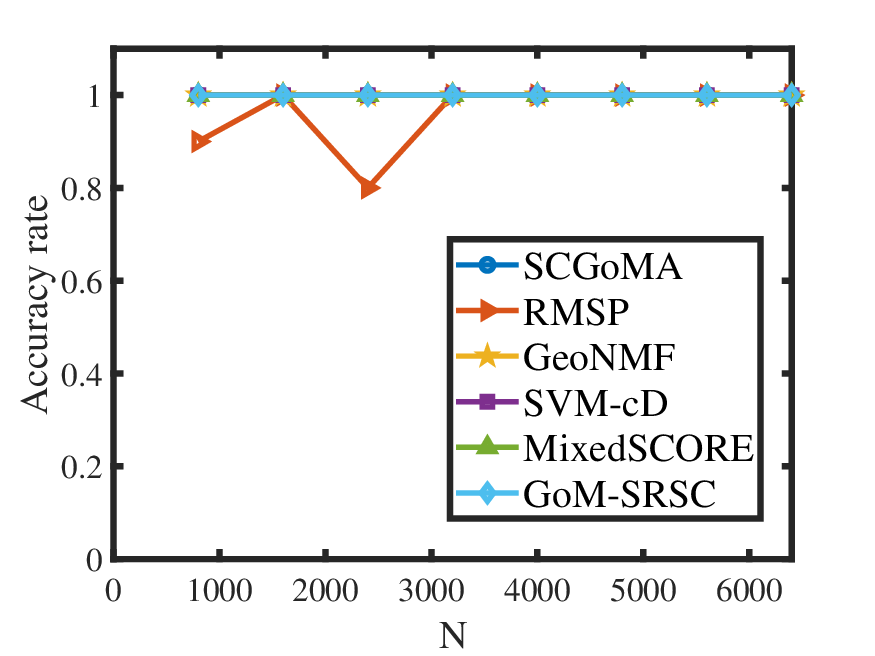}}
}
\resizebox{\columnwidth}{!}{
\subfigure[Experiment 2(c)]{\includegraphics[width=0.33\textwidth]{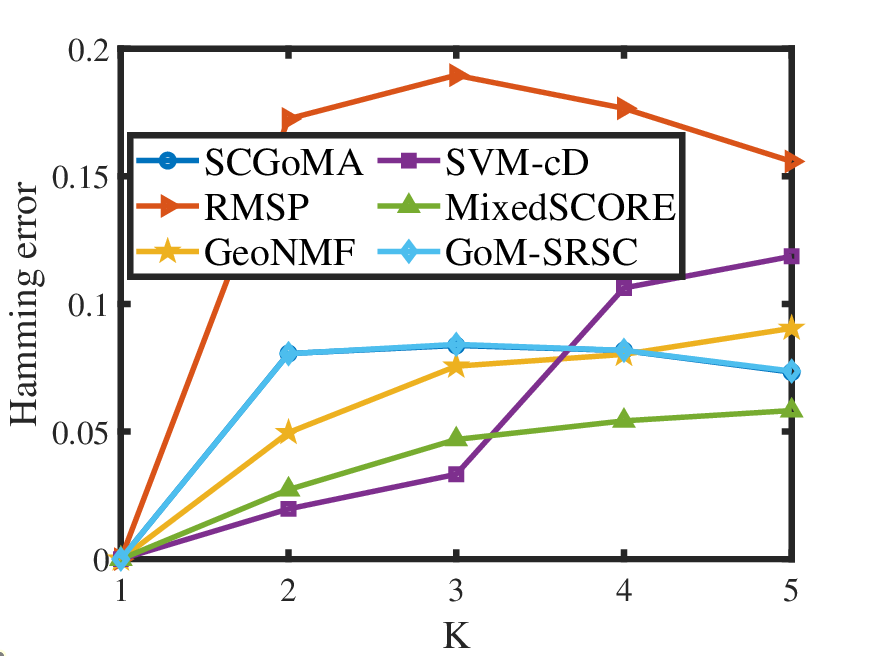}}
\subfigure[Experiment 2(c)]{\includegraphics[width=0.33\textwidth]{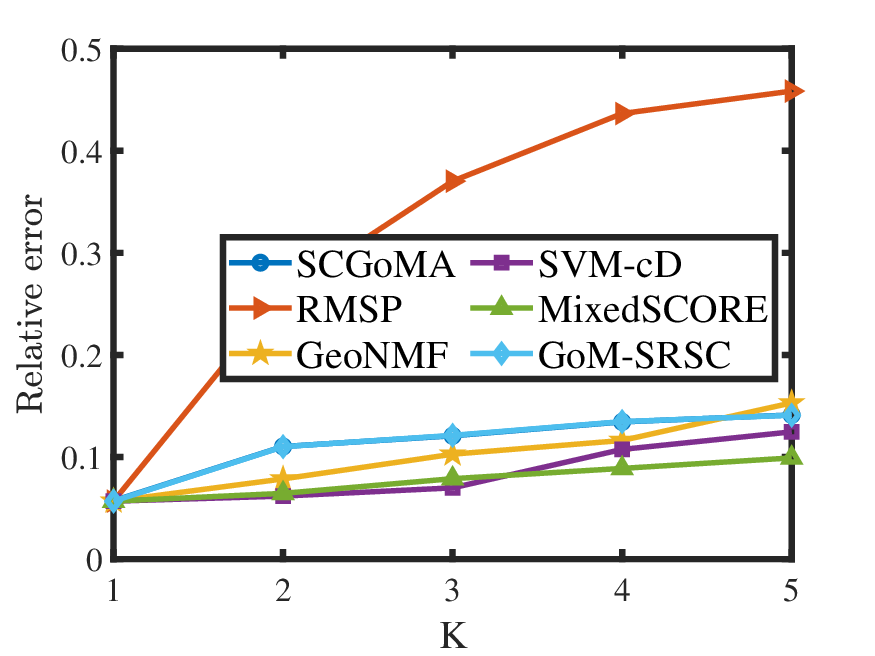}}
\subfigure[Experiment 2(c)]{\includegraphics[width=0.33\textwidth]{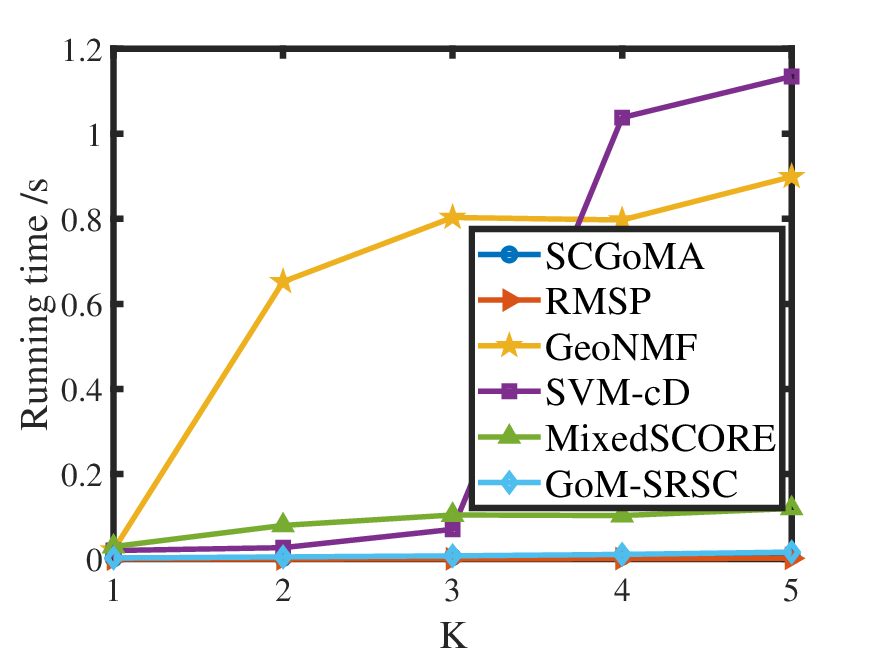}}
\subfigure[Experiment 2(c)]{\includegraphics[width=0.33\textwidth]{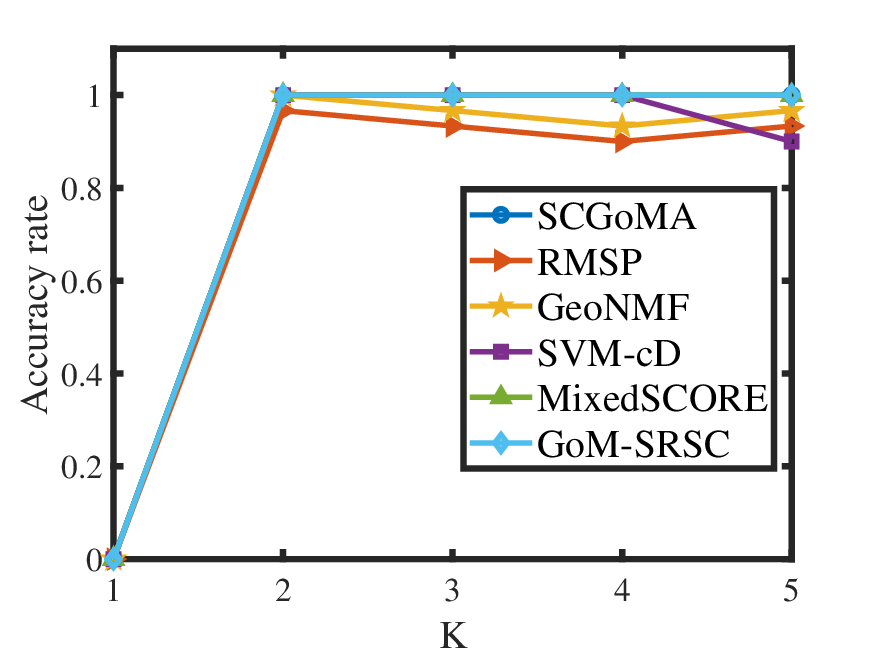}}
}
\resizebox{\columnwidth}{!}{
\subfigure[Experiment 2(d)]{\includegraphics[width=0.33\textwidth]{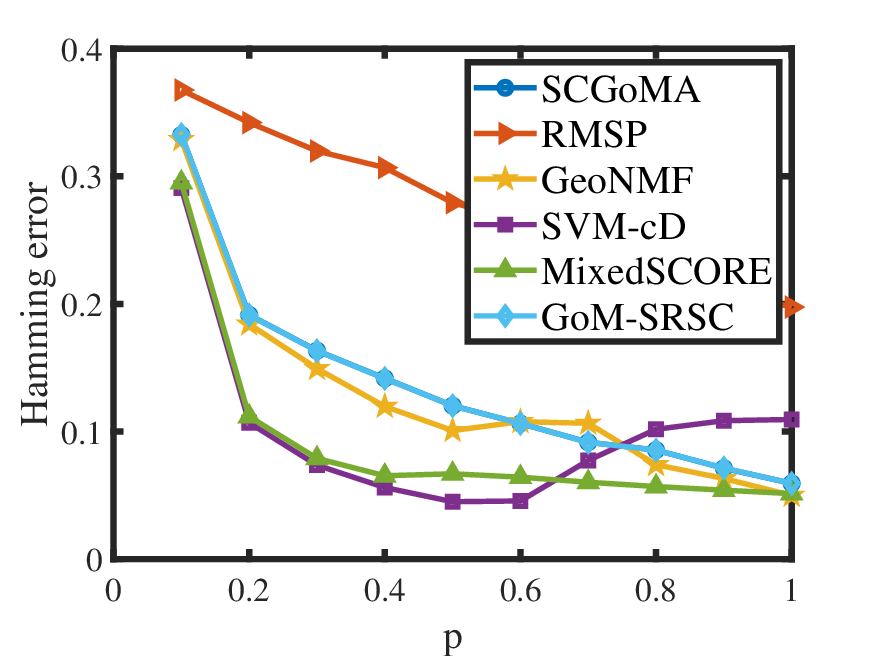}}
\subfigure[Experiment 2(d)]{\includegraphics[width=0.33\textwidth]{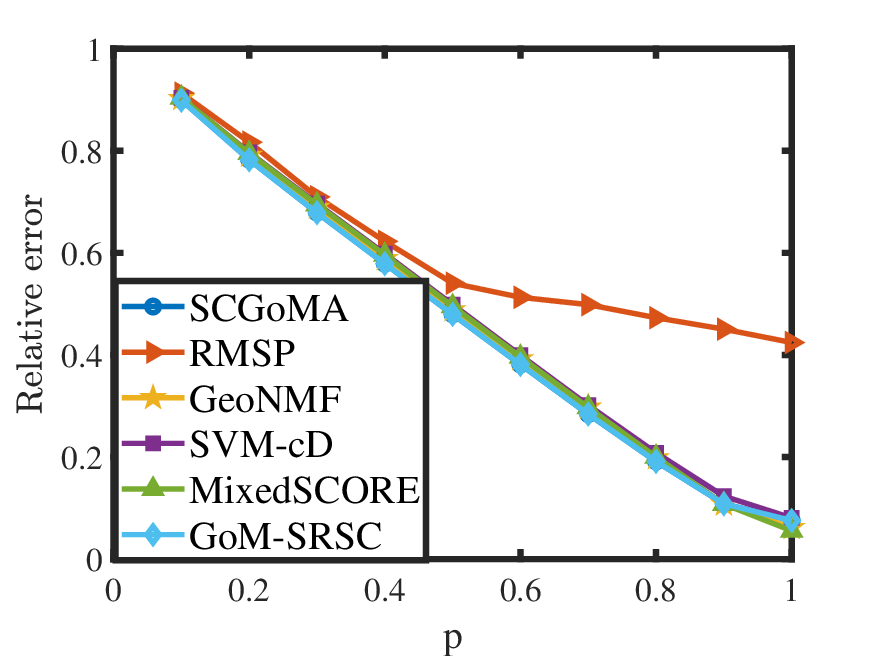}}
\subfigure[Experiment 2(d)]{\includegraphics[width=0.33\textwidth]{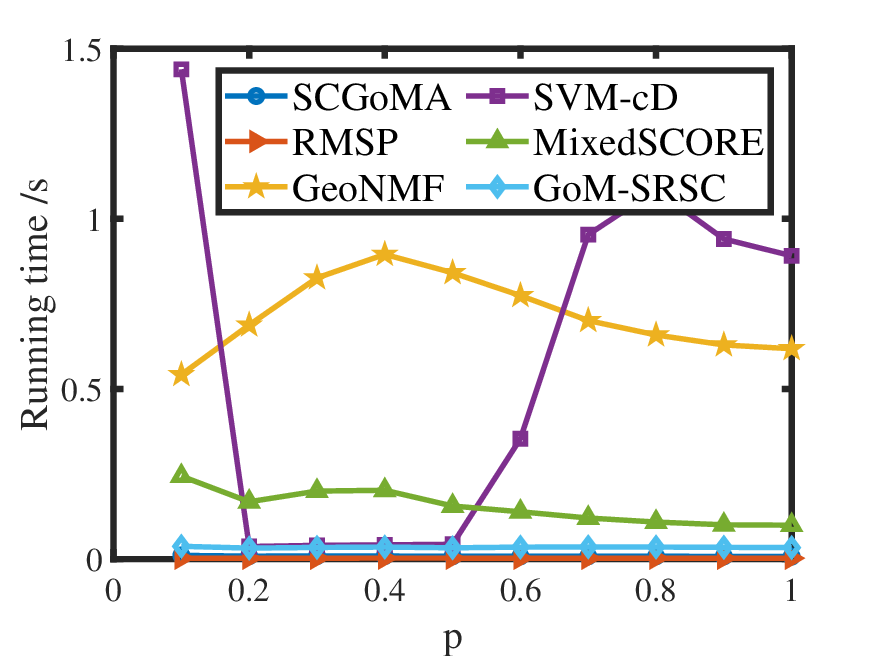}}
\subfigure[Experiment 2(d)]{\includegraphics[width=0.33\textwidth]{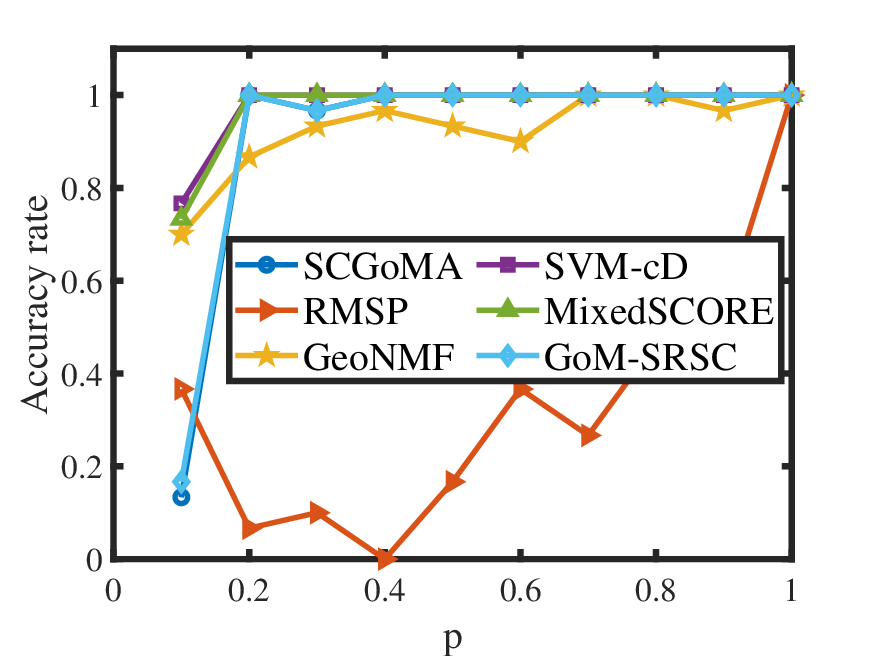}}
}
\caption{Uniform distribution.}
\label{S2} %% label for entire figure
\end{figure}
\subsubsection{Normal distribution}
When $R(i,j)\sim \mathrm{Normal}(R_{0}(i,j),\sigma^{2})$ for $i\in[N]$, $j\in[J]$:

\textbf{Experiment 3(a): changing $\rho$.} Let $N=800$, $\sigma^{2}=1$, $p=1$ and $\rho$ range in $\{0.5,~1,~1.5,~\ldots,~5\}$.

\textbf{Experiment 3(b): changing $N$.} Let $\rho=0.5$, $\sigma^{2}=1$, $p=1$, and $N$ range in $\{800,~1600,~2400,~\ldots,~6400\}$.

\textbf{Experiment 3(c): changing $K$.} Let $\rho=1$, $\sigma^{2}=1$, $p=1$, and $K$ range in $[5]$.

\textbf{Experiment 3(d): changing $p$.} Let $N=800$, $\sigma^{2}=1$, $\rho=1$, and $p$ range in $\{0.1,~0.2,~\ldots,~1\}$.

Fig.~\ref{S3} displays the numerical results, revealing that SCGoMA (similarly to RMSP, GeoNMF, MixedSCORE, and GoM-SRSC) exhibits improved performance as $\rho, N$, and $p$ increase. This supports our discussions in Instance \ref{Normal}, Theorem \ref{mainWGoM}, and Remark \ref{MissingResponses}. As $K$ grows, all methods perform poorer. SCGoMA, GeoNMF, MixedSCORE, and GoM-SRSC demonstrate competitive performance and outperform RMSP and SVM-cD. While RMSP requires less time compared to the other four approaches, it exhibits the poorest performance in estimating $\Pi$ and $\Theta$. Finally, when the true $K$ is no smaller than 2, the high accuracy of all methods except for RMSP in determining $K$ supports the usefulness of the fuzzy weighted modularity.
\begin{figure}
\centering
\resizebox{\columnwidth}{!}{
\subfigure[Experiment 3(a)]{\includegraphics[width=0.33\textwidth]{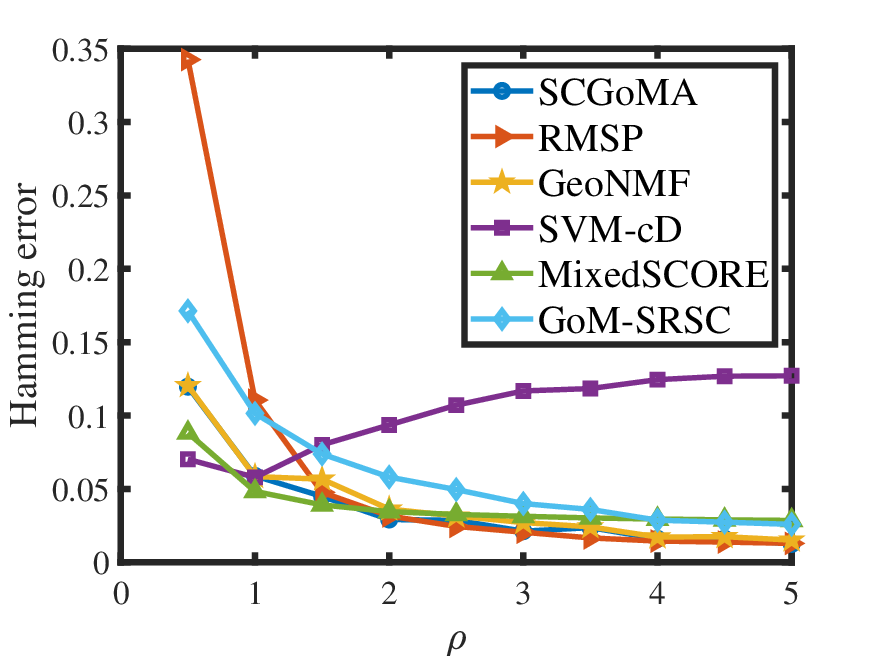}}
\subfigure[Experiment 3(a)]{\includegraphics[width=0.33\textwidth]{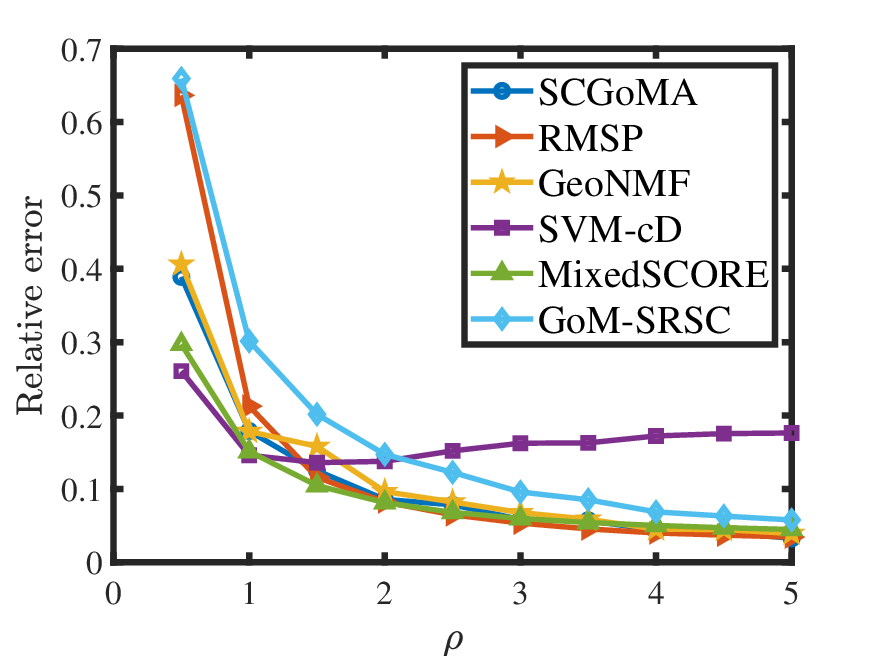}}
\subfigure[Experiment 3(a)]{\includegraphics[width=0.33\textwidth]{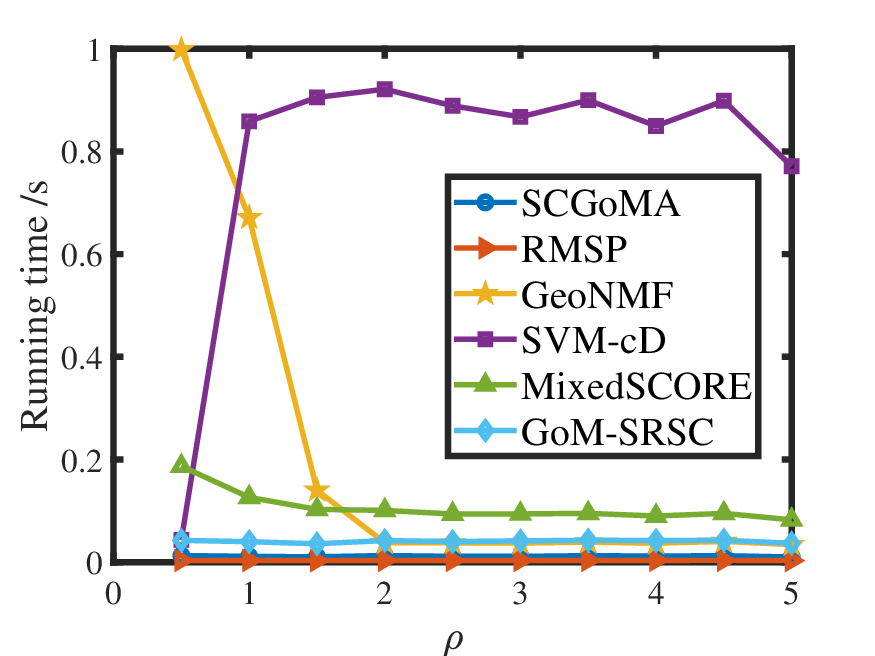}}
\subfigure[Experiment 3(a)]{\includegraphics[width=0.33\textwidth]{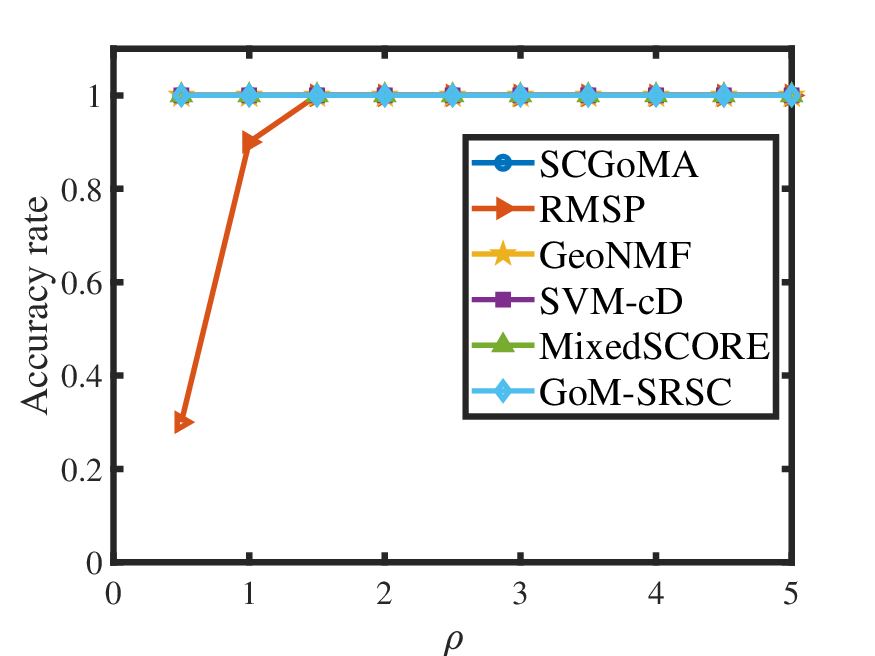}}
}
\resizebox{\columnwidth}{!}{
\subfigure[Experiment 3(b)]{\includegraphics[width=0.33\textwidth]{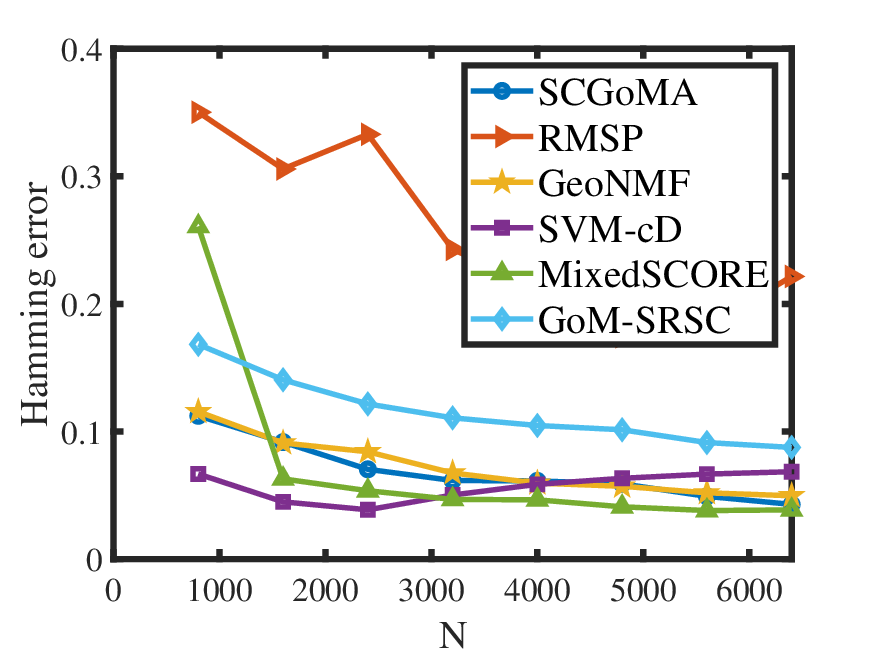}}
\subfigure[Experiment 3(b)]{\includegraphics[width=0.33\textwidth]{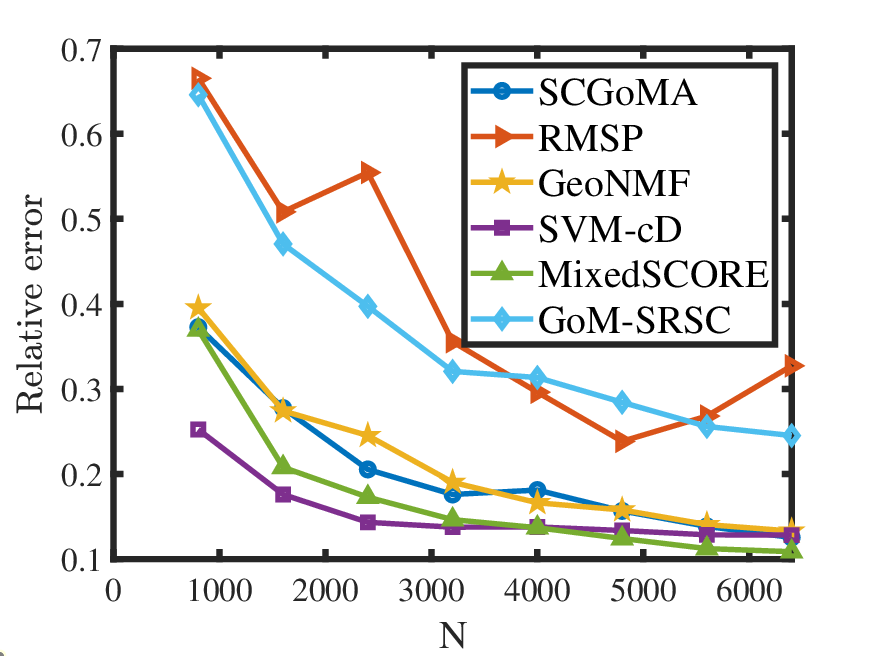}}
\subfigure[Experiment 3(b)]{\includegraphics[width=0.33\textwidth]{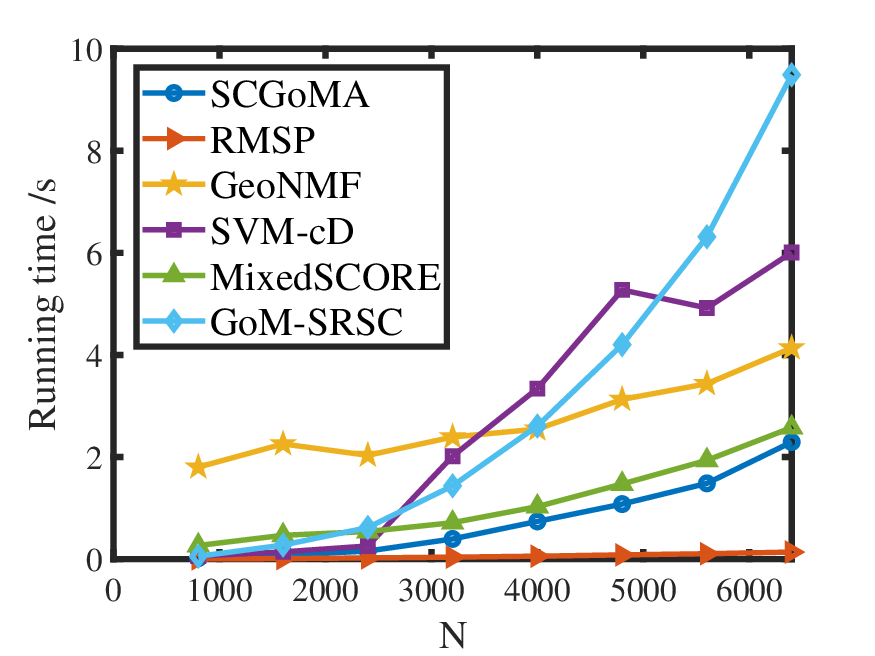}}
\subfigure[Experiment 3(b)]{\includegraphics[width=0.33\textwidth]{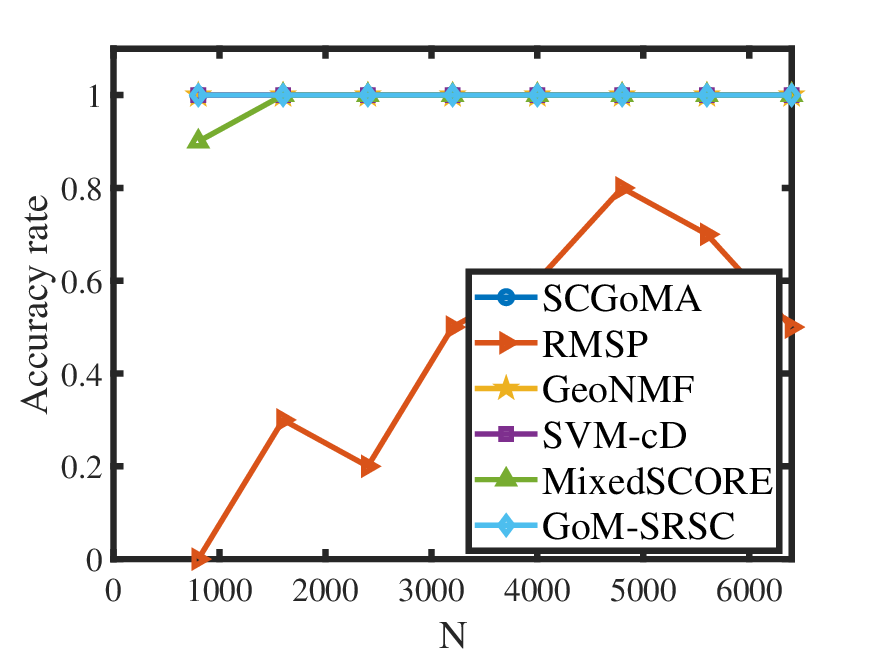}}
}
\resizebox{\columnwidth}{!}{
\subfigure[Experiment 3(c)]{\includegraphics[width=0.33\textwidth]{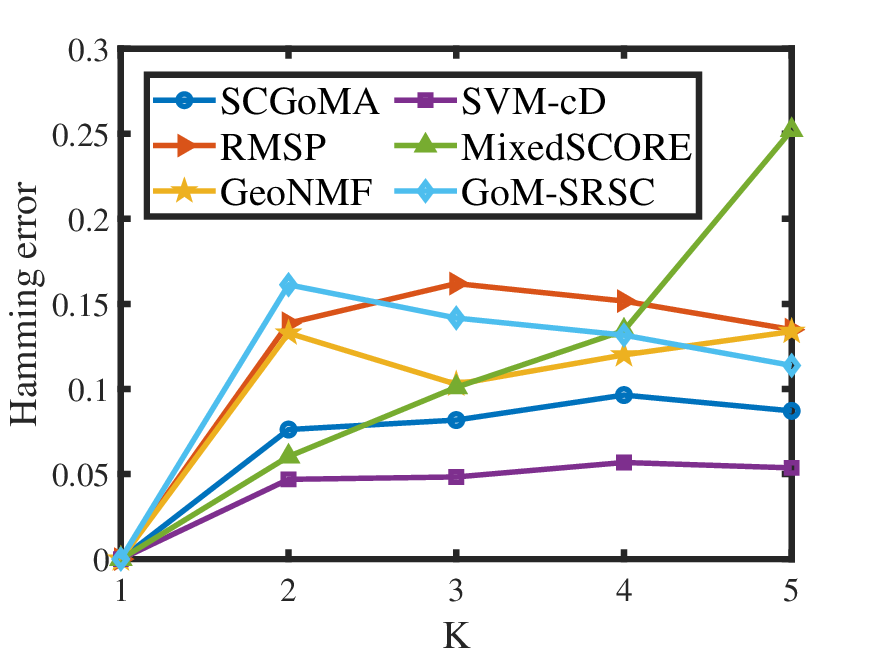}}
\subfigure[Experiment 3(c)]{\includegraphics[width=0.33\textwidth]{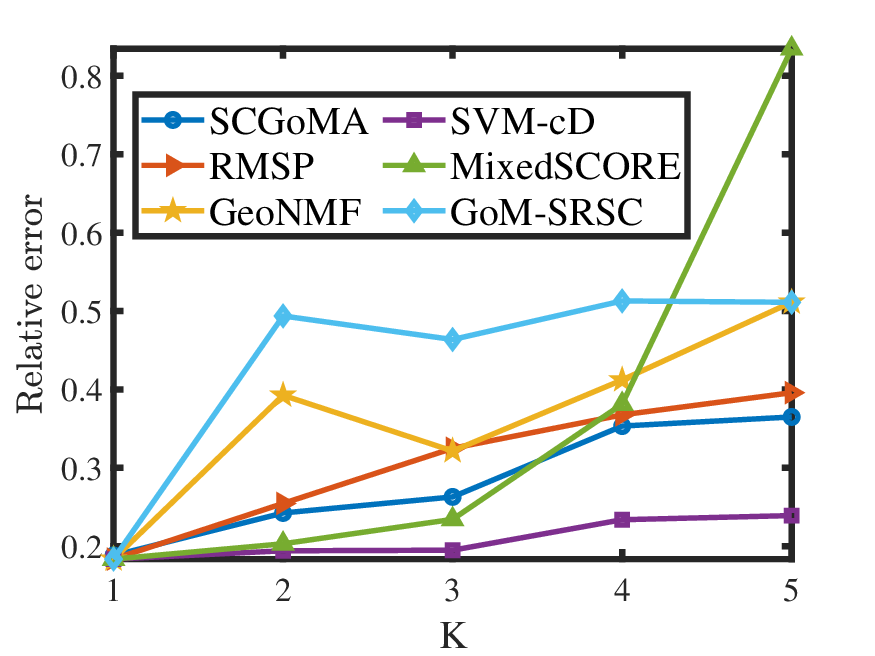}}
\subfigure[Experiment 3(c)]{\includegraphics[width=0.33\textwidth]{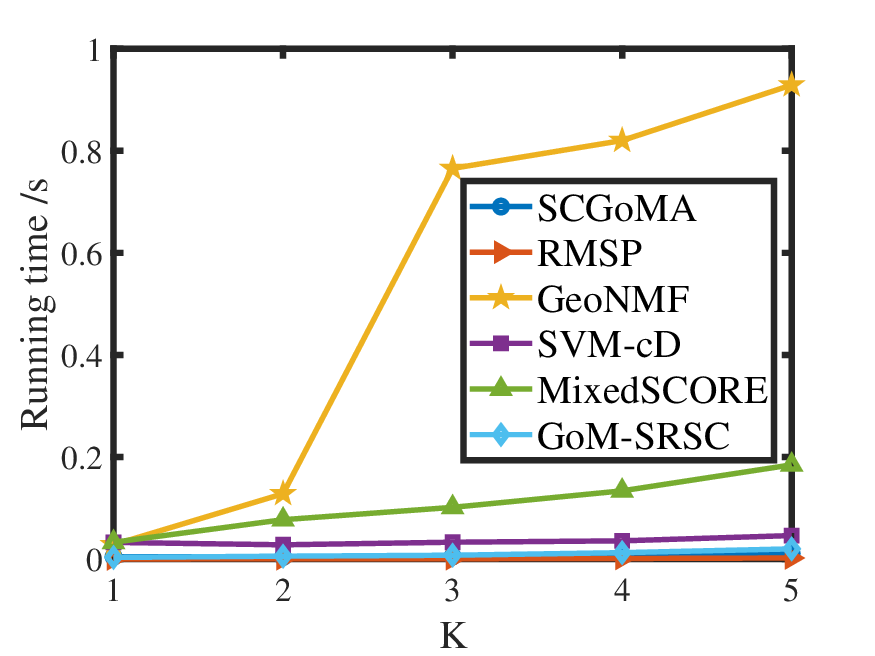}}
\subfigure[Experiment 3(c)]{\includegraphics[width=0.33\textwidth]{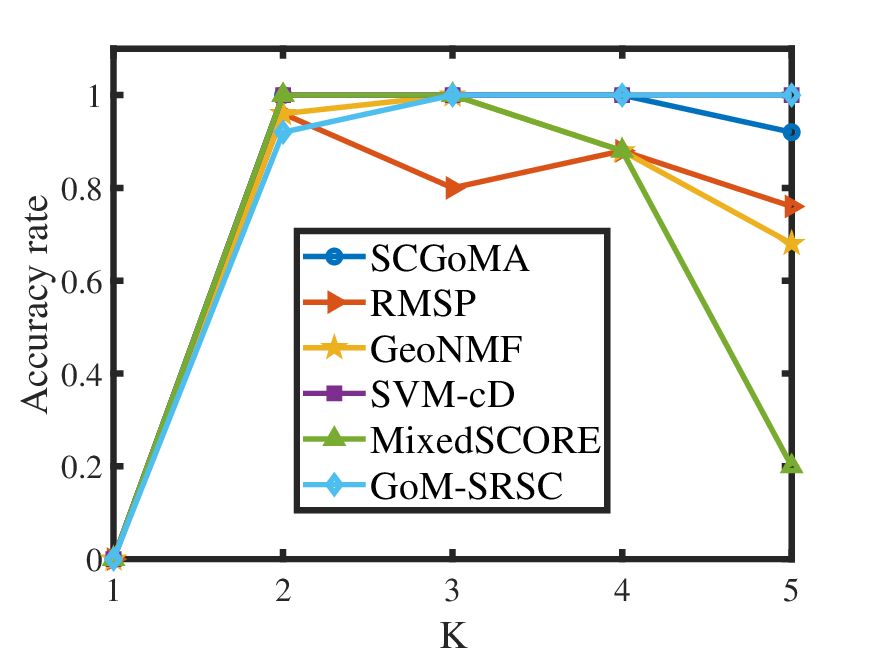}}
}
\resizebox{\columnwidth}{!}{
\subfigure[Experiment 3(d)]{\includegraphics[width=0.33\textwidth]{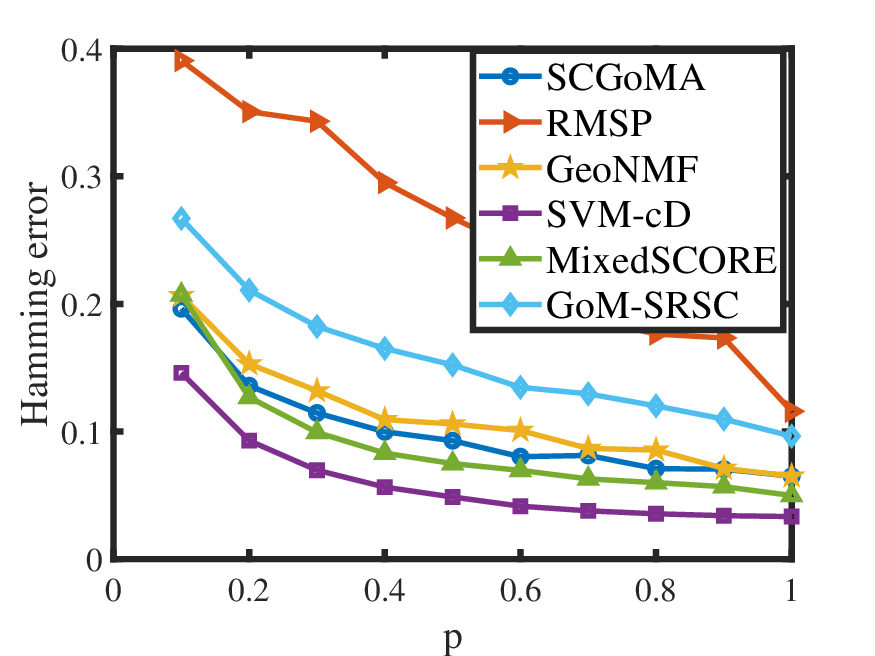}}
\subfigure[Experiment 3(d)]{\includegraphics[width=0.33\textwidth]{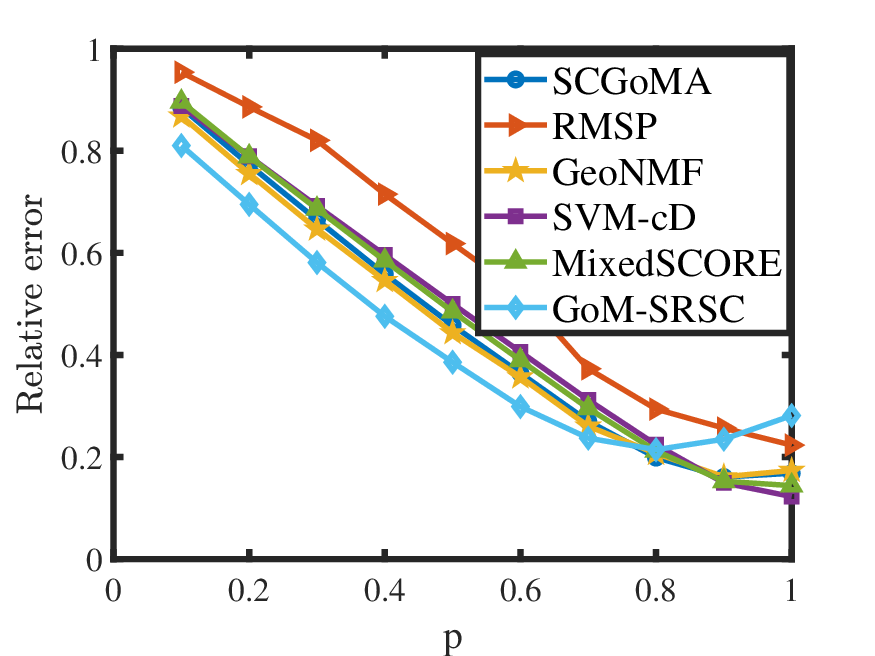}}
\subfigure[Experiment 3(d)]{\includegraphics[width=0.33\textwidth]{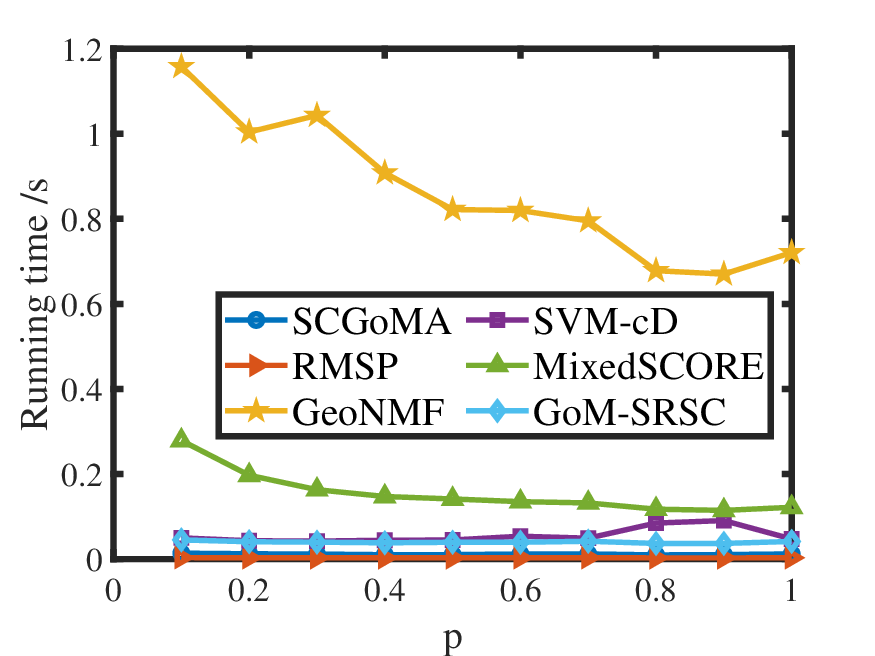}}
\subfigure[Experiment 3(d)]{\includegraphics[width=0.33\textwidth]{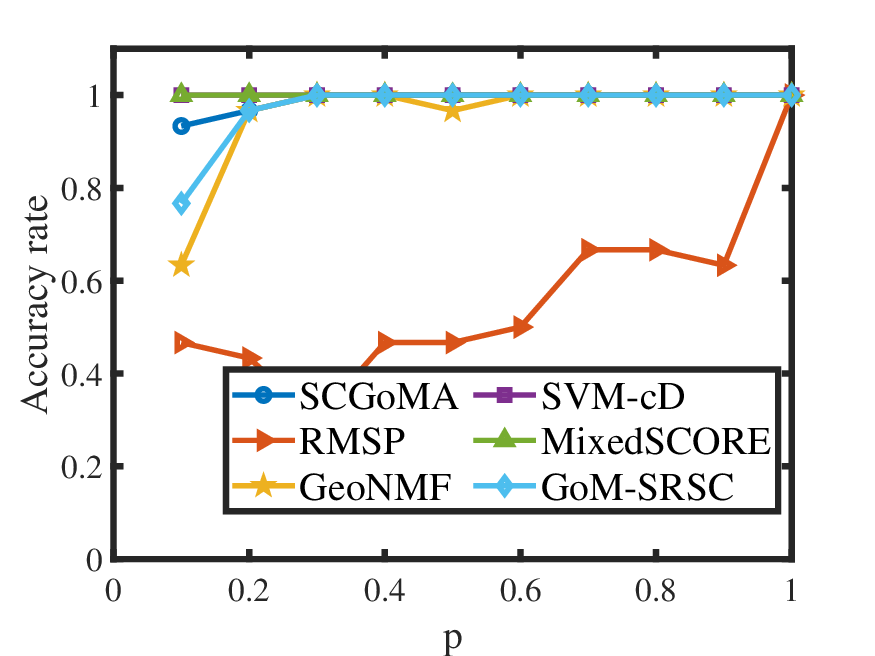}}
}
\caption{Normal distribution.}
\label{S3} %% label for entire figure
\end{figure}
\subsubsection{Signed responses}
When $\mathbb{P}(R(i,j)=1)=\frac{1+R_{0}(i,j)}{2}$ and $\mathbb{P}(R(i,j)=-1)=\frac{1-R_{0}(i,j)}{2}$ for $i\in[N]$, $j\in[J]$:

\textbf{Experiment 4(a): changing $\rho$.} Let $N=800$, $p=1$, and $\rho$ range in $\{0.5,~0.55,~0.6,~\ldots,~1\}$.

\textbf{Experiment 4(b): changing $N$.} Let $\rho=0.6$, $p=1$, and $N$ range in $\{800,~1600,~2400,~\ldots,6400\}$.

\textbf{Experiment 4(c): changing $K$.} Let $\rho=1$, $p=1$, and $K$ range in $[5]$.

\textbf{Experiment 4(d): changing $p$.} Let $N=800$, $\rho=1$, and $p$ range in $\{0.1,~0.2,~\ldots,~1\}$.

Fig.~\ref{S4} presents the results, revealing that as $\rho, N$, and $p$ increase, the accuracies of SCGoMA, RMSP, GeoNMF, MixedSCORE, and GoM-SRSC improve. This aligns with previous our analysis. The performance trends observed are comparable to those in Experiment 3, and we refrain from duplicating the analysis here.
\begin{figure}
\centering
\resizebox{\columnwidth}{!}{
\subfigure[Experiment 4(a)]{\includegraphics[width=0.33\textwidth]{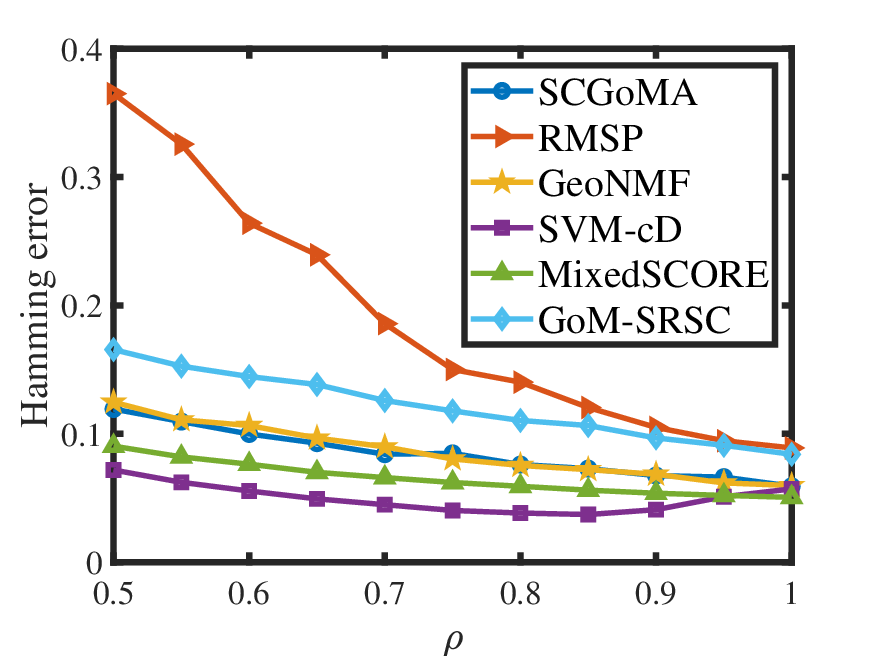}}
\subfigure[Experiment 4(a)]{\includegraphics[width=0.33\textwidth]{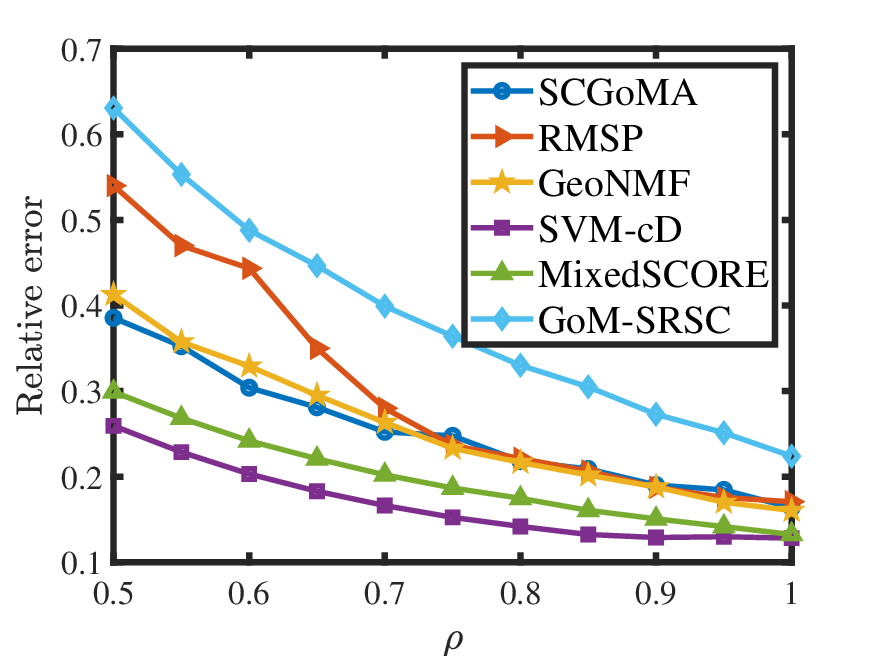}}
\subfigure[Experiment 4(a)]{\includegraphics[width=0.33\textwidth]{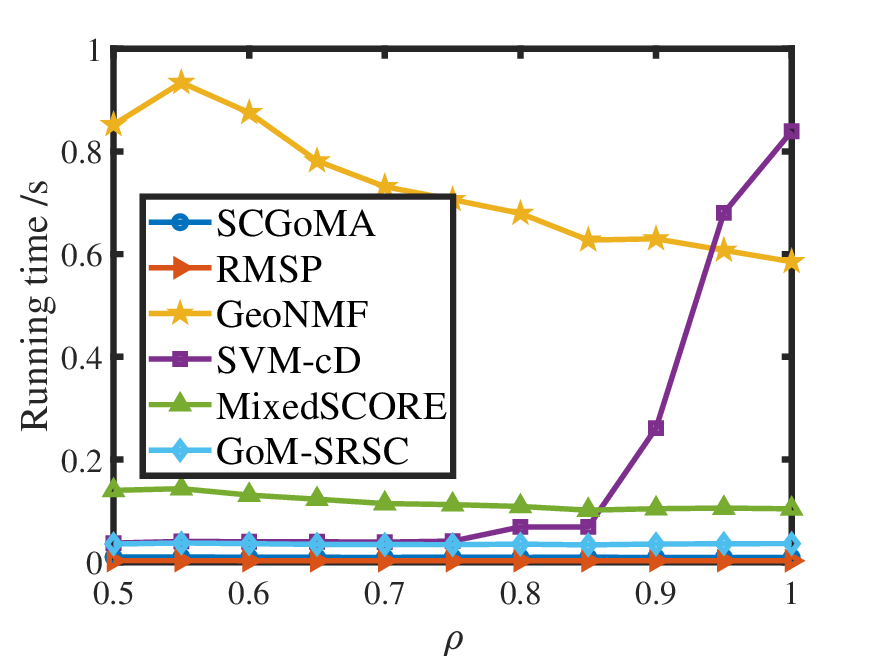}}
\subfigure[Experiment 4(a)]{\includegraphics[width=0.33\textwidth]{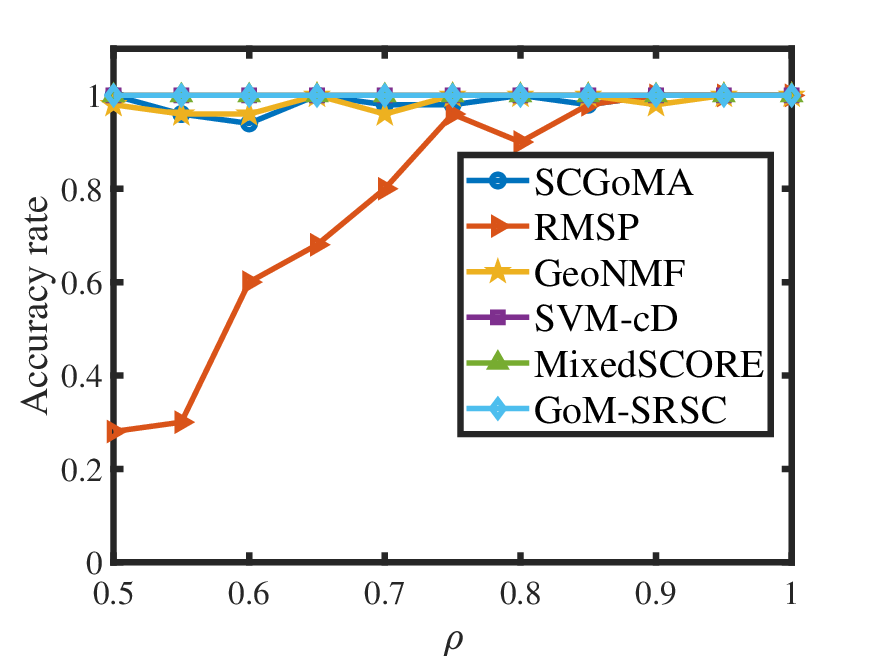}}
}
\resizebox{\columnwidth}{!}{
\subfigure[Experiment 4(b)]{\includegraphics[width=0.33\textwidth]{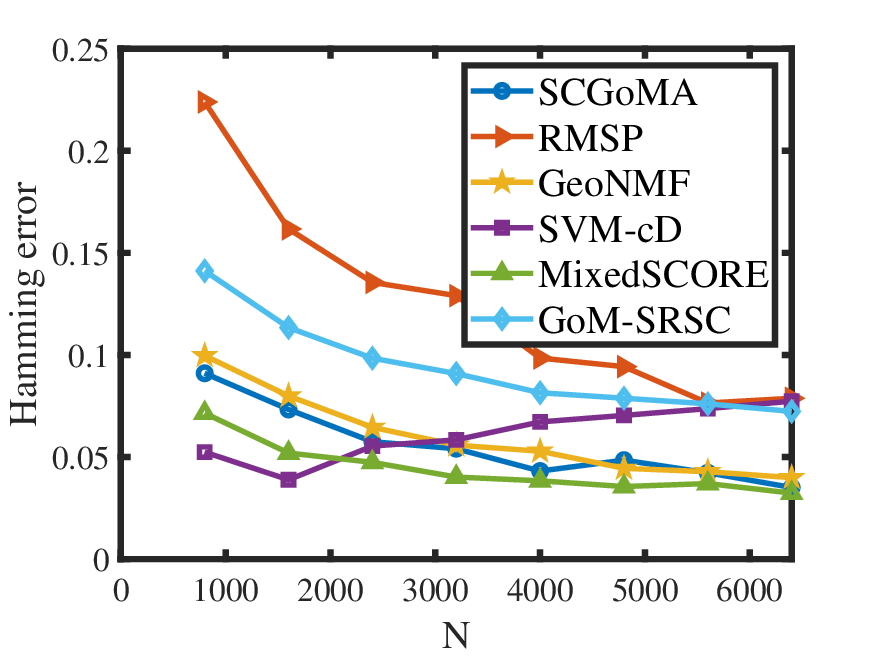}}
\subfigure[Experiment 4(b)]{\includegraphics[width=0.33\textwidth]{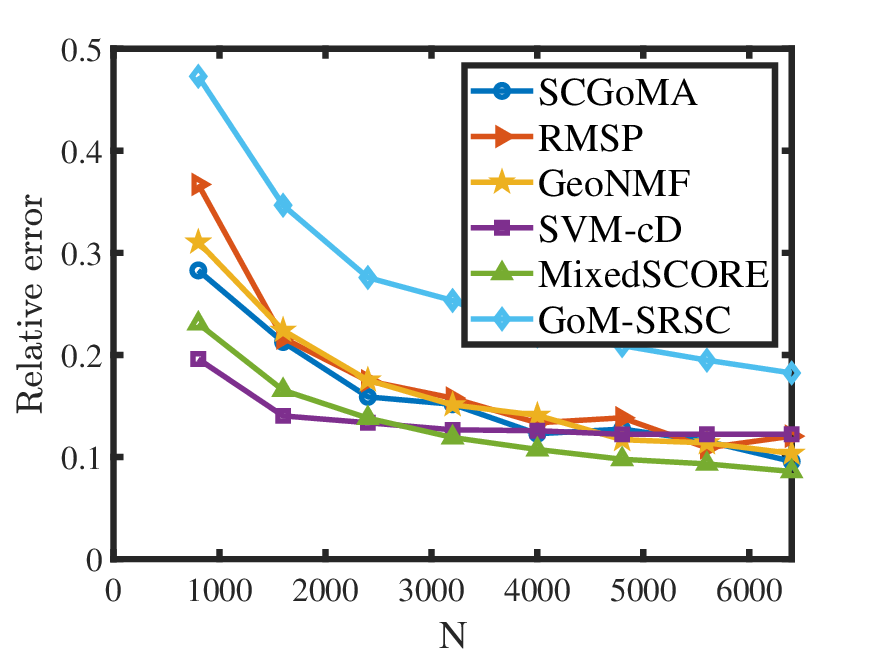}}
\subfigure[Experiment 4(b)]{\includegraphics[width=0.33\textwidth]{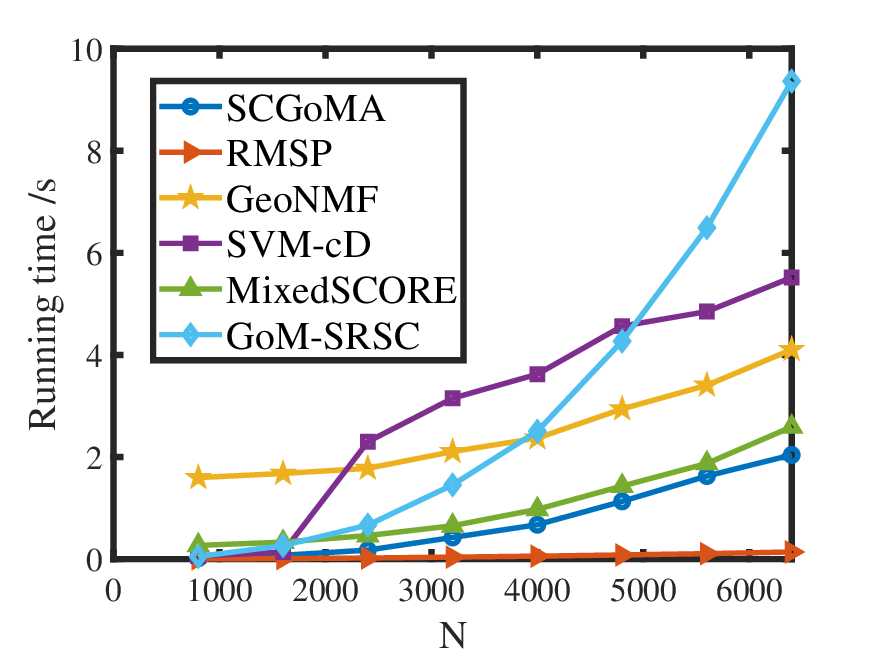}}
\subfigure[Experiment 4(b)]{\includegraphics[width=0.33\textwidth]{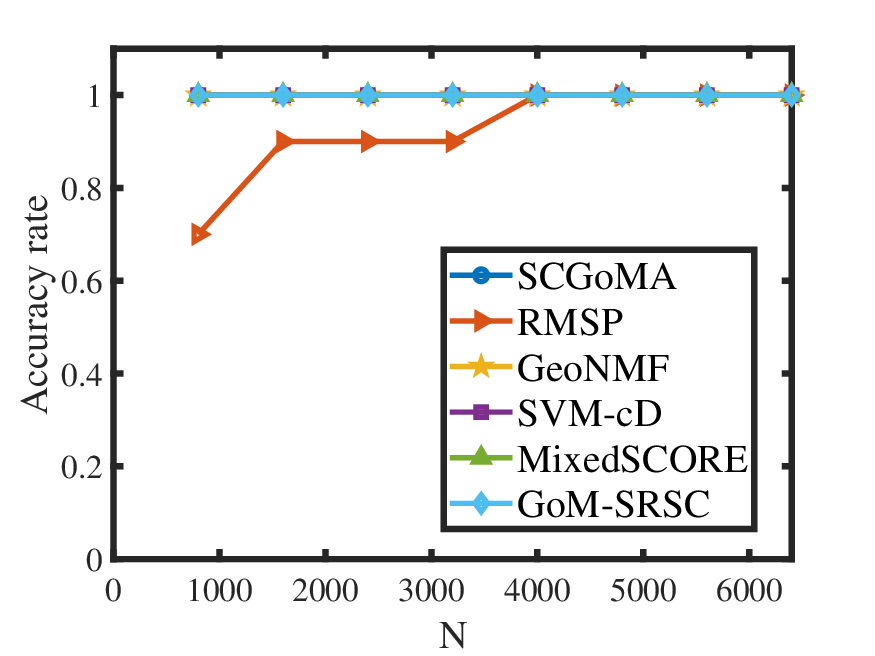}}
}
\resizebox{\columnwidth}{!}{
\subfigure[Experiment 4(c)]{\includegraphics[width=0.33\textwidth]{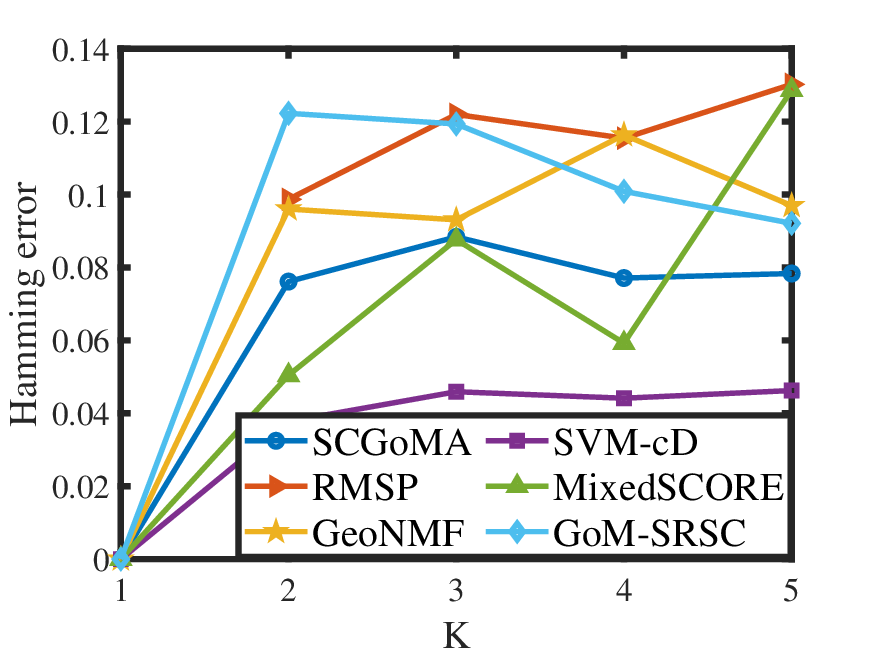}}
\subfigure[Experiment 4(c)]{\includegraphics[width=0.33\textwidth]{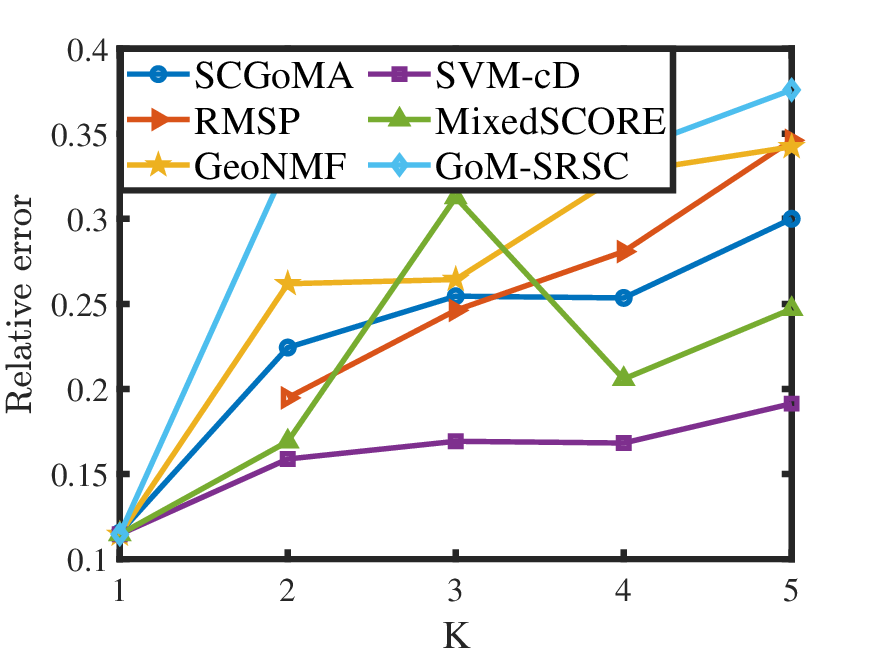}}
\subfigure[Experiment 4(c)]{\includegraphics[width=0.33\textwidth]{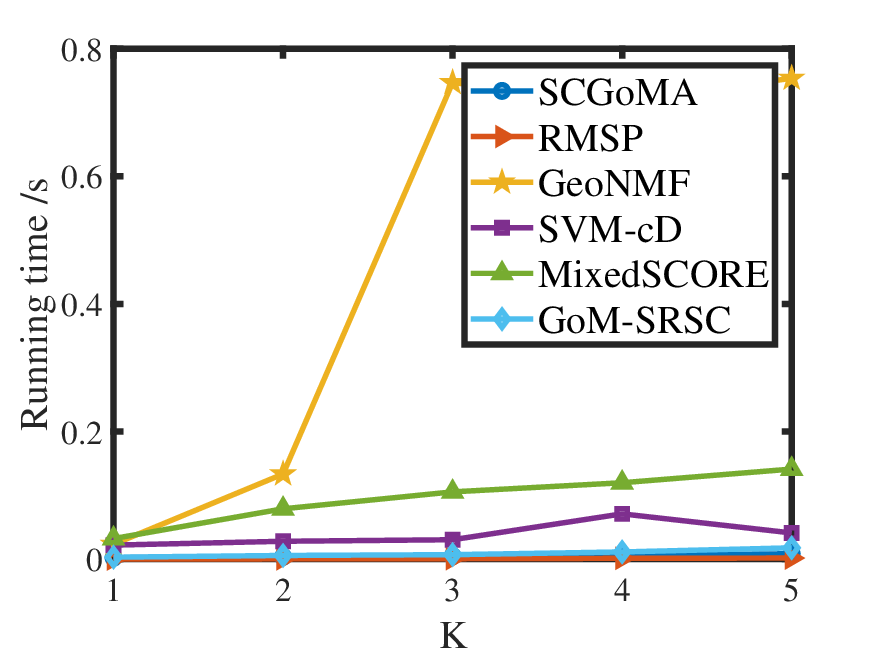}}
\subfigure[Experiment 4(c)]{\includegraphics[width=0.33\textwidth]{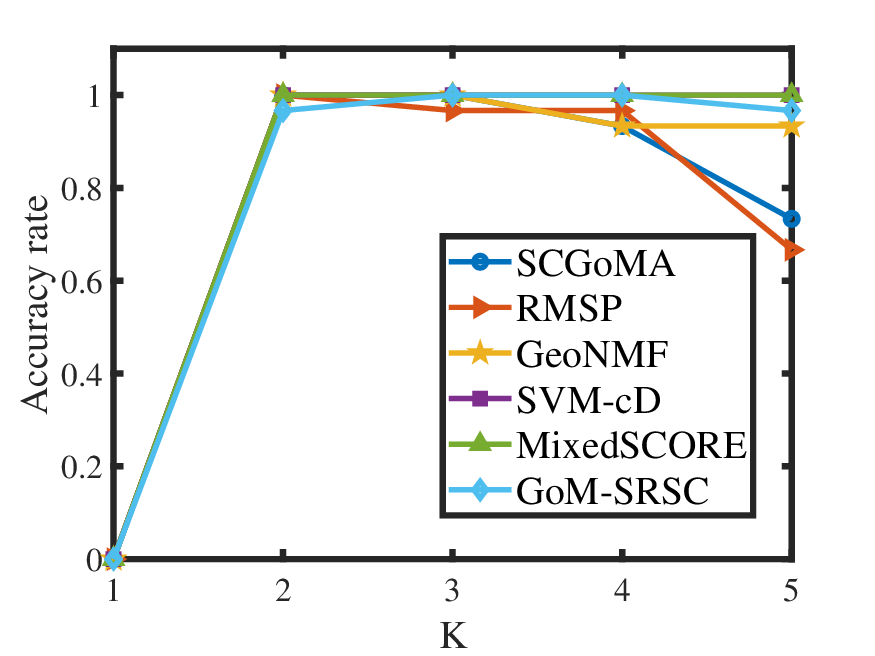}}
}
\resizebox{\columnwidth}{!}{
\subfigure[Experiment 4(d)]{\includegraphics[width=0.33\textwidth]{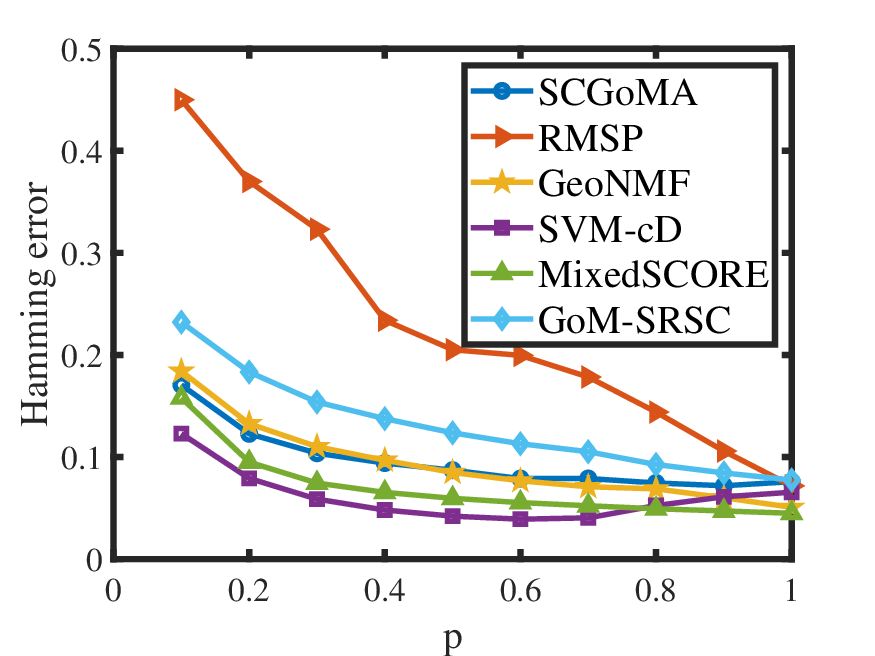}}
\subfigure[Experiment 4(d)]{\includegraphics[width=0.33\textwidth]{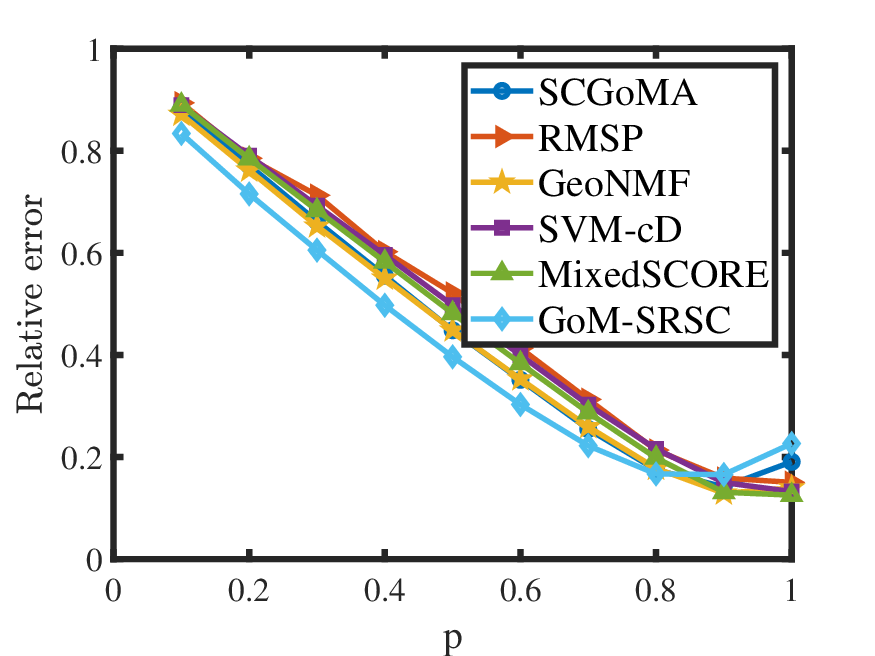}}
\subfigure[Experiment 4(d)]{\includegraphics[width=0.33\textwidth]{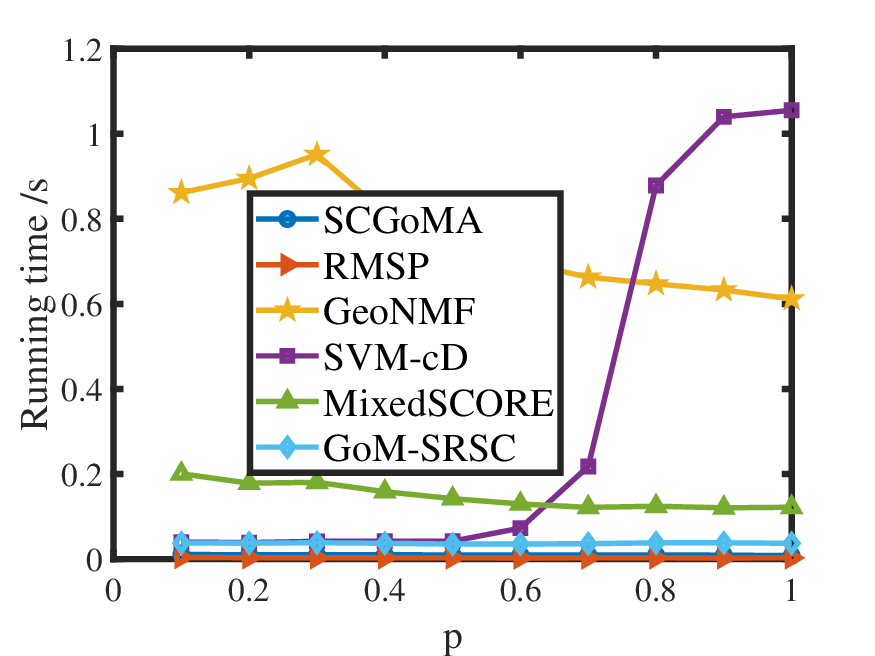}}
\subfigure[Experiment 4(d)]{\includegraphics[width=0.33\textwidth]{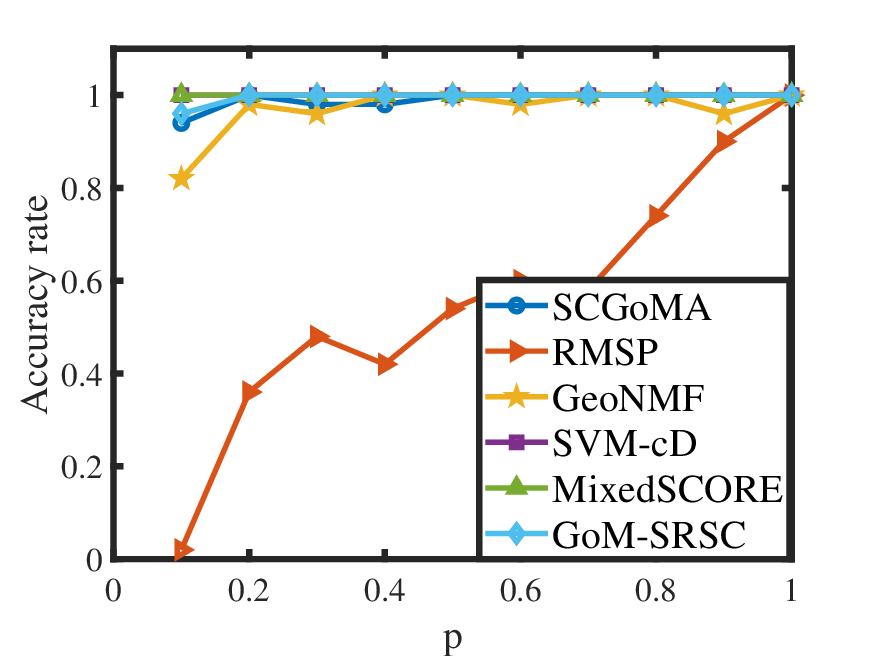}}
}
\caption{Signed responses.}
\label{S4} %% label for entire figure
\end{figure}
\section{Real data analysis}\label{sec7}
In this section, we apply the aforementioned methods to several real-world categorical datasets. Define $\xi=\frac{\mathrm{Number~of~no~responses}}{NJ}$ to describe data's sparsity. After removing subjects that have no response to all items and items that do not receive any response, Table \ref{realdata} summarizes the basic information of these datasets. The results show that Advogato, Wikipedia elections, Filmtrust, and Bitcoin Alpha are much sparser than TMAS. Meanwhile, the first four datasets can be downloaded from \url{http://konect.cc/networks/} while TMAS can be downloaded from \url{https://openpsychometrics.org/_rawdata/}.

\begin{table}[h!]
\tiny
	\centering
	\caption{Basic information of real-world datasets analyzed in this paper, where TMAS is short for Taylor Manifest Anxiety Scale data. Here, $i\in[N]$ and $j\in[J]$.}
	\label{realdata}
\setlength{\tabcolsep}{2pt} % 调整列间距
	\begin{tabular}{cccccccccccc}
\hline
Datasets&Subject meaning&Items meaning&Responses meaning&$N$&$J$&$R(i,j)$'s range&$\xi$\\
\hline
Advogato \citep{massa2009bowling}&Users&Users&Trust&5775&6540&$\{0,~0.6,~0.8,~1\}$&0.9986\\
Wikipedia elections \citep{leskovec2010governance}&Users&Users&Vote&5842&2355&$\{-1,~0,~1\}$&0.9926\\
Filmtrust \citep{guo2016novel}&Users&Films&Rating&1508&2071&$\{0,~0.5,~\ldots,~4.5\}$&0.9886\\
Bitcoin Alpha \citep{kumar2016edge}&Users&Users&Trust&3270&3727&$\{-10,~-9,~\ldots,~10\}$&0.9980\\
TMAS \citep{taylor1985taylor}&Individuals&Statements&Answers&5401&50&$\{0,~1,~2\}$&0.0080\\
\hline
\end{tabular}
\end{table}

Given the unknown true mixed memberships and number of latent classes in these datasets, we apply SCGoMA, GeoNMF, SVM-cD, MixedSCORE, and GoM-SRSC to the observed response matrices to compute the fuzzy weighted modularity in Equation (\ref{Qfwm}). RMSP is not considered here as it fails to output for these datasets. The results are presented in Table \ref{realdataQfwm}. Our observations are as follows:

\begin{itemize}
  \item For the Advogato dataset, GeoNMF achieves the highest modularity score, while the other four methods yield comparable modularity scores. Based on these results, we tend to believe that the estimated number of latent classes should be $3$, as determined by GeoNMF.
  \item For the Wikipedia elections dataset, SCGoMA outperforms the other four methods, with SVM-cD and MixedSCORE demonstrating similar performance. GeoNMF and GoM-SRSC return lower modularity scores than SCGoMA, SVM-cD, and MixedSCORE. For this dataset, we suggest that the estimated number of latent classes should be $K=2$, as estimated by SCGoMA.
  \item For the Filmtrust dataset, SVM-cD exhibits the best performance, with the other four methods showing comparable results. Therefore, we propose that $K$ should be $2$, as estimated by SVM-cD.
  \item For the Bitcoin Alpha dataset, SVM-cD achieves the highest modularity score, with SCGoMA yielding a comparable score. Both methods outperform GeoNMF, MixedSCORE, and GoM-SRSC. For this dataset, we recommend $K=3$, as estimated by SVM-cD.
  \item For the TMAS dataset, all methods estimate the number of latent classes as $2$ and demonstrate comparable performance.
\end{itemize}

\begin{table}[h!]
\footnotesize
	\centering
	\caption{$(\hat{K}_{\mathcal{M}}, Q_{\mathcal{M}}(\hat{K}_{\mathcal{M}}))$ for real data analyzed in this paper, with the largest $Q_{\mathcal{M}}(\hat{K}_{\mathcal{M}})$ in bold.}
	\label{realdataQfwm}
	\begin{tabular}{cccccccccccc}
\hline
Dataset&SCGoMA&GeoNMF&SVM-cD&MixedSCORE&GoM-SRSC\\
\hline
Advogato&(5,~0.1030)&(3,~\textbf{0.1079})&(3,~0.1047)&(4,~0.1031)&(5,~0.1033)\\
Wikipedia elections&(2,~\textbf{0.1934})&(8,~0.0456)&(3,~0.1714)&(3,~0.1629)&(4,~0.0184)\\
Filmtrust&(3,~0.0205)&(2,~0.0290)&(2,~\textbf{0.0415})&(2,~0.0395)&(5,~0.0199)\\
Bitcoin Alpha&(4,~0.1564)&(9,~0.1006)&(3,~\textbf{0.1738})&(3,~0.1193)&(11,~0.0306)\\
TMAS&(2,~0.0013)&(2,~0.0032)&(2,~\textbf{0.0041})&(2,~0.0039)&(2,~0.0013)\\
\hline
\end{tabular}
\end{table}

For each dataset, we compare the running time of each method using the estimated $K$ from the above analysis. The results are presented in Table \ref{realdataTime}. We see that SCGoMA consistently outperforms the other methods in terms of running time, achieving the smallest average running time for all datasets. GeoNMF, SVM-cD, MixedSCORE, and GoM-SRSC exhibit varying performance, with GeoNMF and GoM-SRSC generally being slower.

\begin{table}[h!]
\footnotesize
	\centering
	\caption{The running time (in seconds) of each method for the real datasets analyzed in this paper, with the smallest running time in bold. The number of latent classes for each dataset corresponds to the one determined after Table \ref{realdataQfwm}. The running times reported are the average of 20 independent repetitions for each method.}
	\label{realdataTime}
	\begin{tabular}{cccccccccccc}
\hline
Dataset&SCGoMA&GeoNMF&SVM-cD&MixedSCORE&GoM-SRSC\\
\hline
Advogato&\textbf{1.9861s}&3.3160s&7.6530s&4.9007s&9.4734s\\
Wikipedia elections&\textbf{0.6329s}&1.9236s&1.3812s&2.1633s& 5.8805s\\
Filmtrust&\textbf{0.0823s}&0.6838s&0.7412s&0.2109s&0.2187s\\
Bitcoin Alpha&\textbf{0.5410s}&1.4054s&0.5434s&1.3323s&1.5160s\\
TMAS&\textbf{0.0225s}&4.7885s&0.6968s&1.7605s&4.0447s\\
\hline
\end{tabular}
\end{table}

Given the results presented in Section \ref{sec6} and Table \ref{realdataQfwm}, which demonstrate that SCGoMA's performance is comparable to that of GeoNMF, SVM-cD, MixedSCORE, and GoM-SRSC, our subsequent analysis focuses exclusively on the performance of our SCGoMA approach. We apply it to the datasets listed in Table \ref{realdata}, using the value of $K$ determined for each dataset based on Table \ref{realdataQfwm}. By running SCGoMA on each dataset with the specified $K$, we obtain matrices $\hat{\Pi} \in [0,1]^{N \times K}$ and $\hat{\Theta} \in \mathbb{R}^{J \times K}$, which enable further analysis.

Our primary interest lies in the mixed memberships of subjects. We classify a subject $i$ as highly mixed if $\mathrm{max}_{k \in [K]} \hat{\Pi}(i,k) \leq 0.6$ and as highly pure if $\mathrm{max}_{k \in [K]} \hat{\Pi}(i,k) \geq 0.9$. We define $\omega_{\mathrm{mixed}}$ as the proportion of highly mixed subjects and $\omega_{\mathrm{pure}}$ as the proportion of highly pure subjects, calculated as $\omega_{\mathrm{mixed}} = \frac{\text{Number of highly mixed subjects}}{N}$ and $\omega_{\mathrm{pure}} = \frac{\text{Number of highly pure subjects}}{N}$, respectively. Additionally, we introduce the index $\eta = \frac{\mathrm{min}_{k \in [K]} \sum_{i \in [N]} \hat{\Pi}(i,k)}{\mathrm{max}_{k \in [K]} \sum_{i \in [N]} \hat{\Pi}(i,k)}$, which ranges in $(0,1]$. This index serves as a measure of the balance in the distribution of subjects across latent classes. A higher value of $\eta$ indicates a more balanced distribution of subjects among the latent classes, while a lower value suggests a more unbalanced distribution, with some latent classes having significantly more subjects than others (indicating that subjects are more likely to belong to these latent classes).

The results in Table \ref{realdata3indices} reveal notable differences in the mixed membership structures among the datasets. The Wikipedia elections dataset has the lowest proportion of highly mixed subjects ($\omega_{\mathrm{mixed}} = 4.23\%$) and the highest proportion of highly pure subjects ($\omega_{\mathrm{pure}} = 81.63\%$), along with a highly balanced distribution of subjects across latent classes ($\eta = 0.8779$). In contrast, the Bitcoin Alpha dataset exhibits the highest proportion of highly mixed subjects ($\omega_{\mathrm{mixed}} = 21.22\%$) and a relatively balanced proportion of highly pure subjects ($\omega_{\mathrm{pure}} = 50.86\%$), but a highly unbalanced distribution of subjects across latent classes ($\eta = 0.2328$). The Advogato, Filmtrust, and TMAS datasets show intermediate values for these indices, with TMAS having an almost equal proportion of highly mixed and highly pure subjects ($\omega_{\mathrm{mixed}} = 17.86\%$ and $\omega_{\mathrm{pure}} = 17.76\%$). These variations highlight the diverse mixed membership structures present in different datasets.

\begin{table}[h!]
\footnotesize
	\centering
	\caption{$\omega_{\mathrm{mixed}}$, $\omega_{\mathrm{pure}}$, and $\eta$ for real data in this paper.}
	\label{realdata3indices}
	\begin{tabular}{cccccccccccc}
\hline
Dataset&$\omega_{\mathrm{mixed}}$&$\omega_{\mathrm{pure}}$&$\eta$\\
\hline
Advogato&16.35\%&49.14\%&0.6164\\
Wikipedia elections&4.23\%&81.63\%&0.8779\\
Filmtrust&10.41\%&34.81\%&0.4428\\
Bitcoin Alpha&21.22\%&50.86\%&0.2328\\
TMAS&17.86\%&17.76\%&0.4425\\
\hline
\end{tabular}
\end{table}
\subsection{The TMAS data}
Given that TMAS is a psychological test dataset comprising only 50 statements (i.e., items), we provide some additional analysis of this data in this subsection.

\textbf{Background.} The TMAS data contains 50 true-false statements (see panel (b) of Fig.~\ref{TMASHeatmap} for details). For this dataset, 1 indicates true, 2 indicates false, 0 means not answered. The response matrix $R$ is of dimension $5401\times50$ and takes values in the set $\{0,~1,~2\}$.

\textbf{Analysis.} Results in Table \ref{realdataQfwm} suggest that the estimated number of latent classes for the TMAS data is 2. Applying the SCGoMA method to $R$ with $K=2$ obtains the $5401\times 2$ estimated membership matrix $\hat{\Pi}$ and the $50\times2$ estimated item parameter matrix $\hat{\Theta}$.

\textbf{Results.} For convenience, we use class 1 and class 2 to represent the two estimated latent classes. To better understand  $\hat{\Pi}$ and $\hat{\Theta}$, we plot the membership vectors of the top 50 subjects in panel (a) of Fig.~\ref{TMASHeatmap} and display the heatmap of $\hat{\Theta}$ in panel (b) of Fig.~\ref{TMASHeatmap}, where we only present a subset of the subjects because the number of subjects $N=5401$ is large. Based on Equation (\ref{Rij}), we have the following observation: for subject $i\in[N]$, given its membership score $\Pi(i,k)$ on the $k$-th latent class, if $\Theta(j,k)$ is larger, then the expected value of $R(i,j)$ is larger for the $j$-th item, i.e., $R(i,j)$ tends to be larger if $\Theta(j,k)$ is larger. Recall that in the response matrix $R$ for the TMAS data, 1 represents true and 2 represents false. The above observation suggests that for subject $i\in[N]$, item $j\in[J]$, and class $k\in[K]$, the response for the $i$-th subject to the $j$-th item tends to be false for a larger $\hat{\Theta}(j,k)$ and tends to be true for a smaller $\hat{\Theta}(j,k)$.

Based on the above analysis and carefully examining panel (b) of Fig.~\ref{TMASHeatmap}, class 1 can be interpreted as individuals with a positive mindset towards life, while class 2 can be interpreted as individuals with a negative mindset towards life. Panel (b) of Fig.~\ref{TMASHeatmap} suggests that individuals belonging to class 1 tend to possess positive characteristics such as energy, gentleness, relaxation, concentration, strength, etc.,  while individuals belonging to class 2 tend to have an opposite mindset compared to those in class 1. Although the mindsets of class 1 and class 2 are opposite, for any individual $i\in[N]$, it has different weights towards the two classes, i.e., the $i$-th individual has different extent for negative and positive mindset towards life. For example, from the left panel of Fig.~\ref{TMASHeatmap}, we find that the membership score of subject 2 is $(0.6968, 0.3032)$ which can be interpreted as subject 2 having a 69.86\% probability of possessing a positive mindset towards life and a 30.32\% probability of possessing a negative mindset towards life. Subject 7 is an individual who holds a negative mindset towards life as his/her membership score is $(0.0377, 0.9623)$. In fact, the responses of subject 7 to all items (recorded in the $7$-th row of $R$) are $R(7,:)=(2,1,2,2,1,2,1,1,1,1,1,2,2,1,2,1,1,2,2,2,1,1,1,1,1,1,1,1,2,1,1,2,2,1,1,1,1,2,2,\\1,1,1,1,1,1,1,1,1,1,2)$, which are consistent with our analysis after checking the statement of each item in panel (b) of Fig.~\ref{TMASHeatmap}. For subject 21, he/she is positive because his/her membership score is $(0.9047,0.0953)$. In fact, the responses of subject 21 to all items are $R(21,:)=(1,2,1,1	,2,2,2,2,2,2,2,1,2,2,1,2,2,1,2,1,2,1,2,2,2,2,2,2,1,2,2,1,1,2,2,2,1,1,1,\\2,2,2,2,1,2,2,2,2,2,1)$, which also support our analysis. Meanwhile, it is easy to see that there are 44 opposite responses between subjects 7 and 21, and this supports our finding that subjects 7 and 21 have almost completely different mindsets towards life.

By analyzing panel (a) of Fig.~\ref{TMASHeatmap}, we find that the majority of individuals in the top 50 subjects have a more negative mindset towards life than a positive mindset. This finding also holds for all the 5401 individuals because $\sum_{i=1}^{5401}\hat{\Pi}(i,1)=1656.9$ and $\sum_{i=1}^{5401}\hat{\Pi}(i,2)=3744.1$. Furthermore, we say that an individual has a highly positive mindset towards life if $\hat{\Pi}(i,1)\geq0.9$ and a highly negative mindset towards life if $\hat{\Pi}(i,2)\geq0.9$. We are also interested in the number of individuals who have a highly positive (or negative) mindset towards life. For the TMAS data, we find that 51 subjects have a highly positive mindset towards life and 910 subjects have a highly negative mindset towards life, i.e., the number of individuals with a highly negative mindset towards life is much larger than that of individuals with a highly positive mindset towards life. The reason behind this phenomenon is profound and meaningful for social science and psychological science and it is of independent interest.
\begin{figure}
\centering
\subfigure[Heatmap of $\hat{\Pi}$]{\includegraphics[width=0.2235\textwidth]{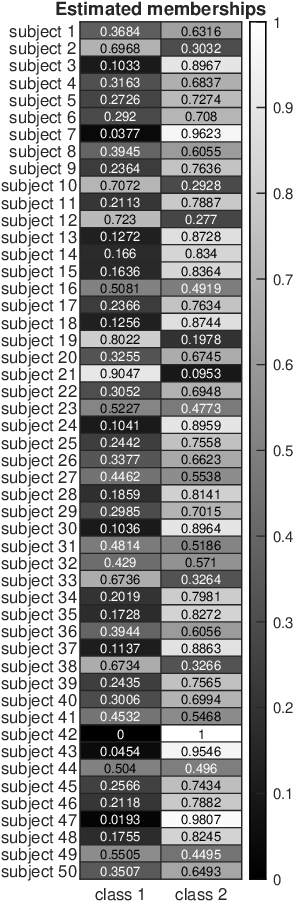}}
\subfigure[Heatmap of $\hat{\Theta}$]{\includegraphics[width=0.767\textwidth]{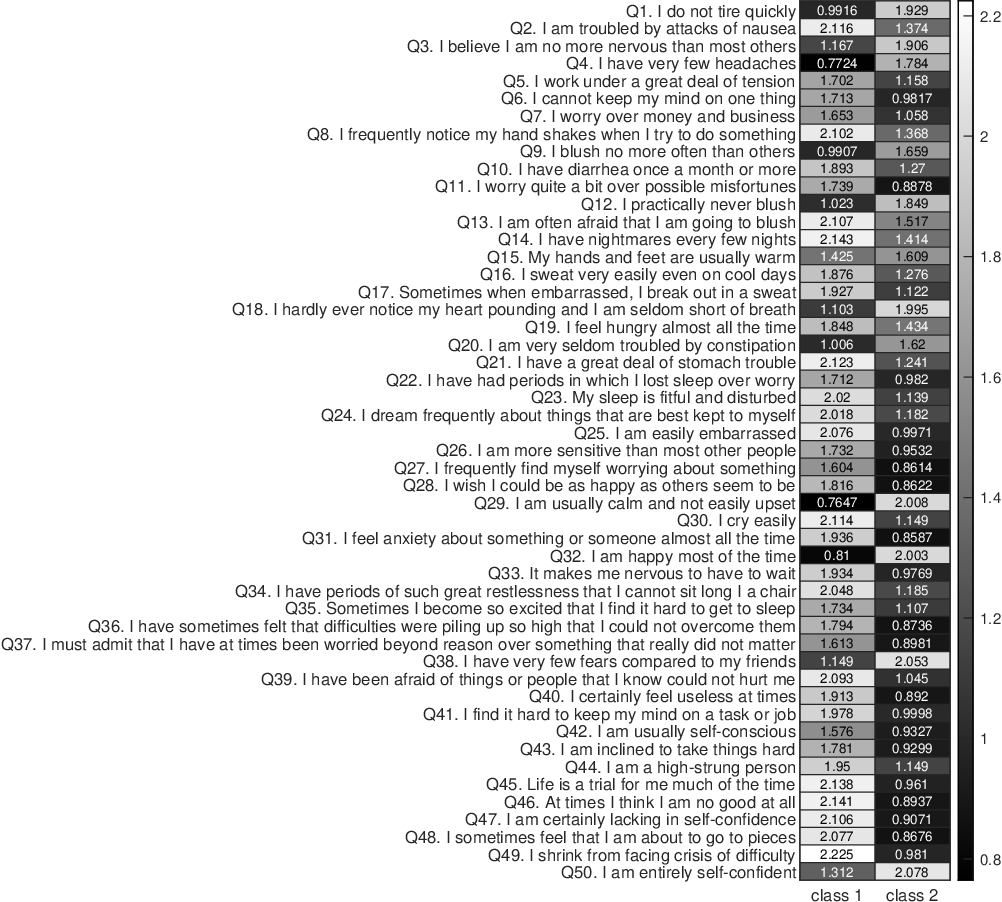}}
\caption{Numerical results for the TMAS data.}
\label{TMASHeatmap} %% label for entire figure
\end{figure}
\section{Conclusion}\label{sec8}
This paper presents the Weighted Grade of Membership (WGoM) model, a novel statistical model in the grade of membership analysis for categorical data with weighted responses. The proposed model allows each subject to belong to multiple latent classes and provides a generative framework for weighted response matrices based on arbitrary distribution. The Grade of Membership model, a well-known paradigm in the literature, is a special case of our WGoM framework. To the best of our knowledge, our WGoM is the first model that enables weighted responses in grade of membership analysis for categorical data. An efficient SVD-based algorithm is developed for estimating the mixed membership matrix and the item parameter matrix under the WGoM framework. We establish the rate of convergence of our method by considering the scaling parameter. It is shown that the proposed method enjoys consistent estimation under WGoM. We provide several instances to demonstrate the generality of our model and analyze the behaviors of our method when the observed weighted response matrices are generated from different distributions under WGoM. We also propose a method for estimating the number of latent classes in categorical data with weighted responses by maximizing the fuzzy weighted modularity. Extensive experiments are conducted to verify our theoretical findings and evaluate the effectiveness of our methods.

The contribution of this paper is fourfold: (1) We develop a novel WGoM model that surpasses the limitations of the Grade of Membership model, offering a more flexible framework for analyzing categorical data with weighted responses, thereby opening new opportunities for practical applications in various research areas; (2) We present an efficient SVD-based algorithm for estimating the parameters in WGoM; (3) We conduct a comprehensive evaluation of our method through theoretical analysis, simulation studies, and real-world applications, demonstrating its effectiveness both theoretically and empirically; and (4) We propose an efficient tool for determining $K$ in categorical data with weighted responses.

Future work may explore an extension of WGoM to handle a more complex scenario that includes additional covariates. It is worth noting that our approach for estimating $K$ involves maximizing fuzzy weighted modularity, which can be time-consuming for large-scale categorical data or when dealing with a large candidate value $K_{C}$. Developing a method with theoretical guarantees that can effectively determine $K$ for large-scale categorical data under WGoM is an interesting and challenging future direction. Furthermore, developing more efficient algorithms to enhance the computational efficiency of the proposed method would be a valuable contribution. Finally, the current model is static, assuming that the latent class memberships and item parameters do not change over time. Extending the model to a dynamic setting, where the parameters evolve over time, could provide insights into how latent structures change and evolve in real-world applications.

\bmhead{Acknowledgements} Qing's work was sponsored by the Scientific Research Foundation of Chongqing University of Technology (Grant No. 2024ZDR003), the Science and Technology Research Program of Chongqing Municipal Education Commission (Grant No. KJQN202401168), and the Natural Science Foundation of Chongqing, China (Grant No: CSTB2023NSCQ-LZX0048).
%\bmhead{Author Contributions} Huan Qing: Conceptualization, Methodology, Investigation, Software, Formal analysis, Data curation, Writing-original draft, Writing-reviewing \& editing, Funding acquisition.
\section*{Declarations}
\textbf{Conflict of interest} The author declares no conflict of interest.
\begin{appendices}
\section{Main notations}
We list main notations used in this paper in Table \ref{tab:notations}.
\begin{table}[htp!]
\centering
\caption{Table of the main notations.}
\label{tab:notations}
\begin{tabular}{|c|l|l|}
\hline
\textbf{Symbol} & \textbf{Description} \\
\hline
$N$ & Number of subjects \\
$J$ & Number of items \\
$\mathbb{R}$ & Set of real numbers \\
$R \in \mathbb{R}^{N \times J}$ & Observed weighted response matrix \\
$[m]$ & Set $\{1,~2,~\ldots,~m\}$ for any integer $m > 0$ \\
$I_{m \times m}$ & $m \times m$ identity matrix \\
$|\cdot|$ & Absolute value of a scalar \\
$\|\cdot\|_q$ & $l_q$-norm of a vector \\
$\|\cdot\|$ & Spectral norm of a matrix \\
$\|\cdot\|_F$ & Frobenius norm of a matrix \\
$\|\cdot\|_{2 \rightarrow \infty}$ & Maximum $\ell_2$ norm among all rows of a matrix\\
$\lambda_k(\cdot)$ & $k$-th largest eigenvalue (in magnitude) of a matrix \\
$\sigma_k(\cdot)$ & $k$-th largest singular value of a matrix \\
$\kappa(\cdot)$ & Condition number of a matrix \\
$\mathrm{rank}(\cdot)$ & Rank of a matrix \\
$e_i$ & Standard basis vector with $i$-th entry being 1 \\
$\text{max}(0, \cdot)$ & Nonnegative part of a matrix \\
$\mathbb{E}(\cdot)$ & Expectation of a random variable \\
$\mathbb{P}(\cdot)$ & Probability of a random variable \\
$a = O(b)$ & $a$ and $b$ are of the same order for any $a$ and $b$ \\
$K$ & Number of latent classes \\
$\Pi \in [0, 1]^{N \times K}$ & Mixed membership matrix \\
$\Theta \in \mathbb{R}^{J \times K}$ & Item parameter matrix \\
$\mathcal{I}$ & Index set of pure subjects \\
$R_0 \in \mathbb{R}^{N \times J}$ & $\Pi \Theta'$ \\
$\mathcal{F}$ & Arbitrary distribution \\
$U, \Sigma, V$ & Compact SVD of $R_0$, i.e., $R_0 = U \Sigma V'$ \\
$X \in \mathbb{R}^{K \times K}$ & Corner matrix $X = U(\mathcal{I}, :)$ \\
$Z \in [0, 1]^{N \times K}$ & $U X^{-1}$ \\
$\hat{U}, \hat{\Sigma}, \hat{V}$ & Top $K$ SVD of $R$ \\
$\hat{R}$ & Top $K$ SVD approximation of $R$, i.e., $\hat{R} = \hat{U} \hat{\Sigma} \hat{V}'$ \\
$\hat{\mathcal{I}}$ & Index set of estimated pure subjects \\
$\hat{Z} \in [0, +\infty)^{N \times K}$ & $\text{max}(0, \hat{U} \hat{U}^{-1}(\hat{\mathcal{I}}, :))$ \\
$\hat{\Pi} \in [0, 1]^{N \times K}$ & Estimated mixed membership matrix \\
$\hat{\Theta} \in \mathbb{R}^{J \times K}$ & Estimated item parameter matrix \\
$\rho$ & Scaling parameter $\rho = \text{max}_{j \in [J], k \in [K]} |\Theta(j, k)|$ \\
$B \in \mathbb{R}^{J \times K}$ & Matrix such that $\Theta = \rho B$ \\
$\tau$ & $\text{max}_{i \in [N], j \in [J]} |R(i, j) - R_0(i, j)|$ \\
$\text{Var}(R(i, j))$ & Variance of $R(i, j)$ \\
$\gamma$ & $\frac{\text{max}_{i \in [N], j \in [J]} \text{Var}(R(i, j))}{\rho}$ \\
$\mathcal{P} \in \mathbb{R}^{K \times K}$ & Permutation matrix \\
$p \in [0, 1]$ & Sparsity parameter \\
$\mathcal{B} \in \{0, 1\}^{N \times J}$ & Binary matrix \\
$A \in \mathbb{R}^{N \times N}$ & $R R'$ \\
$\mathcal{M}$ & Specific method \\
$A_{+}$ & $\text{max}(0, A)$ \\
$A_{-}$ & $\text{max}(0, -A)$ \\
$d_{+}$ & Vector with $d_{+}(i) = \sum_{j \in [N]} A_{+}(i, j)$ for $i \in [N]$ \\
$d_{-}$ & Vector with $d_{-}(i) = \sum_{j \in [N]} A_{-}(i, j)$ for $i \in [N]$ \\
$m_{+}$ & $\sum_{i \in [N]} d_{+}(i) / 2$ \\
$m_{-}$ & $\sum_{i \in [N]} d_{-}(i) / 2$ \\
$Q_{+}$ & Fuzzy modularity computed from $A_{+}$ \\
$Q_{-}$ & Fuzzy modularity computed from $A_{-}$ \\
$Q_{\mathcal{M}}(k)$ & Fuzzy weighted modularity computed by method $\mathcal{M}$ with $k$ latent classes \\
\hline
\end{tabular}
\end{table}

\section{Proofs under WGoM}
\subsection{Proof of Theorem \ref{ExistDisF}}
\begin{proof}
First, we consider the case when $Q=2$. For this case, we let $\mathbb{P}(R(i,j)=a_{1})=p_{1}$ and $\mathbb{P}(R(i,j)=a_{2})=p_{2}$. $p_{1}$ and $p_{2}$ are two probabilities which should satisfy the following conditions:
\begin{align}\label{Q2p1p2}
p_{1}+p_{2}=1,~0\leq p_{1}\leq 1,~0\leq p_{2}\leq1.
\end{align}
To make Equation (\ref{RFR0}) hold (i.e., $\mathbb{E}(R(i,j))=R_{0}(i,j)$), we need
\begin{align}\label{Q2R0}
  a_{1}p_{1}+a_{2}p_{2}=R_{0}(i,j).
\end{align}
Combining Equation (\ref{Q2p1p2}) and Equation (\ref{Q2R0}) gives
\begin{align*}
p_{1}=\frac{a_{2}-R_{0}(i,j)}{a_{2}-a_{1}},~p_{2}=\frac{R_{0}(i,j)-a_{1}}{a_{2}-a_{1}},
\end{align*}
where $R_{0}(i,j)\in[a_{1},a_{2}]$. Thus, for the case $Q=2$, there exists only one discrete distribution $\mathcal{F}$ satisfying Equation (\ref{RFR0}), and the exact form of $\mathcal{F}$ is
\begin{align*}
\mathbb{P}(R(i,j)=a_{1})=\frac{a_{2}-R_{0}(i,j)}{a_{2}-a_{1}},~\mathbb{P}(R(i,j)=a_{2})=\frac{R_{0}(i,j)-a_{1}}{a_{2}-a_{1}},
\end{align*}
where $R_{0}(i,j)\in[a_{1},a_{2}]$.

Next, we consider the case when $Q\geq3$. For this case, we let $\mathbb{P}(R(i,j)=a_{q})=p_{q}$ for $q\in[Q]$, where $p_{q}$ is a probability. Since the summation of all probabilities should be 1, $\{p_{q}\}^{Q}_{q=1}$ should satisfy
\begin{align}\label{Qp}
\sum_{q=1}^{Q}p_{q}=1,~0\leq p_{q}\leq 1 \mathrm{~for~}q\in[Q].
\end{align}
To make Equation (\ref{RFR0}) hold, $\{p_{q}\}^{Q}_{q=1}$ should satisfy
\begin{align}\label{QR0}
\sum_{q=1}^{Q}a_{q}p_{q}=R_{0}(i,j).
\end{align}
By analyzing Equation (\ref{Qp}) and Equation (\ref{QR0}), we find that the $Q$ unknown elements $\{p_{q}\}^{Q}_{q=1}$ should satisfy 2 equalities and $Q$ inequalities. Therefore, there must exist $\{p_{q}\}^{Q}_{q=1}$ satisfying Equation (\ref{Qp}) and Equation (\ref{QR0}).

Now, we construct $Q$ solutions of $\{p_{q}\}^{Q}_{q=1}$ that satisfy Equation (\ref{Qp}) and Equation (\ref{QR0}). Let $p_{2}=p_{3}=\ldots=p_{Q}\equiv y$. Equation (\ref{Qp}) gives
\begin{align}\label{pQp1y}
p_{1}+(Q-1)y=1,~0\leq p_{1}\leq1,~0\leq y\leq1.
\end{align}
Equation (\ref{QR0}) gives
\begin{align}\label{pQp2y}
a_{1}p_{1}+(a_{2}+a_{3}+\ldots+a_{Q})y=R_{0}(i,j).
\end{align}
Combining Equation (\ref{pQp1y}) and Equation (\ref{pQp2y}) gives
\begin{align*}
 p_{1}=\frac{a_{2}+a_{3}+\ldots+a_{Q}-(Q-1)R_{0}(i,j)}{\sum_{q=1}^{Q}a_{q}-Qa_{1}},~p_{2}=p_{3}=\ldots=p_{Q}=y=\frac{R_{0}(i,j)-a_{1}}{\sum_{q=1}^{Q}a_{q}-Qa_{1}},
\end{align*}
where $a_{1}\leq R_{0}(i,j)\leq\frac{a_{2}+a_{3}+\ldots+a_{Q}}{Q-1}$.

Therefore, the exact form of a $\mathcal{F}$ satisfying Equation (\ref{RFR0}) is
\begin{align*}
&\mathbb{P}(R(i,j)=a_{1})=\frac{a_{2}+a_{3}+\ldots+a_{Q}-(Q-1)R_{0}(i,j)}{\sum_{q=1}^{Q}a_{q}-Qa_{1}},\\
&\mathbb{P}(R(i,j)=a_{q})=\frac{R_{0}(i,j)-a_{1}}{\sum_{q=1}^{Q}a_{q}-Qa_{1}} \mathrm{~for~}q=2,3,\ldots,Q,
\end{align*}
where $R_{0}(i,j)\in[a_{1},\frac{a_{2}+a_{3}+\ldots+a_{Q}}{Q-1}]$. Since $a_{1}<~a_{2}<~\ldots<~a_{Q}$, we have $a_{1}<\frac{\sum_{q=2}^{Q}a_{q}}{Q-1}$, which implies that $[a_{1}, \frac{a_{2}+a_{3}+\ldots+a_{Q}}{Q-1}]$, the range of $R_{0}(i,j)$, is not empty. Sure, there are $Q$ choices such that $p_{q}$ is fixed while $p_{1}=p_{2}=\ldots=p_{q-1}=p_{q+1}=\ldots=p_{Q}$ for $q\in[Q]$. Therefore, there are at least $Q$ distinct discrete distributions $\mathcal{F}$ satisfying Equation (\ref{RFR0}).
\end{proof}
\subsection{Proof of Proposition \ref{idWGoM}}
\begin{proof}
Point (a) of Theorem 2 in \citep{chen2023spectral} always holds provided that each latent class has at least one pure subject and $\Theta$ is a rank-$K$ matrix, thus it guarantees WGoM's identifiability since WGoM requires $\Pi$ to satisfy Condition \ref{pure} and $\Theta$'s rank to be $K$. Note that as stated by point (b) of Theorem 2 in \citep{chen2023spectral}, WGoM is also identifiable if $\Theta$'s rank is $(K-1)$ and $\Theta$ satisfies some extra conditions in point (b) of Theorem 2 in \citep{chen2023spectral}. However, we only consider the case $\mathrm{rank}(\Theta)=K$ in this paper mainly for the convenience of our further theoretical analysis.
\end{proof}
\subsection{Proof of Lemma \ref{UVWGoM}}
\begin{proof}
Since $R_{0}=\Pi\Theta'=U\Sigma V'$, we have $U=R_{0}V\Sigma^{-1}=\Pi\Theta'V\Sigma^{-1}=\Pi X$, where we set $X$ as $\Theta'V\Sigma^{-1}$. Since $U=\Pi X$ and we have assumed that $\Pi(\mathcal{I},:)=I_{K\times K}$, we have $U(\mathcal{I},:)=\Pi(\mathcal{I},:)X=X$, i.e., $X=U(\mathcal{I},:)$.
\end{proof}
\subsection{Basic properties of $R_{0}$}
The following lemmas are useful for building our main theoretical result.
\begin{lem}\label{P2}
Under $WGoM(\Pi,\Theta,\mathcal{F})$, for $i\in[N]$, $j\in[J]$, we have
\begin{align*}
\sqrt{\frac{1}{K\lambda_{1}(\Pi'\Pi)}}\leq\|U(i,:)\|_{F}\leq \sqrt{\frac{1}{\lambda_{K}(\Pi'\Pi)}} \mathrm{~and~} 0\leq\|V(j,:)\|_{F}\leq\frac{K^{1.5}\kappa(\Pi)}{\sqrt{\lambda_{K}(B'B)}}.
\end{align*}
\end{lem}
\begin{proof}
For $U$, since $U=\Pi X$ by Lemma \ref{UVWGoM}, we have
\begin{align*}
&\mathrm{min}_{i\in[N]}\|e'_{i}U\|^{2}_{F}=\mathrm{min}_{i\in[N]}e'_{i}UU'e_{i}=\mathrm{min}_{i\in[N]}\Pi(i,:)XX'\Pi'(i,:)\\
&=\mathrm{min}_{i\in[N]}\|\Pi(i,:)\|^{2}_{F}\frac{\Pi(i,:)}{\|\Pi(i,:)\|_{F}}XX'\frac{\Pi'(i,:)}{\|\Pi(i,:)\|_{F}}\geq \mathrm{min}_{i\in[N]}\|\Pi(i,:)\|^{2}_{F}\mathrm{min}_{\|y\|_{F}=1}y'XX'y\\
&=\mathrm{min}_{i\in[N]}\|\Pi(i,:)\|^{2}_{F}\lambda_{K}(XX')\overset{\mathrm{By~Lemma~}\ref{P3}}{=}\frac{\mathrm{min}_{i\in[N]}\|\Pi(i,:)\|^{2}_{F}}{\lambda_{1}(\Pi'\Pi)}\geq \frac{1}{K\lambda_{1}(\Pi'\Pi)},
\end{align*}
where the last inequality holds since $\|\Pi(i,:)\|_{F}\geq \|\Pi(i,:)\|_{1}/\sqrt{K}=1/\sqrt{K}$ for $i\in[N]$ and $e_{i}$ is a $N\times 1$ vector with $i$-th entry being 1 while the other entries being zero. Meanwhile,
\begin{align*}
\mathrm{max}_{i\in[N]}\|e'_{i}U\|^{2}_{F}&=\mathrm{max}_{i\in[N]}\|\Pi(i,:)\|^{2}_{F}\frac{\Pi(i,:)}{\|\Pi(i,:)\|_{F}}XX'\frac{\Pi'(i,:)}{\|\Pi(i,:)\|_{F}}\\
&\leq\mathrm{max}_{i\in[N]}\|\Pi(i,:)\|^{2}_{F}\mathrm{max}_{\|y\|_{F}=1}y'XX'y\\
&=\mathrm{max}_{i\in[N]}\|\Pi(i,:)\|^{2}_{F}\lambda_{1}(XX')\\
&\overset{\mathrm{By~Lemma~}\ref{P3}}{=}\frac{\mathrm{max}_{i\in[N]}\|\Pi(i,:)\|^{2}_{F}}{\lambda_{K}(\Pi'\Pi)}\\
&\leq \frac{1}{\lambda_{K}(\Pi'\Pi)}.
\end{align*}

For $V$, recall that $R_{0}=\Pi\Theta'=U\Sigma V', \Theta=\rho B, \mathrm{max}_{j\in[J],k\in[K]}|B(j,k)|=1$, and $U'U=I_{K\times K}$, we have $V=\Theta\Pi'U\Sigma^{-1}$ which gives that
\begin{align*}
\mathrm{max}_{j\in[J]}\|V(j,:)\|_{F}&=\mathrm{max}_{j\in[J]}\rho\|B(i,:)\Pi'U\Sigma^{-1}\|_{F}\leq \mathrm{max}_{j\in[J]}\rho\|B(i,:)\Pi'U\|_{F}\|\Sigma^{-1}\|_{F}\\
&=\mathrm{max}_{j\in[J]}\rho\|B(i,:)\Pi'\|_{F}\|\Sigma^{-1}\|_{F}\leq\mathrm{max}_{j\in[J]}\rho\|B(i,:)\Pi'\|_{F}\frac{\sqrt{K}}{\sigma_{K}(R_{0})}\\
&\overset{\mathrm{By~Lemma~}\ref{P4}}{\leq}\mathrm{max}_{j\in[J]}\|B(i,:)\Pi'\|_{F}\frac{\sqrt{K}}{\sigma_{K}(\Pi)\sigma_{K}(B)}\\
&\leq\mathrm{max}_{j\in[J]}\|B(i,:)\|_{F}\|\Pi\|_{F}\frac{\sqrt{K}}{\sigma_{K}(\Pi)\sigma_{K}(B)}\\
&\leq\frac{K\|\Pi\|_{F}}{\sigma_{K}(\Pi)\sigma_{K}(B)}\leq\frac{K^{1.5}\|\Pi\|}{\sigma_{K}(\Pi)\sigma_{K}(B)}=\frac{K^{1.5}\kappa(\Pi)}{\sigma_{K}(B)},
\end{align*}
where we have used the fact that $\|Y\|_{F}\leq\sqrt{\mathrm{rank}(Y)}\|Y\|$ for any matrix $Y$. The lower bound of $\|V(j,:)\|_{F}$ is 0 because $\|B(i,:)\|_{F}$ can be set as 0.
\end{proof}
\begin{lem}\label{P3}
Under $WGoM(\Pi, \Theta,\mathcal{F})$, we have
\begin{align*}	
\lambda_{1}(XX')=\frac{1}{\lambda_{K}(\Pi'\Pi)}\mathrm{~and~}\lambda_{K}(XX')=\frac{1}{\lambda_{1}(\Pi'\Pi)}.
\end{align*}
\end{lem}
\begin{proof}
Recall that $U=\Pi X$ by Lemma \ref{UVWGoM} and $U'U=I_{K\times K}$, we have $I_{K\times K}=X'\Pi'\Pi X$. Because $X$ is a nonsingular matrix, we have $\Pi'\Pi=(XX')^{-1}$, which gives
\begin{align*}
\lambda_{1}(XX')=\frac{1}{\lambda_{K}(\Pi'\Pi)},\lambda_{K}(XX')=\frac{1}{\lambda_{1}(\Pi'\Pi)}.
\end{align*}
\end{proof}
\begin{lem}\label{P4}
Under $WGoM(\Pi,\Theta,\mathcal{F})$, we have
\begin{align*} \sigma_{K}(R_{0})\geq\rho\sigma_{K}(\Pi)\sigma_{K}(B)\mathrm{~and~} \sigma_{1}(R_{0})\leq\rho\sigma_{1}(\Pi)\sigma_{1}(B).
\end{align*}
\end{lem}
\begin{proof}
For $\sigma_{K}(R_{0})$, we have
\begin{align*}
\sigma^{2}_{K}(R_{0})&=\lambda_{K}(R_{0}R'_{0})=\lambda_{K}(\Pi \Theta'\Theta\Pi')=\lambda_{K}(\Pi'\Pi\Theta'\Theta)\geq \lambda_{K}(\Pi'\Pi)\lambda_{K}(\Theta'\Theta)\\
&=\lambda_{K}(\Pi'\Pi)\lambda_{K}(\rho^{2}B'B)=\rho^{2}\sigma^{2}_{K}(\Pi)\sigma^{2}_{K}(B),
\end{align*}
where we have frequently used the fact the nonzero eigenvalues of $Y_{1}Y_{2}$ are equal to that of $Y_{2}Y_{1}$ for any two matrices $Y_{1}$ and $Y_{2}$.

For $\sigma_{1}(R_{0})$, we have
\begin{align*}
\sigma_{1}(R_{0})=\|R_{0}\|=\rho\|\Pi B'\|\leq \rho \|\Pi\|\|B\|=\rho\sigma_{1}(\Pi)\sigma_{1}(B).
\end{align*}
\end{proof}
\subsection{Proof of Theorem \ref{mainWGoM}}
\begin{proof}
Similar to Theorem 3.1 \citep{mao2021estimating}, we obtain the row-wise eigenspace error first. Define  $H_{\hat{U}}$ as $H_{\hat{U}}=\hat{U}'U$. Let $H_{\hat{U}}=U_{H_{\hat{U}}}\Sigma_{H_{\hat{U}}}V'_{H_{\hat{U}}}$ be the top $K$ SVD of $H_{\hat{U}}$. Set $\mathrm{sgn}(H_{\hat{U}})=U_{H_{\hat{U}}}V'_{H_{\hat{U}}}$. Define $H_{\hat{V}}$ as $H_{\hat{V}}=\hat{V}'V$, let $H_{\hat{V}}=U_{H_{\hat{V}}}\Sigma_{H_{\hat{V}}}V'_{H_{\hat{V}}}$ be the top K SVD of $H_{\hat{V}}$, and set $\mathrm{sgn}(H_{\hat{V}})=U_{H_{\hat{V}}}V'_{H_{\hat{V}}}$. For $i\in[N]$, $j\in[J]$, the following statements are true:
\begin{itemize}
  \item $\mathbb{E}(R(i,j)-R_{0}(i,j))=0$ under WGoM.
  \item $\mathbb{E}((R(i,j)-R_{0}(i,j))^{2})=\mathrm{Var}(R(i,j))\leq \rho\gamma$.
  \item Set $\mu=\mathrm{max}(\frac{N\|U\|^{2}_{2\rightarrow\infty}}{K},\frac{J\|V\|^{2}_{2\rightarrow\infty}}{K})$ as the incoherence parameter introduced in Definition 3.1. of \citep{chen2021spectral}. Based on Lemma \ref{P2} and Condition \ref{c1}, we have $\frac{N\|U\|^{2}_{2\rightarrow\infty}}{K}\leq\frac{N}{K\lambda_{K}(\Pi'\Pi)}=O(1)$ and $\frac{J\|V\|^{2}_{2\rightarrow\infty}}{K}\leq\frac{JK^{3}\kappa^{2}(\Pi)}{K\lambda_{K}(B'B)}=O(1)$. Therefore, we have $\mu=O(1)$.
  \item Let $c_{b}=\frac{\tau}{\sqrt{\rho\gamma \mathrm{max}(N,J)/(\mu \mathrm{log}(N+J))}}$. By $\mu=O(1)$ and Assumption \ref{a1}, we have $c_{b}\leq O(1)$.
  \item By Lemma \ref{P4} and Condition \ref{c1}, we have $\kappa(R_{0})=\frac{\sigma_{1}(R_{0})}{\sigma_{K}(R_{0})}\leq\kappa(\Pi)\kappa(B)=O(1)$.
\end{itemize}

By the first four statements, we see that all conditions of Theorem 4.4 \citep{chen2021spectral} are satisfied. Then, according to Theorem 4.4. \citep{chen2021spectral}, when $\sigma_{K}(R_{0})\gg \sqrt{\rho\gamma(N+J)\mathrm{log}(N+J)}$, with probability at least $1-O(\frac{1}{(N+J)^{5}})$, we have
\begin{align*}
&\mathrm{max}(\|\hat{U}\mathrm{sgn}(H_{\hat{U}})-U\|_{2\rightarrow\infty},\|\hat{V}\mathrm{sgn}(H_{\hat{V}})-V\|_{2\rightarrow\infty})\\
&\leq C\frac{\sqrt{\rho\gamma K\mathrm{log}(N+J)})+\sqrt{\rho\gamma K}\kappa(R_{0})\sqrt{\frac{\mathrm{max}(N,J)}{\mathrm{min}(N,J)}\mu}}{\sigma_{K}(R_{0})}\\
&=O(\frac{\sqrt{\rho\gamma K\mathrm{log}(N+J)}}{\sigma_{K}(R_{0})})\leq O(\frac{\sqrt{\rho\gamma K\mathrm{log}(N+J)}}{\rho\sigma_{K}(\Pi)\sigma_{K}(B)}),
\end{align*}
where the last inequality is held by Lemma \ref{P4}.
Set $\varpi=\mathrm{max}(\|\hat{U}\hat{U}'-UU'\|_{2\rightarrow\infty},~\|\hat{V}\hat{V}'-VV'\|_{2\rightarrow\infty})$ as the row-wise eigenspace error. Because $\|\hat{U}\hat{U}'-UU'\|_{2\rightarrow\infty}\leq2\|U-\hat{U}\mathrm{sgn}(H_{\hat{U}})\|_{2\rightarrow\infty}$ and $\|\hat{V}\hat{V}'-VV'\|_{2\rightarrow\infty}\leq2\|V-\hat{V}\mathrm{sgn}(H_{\hat{V}})\|_{2\rightarrow\infty}$, we have
\begin{align}\label{rowwise}	
\varpi=O(\frac{\sqrt{\rho\gamma K\mathrm{log}(N+J)}}{\rho\sigma_{K}(\Pi)\sigma_{K}(B)}).
\end{align}

The per-subject error rates of the membership matrix $\Pi$ are the same as that of mixed membership community detection. By the proofs of Theorem 3.2 \citep{mao2021estimating} and Theorem 1 \citep{qing2023bipartite}, we know that with probability at least $1-O(\frac{1}{(N+J)^{5}})$, there exists a $K\times K$ permutation matrix $\mathcal{P}$ such that for $i\in[N]$,
\begin{align}\label{generalPi}	\|e'_{i}(\hat{\Pi}-\Pi\mathcal{P})\|_{1}=O(\varpi\kappa(\Pi'\Pi)K\sqrt{\lambda_{1}(\Pi'\Pi)}).
\end{align}

By Equations (\ref{rowwise})-(\ref{generalPi}) and Condition \ref{c1}, we have
\begin{align}\label{boundPihat}	\|e'_{i}(\hat{\Pi}-\Pi\mathcal{P})\|_{1}=O(\varpi\kappa(\Pi'\Pi)K\sqrt{\lambda_{1}(\Pi'\Pi)})=O(\sqrt{\frac{\gamma\mathrm{log}(N+J)}{\rho J}}).
\end{align}

If we do not consider Condition \ref{c1}, though we can always obtain the bound of $\|e'_{i}(\hat{\Pi}-\Pi\mathcal{P})\|_{1}$ using Equation (\ref{generalPi}), we can not get the transparent form provided in Equation (\ref{boundPihat}) since we fail to simplify $\varpi$ without Condition \ref{c1} when we use Theorem 4.4 of \citep{chen2021spectral} to obtain $\varpi$'s upper bound.

Next, we consider the theoretical upper bound of $\|\hat{\Theta}-\Theta\mathcal{P}\|$. First, we have
\begin{align}
\|\hat{\Theta}-\Theta\mathcal{P}\|&=\|\hat{R}'\hat{\Pi}(\hat{\Pi}'\hat{\Pi})^{-1}-R'_{0}\Pi(\Pi'\Pi)^{-1}\mathcal{P}\|\notag\\
&=\|(\hat{R}'-R'_{0})\hat{\Pi}(\hat{\Pi}'\hat{\Pi})^{-1}+R'_{0}(\hat{\Pi}(\hat{\Pi}'\hat{\Pi})^{-1}-\Pi(\Pi'\Pi)^{-1}\mathcal{P})\|\notag\\
&\leq\|(\hat{R}'-R'_{0})\hat{\Pi}(\hat{\Pi}'\hat{\Pi})^{-1}\|+\|R'_{0}(\hat{\Pi}(\hat{\Pi}'\hat{\Pi})^{-1}-\Pi(\Pi'\Pi)^{-1}\mathcal{P})\|\notag\\
&\leq\|\hat{R}-R_{0}\|\|\hat{\Pi}(\hat{\Pi}'\hat{\Pi})^{-1}\|+\|R_{0}\|\hat{\Pi}(\hat{\Pi}'\hat{\Pi})^{-1}-\Pi(\Pi'\Pi)^{-1}\mathcal{P}\|\notag\\
&\leq2\|R-R_{0}\|\|\hat{\Pi}(\hat{\Pi}'\hat{\Pi})^{-1}\|+\|R_{0}\|\hat{\Pi}(\hat{\Pi}'\hat{\Pi})^{-1}-\Pi(\Pi'\Pi)^{-1}\mathcal{P}\|\notag\\
&=2\|R-R_{0}\|\|\hat{\Pi}(\hat{\Pi}'\hat{\Pi})^{-1}\|+\|\rho \Pi B'\|\hat{\Pi}(\hat{\Pi}'\hat{\Pi})^{-1}-\Pi(\Pi'\Pi)^{-1}\mathcal{P}\|\notag\\
&\leq2\|R-R_{0}\|\|\hat{\Pi}(\hat{\Pi}'\hat{\Pi})^{-1}\|+\rho \|\Pi\|\|B\|\|\hat{\Pi}(\hat{\Pi}'\hat{\Pi})^{-1}-\Pi(\Pi'\Pi)^{-1}\mathcal{P}\|\notag\\
&=\frac{2\|R-R_{0}\|}{\sigma_{K}(\hat{\Pi})}+\rho\sigma_{1}(\Pi)\sigma_{1}(B)\|\hat{\Pi}(\hat{\Pi}'\hat{\Pi})^{-1}-\Pi(\Pi'\Pi)^{-1}\mathcal{P}\|\label{boundThetahat1}.
\end{align}

The upper bound of $\|\hat{\Theta}-\Theta\mathcal{P}\|$ can be obtained immediately as long as we can obtain the lower bound of $\sigma_{K}(\hat{\Pi})$, the upper bounds of $\|R-R_{0}\|$ and $\|\hat{\Pi}(\hat{\Pi}'\hat{\Pi})^{-1}-\Pi(\Pi'\Pi)^{-1}\mathcal{P}\|$. We bound the three terms subsequently.

For $\sigma_{K}(\hat{\Pi})$, Weyl's inequality for singular values \citep{weyl1912das} gives us:
\begin{align}\label{BoundUsed}
|\sigma_{K}(\hat{\Pi})-\sigma_{K}(\Pi)|&\leq\|\hat{\Pi}-\Pi\mathcal{P}\|\notag\\
&\leq \sqrt{N}\mathrm{~max}_{i\in[N]}\|e'_{i}(\hat{\Pi}-\Pi\mathcal{P})\|_{1}\notag\\
&\overset{\mathrm{By~Equation~}(\ref{boundPihat})}{=}O(\sqrt{\frac{\gamma N\mathrm{log}(N+J)}{\rho J}}).
\end{align}

Then we have $\sigma_{K}(\hat{\Pi})\geq\sigma_{K}(\Pi)-O(\sqrt{\frac{\gamma N\mathrm{log}(N+J)}{\rho J}})$. Note that $\sigma_{K}(\Pi)=O(\sqrt{\frac{N}{K}})$ by Condition \ref{c1} and $\frac{\gamma\mathrm{log}(N+J)}{\rho J}\ll1$ should hold to make the error rate in Equation (\ref{boundPihat}) close to zero, we have $\sigma_{K}(\Pi)-O(\sqrt{N}\sqrt{\frac{\gamma\mathrm{log}(N+J)}{\rho J}})=O(\sqrt{N/K})$, i.e.,
\begin{align}\label{boundsigmaK}
\sigma_{K}(\hat{\Pi})\geq O(\sqrt{N/K}).
\end{align}

For the term $\|R-R_{0}\|$, when Assumption \ref{a1} holds, Lemma 2 of \citep{qing2023community} says that, with probability at least $1-o((N+J)^{-5})$, we have
\begin{align}\label{boundRR0}
\|R-R_{0}\|=O(\sqrt{\gamma\rho\mathrm{max}(N,J)\mathrm{log}(N+J)}).
\end{align}

For $\|\hat{\Pi}(\hat{\Pi}'\hat{\Pi})^{-1}-\Pi(\Pi'\Pi)^{-1}\mathcal{P}\|$, by Condition \ref{c1},  Equation (\ref{BoundUsed}), and Equation (\ref{boundsigmaK}), we have
\begin{align}
\|\hat{\Pi}(\hat{\Pi}'\hat{\Pi})^{-1}-\Pi(\Pi'\Pi)^{-1}\mathcal{P}\|&=\|(\hat{\Pi}-\Pi\mathcal{P})(\hat{\Pi}'\hat{\Pi})^{-1}+\Pi(\mathcal{P}(\hat{\Pi}'\hat{\Pi})^{-1}-(\Pi'\Pi)^{-1}\mathcal{P})\|\notag\\
&\leq\|\hat{\Pi}-\Pi\mathcal{P}\|\|(\hat{\Pi}'\hat{\Pi})^{-1}\|+\|\Pi\|\|\mathcal{P}(\hat{\Pi}'\hat{\Pi})^{-1}-(\Pi'\Pi)^{-1}\mathcal{P}\|\notag\\
&=\frac{\|\hat{\Pi}-\Pi\mathcal{P}\|}{\sigma^{2}_{K}(\hat{\Pi})}+\sigma_{1}(\Pi)\|\mathcal{P}(\hat{\Pi}'\hat{\Pi})^{-1}-(\Pi'\Pi)^{-1}\mathcal{P}\|\notag\\
&\leq\frac{\|\hat{\Pi}-\Pi\mathcal{P}\|_{F}}{\sigma^{2}_{K}(\hat{\Pi})}+\sigma_{1}(\Pi)\|\mathcal{P}(\hat{\Pi}'\hat{\Pi})^{-1}-(\Pi'\Pi)^{-1}\mathcal{P}\|\notag\\
&\leq\frac{\|\hat{\Pi}-\Pi\mathcal{P}\|_{F}}{\sigma^{2}_{K}(\hat{\Pi})}+\sigma_{1}(\Pi)(\|(\hat{\Pi}'\hat{\Pi})^{-1}\|+\|(\Pi'\Pi)^{-1}\|)\notag\\
&=\frac{\|\hat{\Pi}-\Pi\mathcal{P}\|_{F}}{\sigma^{2}_{K}(\hat{\Pi})}+\sigma_{1}(\Pi)(\frac{1}{\sigma^{2}_{K}(\hat{\Pi})}+\frac{1}{\sigma^{2}_{K}(\Pi)})\notag\\
&=O(\sqrt{\frac{\gamma\mathrm{log}(N+J)}{\rho NJ}})+O(\sqrt{\frac{1}{N}})=O(\sqrt{\frac{1}{N}}).\label{BoundInvPi}
\end{align}

By Condition \ref{c1}, Equations (\ref{boundThetahat1}), (\ref{boundsigmaK}), (\ref{boundRR0}), and (\ref{BoundInvPi}), we have
\begin{align*}
\|\hat{\Theta}-\Theta\mathcal{P}\|&=O(\sqrt{\frac{\gamma\rho\mathrm{max}(N,J)\mathrm{log}(N+J)}{N}})+O(\rho\sqrt{J}).
\end{align*}

Since $\|\hat{\Theta}-\Theta\mathcal{P}\|_{F}\leq\sqrt{K}\|\hat{\Theta}-\Theta\mathcal{P}\|$, by Condition \ref{c1}, we have
\begin{align*}
\|\hat{\Theta}-\Theta\mathcal{P}\|_{F}=O(\sqrt{\frac{\gamma \rho\mathrm{max}(N,J)\mathrm{log}(N+J)}{N}})+O(\rho\sqrt{J}).
\end{align*}
Since $\|\Theta\|_{F}\geq\|\Theta\|=\rho\sigma_{1}(B)=O(\rho\sqrt{\frac{J}{K}})=O(\rho\sqrt{J})$ by Condition \ref{c1}, we have
\begin{align*}
\frac{\|\hat{\Theta}-\Theta\mathcal{P}\|_{F}}{\|\Theta\|_{F}}=O(\sqrt{\frac{\gamma\mathrm{max}(N,J)\mathrm{log}(N+J)}{\rho NJ}}).
\end{align*}
\begin{rem}
By Lemma \ref{P4} and Condition \ref{c1}, we have $\sigma_{K}(R_{0})\geq O(\rho\sqrt{NJ})$. Since $\rho\sqrt{NJ}\gg\sqrt{\rho\gamma(N+J)\mathrm{log}(N+J)}=O(\sqrt{\rho\gamma N\mathrm{log}(N+J)})$ by Condition \ref{c1} $\Leftrightarrow\rho J\gg\gamma\mathrm{log}(N+J)$, we see that the inequality $\sigma_{K}(R_{0})\gg \sqrt{\rho\gamma(N+J)\mathrm{log}(N+J)}$ (this inequality is required when we use Theorem 4.4. \citep{chen2021spectral}) holds naturally as long as the bound in Equation (\ref{boundPihat}) is sufficiently small, i.e., $\rho J\gg\gamma\mathrm{log}(N+J)$.
\end{rem}
\end{proof}
\section{Extra instances}\label{ExtraInstances}
\begin{Ex}\label{Poisson}
Let $\mathcal{F}$ be a \texttt{Poisson distribution}, we have $R(i,j)\sim \mathrm{Poisson}(R_{0}(i,j))$. Sure, $\mathbb{E}(R(i,j))=R_{0}(i,j)$ holds. For this case, we have the below results.
\begin{itemize}
  \item $R(i,j)$ is a nonnegative integer, $B(i,j)]\in[0,1]$, and $\rho\in(0,+\infty)$ since the mean of Poisson distribution can be set as any positive value. Similar to the Binomial distribution, $R(i,j)$ can represent counts or frequencies of a response. The sorting aspect here is the same: higher values indicate a greater or more frequent response.
  \item $\tau$'s upper bound is unknown since we never know the  exact upper bound of $R(i,j)$ when we generate $R$ from Poisson distribution under WGoM; $\gamma\leq1$ since $\gamma=\frac{\mathrm{max}_{i\in[N],j\in[J]}\mathrm{Var}(R(i,j))}{\rho}=\frac{\mathrm{max}_{i\in[N],j\in[J]}R_{0}(i,j)}{\rho}\leq1$.
  \item Let $\gamma=1$, Assumption \ref{a1} is $\rho\geq\frac{\tau^{2}\mathrm{log}(N+j)}{\mathrm{max}(N,J)}$; Error bounds in Theorem \ref{mainWGoM} are the same as Instance \ref{Bernoulli}.
\end{itemize}
\end{Ex}
\begin{Ex}\label{Exponential}
Let $\mathcal{F}$ be a \texttt{Exponential distribution} such that $R(i,j)\sim \mathrm{Exponential}(\frac{1}{R_{0}(i,j)})$. Sure, $\mathbb{E}(R(i,j))=R_{0}(i,j)$ holds. For this case, we have the following conclusions.
\begin{itemize}
\item $R(i,j)\in(0,+\infty)$, $B(i,j)\in(0,1]$, and $\rho\in(0,+\infty)$. Similar to the Uniform distribution, though $R$'s elements are continuous, they can represent a sorting aspect where higher values correspond to stronger or more intense or more positive responses.
\item $\tau$ is unknown and $\gamma\leq\rho$ since $\gamma=\frac{\mathrm{max}_{i\in[N],j\in[J]}\mathrm{Var}(R(i,j))}{\rho}=\frac{\mathrm{max}_{i\in[N],j\in[J]}R^{2}_{0}(i,j)}{\rho}\leq\rho$.
\item Let $\gamma=\rho$, Assumption \ref{a1} is $\rho^{2}\geq\frac{\tau^{2}\mathrm{log}(N+J)}{\mathrm{max}(N, J)}$; Error bounds in Theorem \ref{mainWGoM} are the same as Instance \ref{Uniform}. Thus increasing $\rho$ has almost no influence on the error rates of the proposed method.
\end{itemize}
\end{Ex}
\begin{Ex}
Suppose that the elements of $R$ range in $\{-2,1,1.5\}$ (although real-world responses may not take values like $\{-2,~1,~1.5\}$, we still consider this case to emphasize the generality of our WGoM). We aim at constructing a discrete distribution $\mathcal{F}$ satisfying Equation (\ref{RFR0}), where $R\in\{-2,~1,~1.5\}^{N\times J}$ can be generated from distribution $\mathcal{F}$ under our WGoM. Next, we discuss how to construct $\mathcal{F}$. Let $\mathbb{P}(R(i,j)=-2)=p_{1}$, $\mathbb{P}(R(i,j)=1)=p_{2}$, and $\mathbb{P}(R(i,j)=1.5)=p_{3}$, where $p_{q}\geq0$ for $q\in[3]$. Because the summation of all probabilities should be 1, we have
\begin{align}\label{DisP}
p_{1}+p_{2}+p_{3}=1.
\end{align}
To make Equation (\ref{RFR0}) hold, we need
\begin{align}\label{DisR0}
\mathbb{E}(R(i,j))=-2p_{1}+p_{2}+1.5p_{3}\equiv R_{0}(i,j).
\end{align}
Combine Equation (\ref{DisP}) with Equation (\ref{DisR0}), by simple calculation, we have
\begin{align}\label{p2p3}
p_{2}=3-7p_{1}-2R_{0}(i,j),~p_{3}=6p_{1}+2R_{0}(i,j)-2.
\end{align}
For simplicity, assume that $p_{1}=R_{0}(i,j)$. Then Equation (\ref{p2p3}) gives
\begin{align}\label{p1p2p3}
p_{1}=R_{0}(i,j),~p_{2}=3-9R_{0}(i,j),~p_{3}=8R_{0}(i,j)-2.
\end{align}
Because $p_{1}\in[0,1]$, $p_{2}\in[0,1]$, and $p_{3}\in[0,1]$, $R_{0}(i,j)$ should satisfy $R_{0}(i,j)\in[\frac{1}{4},\frac{1}{3}]$. As a result, the discrete distribution $\mathcal{F}$ should be specified as
\begin{align}\label{DisF}
&\mathbb{P}(R(i,j)=-2)=R_{0}(i,j),\notag\\ &\mathbb{P}(R(i,j)=1)=3-9R_{0}(i,j),\notag\\ &\mathbb{P}(R(i,j)=1.5)=8R_{0}(i,j)-2,
\end{align}
where $R_{0}(i,j)$'s range is $[\frac{1}{4},\frac{1}{3}]$. When $\mathcal{F}$ is set as Equation (\ref{DisF}), we can generate a $R\in\{-2,~1,~1.5\}^{N\times J}$ from our model WGoM. Equation (\ref{DisF}) also gives that $\mathrm{Var}(R_{0}(i,j))=-R^{2}_{0}(i,j)+13R_{0}(i,j)-\frac{3}{2}$. Note that for this case, since $R_{0}(i,j)\in[\frac{1}{4},\frac{1}{3}]$ and $R_{0}=\Pi\Theta'$, we have $\Theta(j,k)\in[\frac{1}{4},\frac{1}{3}]$ for $j\in[J], k\in[K]$. Therefore, it's meaningless to set $\Theta=\rho B$ where $\mathrm{max}_{j\in[J],k\in[K]}|B(j,k)|=1$. Instead, we should use $\Theta$ to replace $\rho B$ in the proof of Theorem \ref{mainWGoM}, which gives that Assumption \ref{a1} is a lower bound requirement of $N$ (and $J$) and there is no need to consider $\rho$ and $\gamma$ anymore. Furthermore, we set $p_{1}, p_{2}, p_{3}$ in Equation (\ref{p1p2p3}) as a solution of Equations (\ref{DisP}) and (\ref{DisR0}), actually, there exist other solutions satisfying Equations (\ref{DisP}) and (\ref{DisR0}) by Theorem \ref{ExistDisF}. For example, we can also set $\mathcal{F}$ as
\begin{align*}
&\mathbb{P}(R(i,j)=-2)=-R_{0}(i,j), \\ &\mathbb{P}(R(i,j)=1)=3+5R_{0}(i,j), \\ &\mathbb{P}(R(i,j)=1.5)=-4R_{0}(i,j)-2,
\end{align*}
where $R_{0}(i,j)\in[-\frac{3}{5},-\frac{1}{2}]$. Therefore, there are many possible solutions to make Equations (\ref{DisP}) and (\ref{DisR0}) hold.
\end{Ex}
\begin{Ex}
Suppose that $R$'s entries range in $\{-2,1\}$. By Theorem \ref{ExistDisF}, $\mathcal{F}$ should be defined as
\begin{align*}
\mathbb{P}(R(i,j)=-2)=\frac{1-R_{0}(i,j)}{3},~\mathbb{P}(R(i,j)=1)=\frac{2+R_{0}(i,j)}{3},
\end{align*}
where $R_{0}(i,j)\in[-2,1]$. We also have $\mathrm{Var}(R(i,j))=-R^{2}_{0}(i,j)-R_{0}(i,j)+2$, which ranges in $[0,\frac{9}{4}]$. Since $R_{0}=\Pi\Theta'$, we see that $\Theta(j,k)\in[-2,1]$ for $j\in[J]$, $k\in[K]$. Because $\Theta\in[-2,1]^{J\times K}$, if we set $\Theta=\rho B$ such that $\mathrm{max}_{j\in[J],k\in[K]}|B(j,k)|=1$, we see that $\rho$'s range is $(0,2]$ when all entries of $B$ are negative and $\rho$'s range is $(0,1]$ when $B$'s entries are nonnegative. We also have $\tau\leq4$ and $\gamma=\frac{9}{4\rho}$. Setting $\gamma=\frac{9}{4\rho}$ in Theorem \ref{mainWGoM}, we see that increasing $\rho$ improves SCGoMA's performance.
\end{Ex}
\end{appendices}


\bibliography{refWGoM}% common bib file

\end{document}